\batchmode
\documentclass[12pt,english]{article}
\usepackage[]{graphicx}\usepackage[]{color}
\makeatletter
\def\maxwidth{ %
  \ifdim\Gin@nat@width>\linewidth
    \linewidth
  \else
    \Gin@nat@width
  \fi
}
\makeatother

\definecolor{fgcolor}{rgb}{0.345, 0.345, 0.345}
\newcommand{\hlkwb}[1]{\textcolor[rgb]{0.69,0.353,0.396}{#1}}%

\usepackage{mathtools}

\usepackage{framed}				
\makeatletter
 {\par\unskip\endMakeFramed%
 \at@end@of@kframe}
\makeatother

\definecolor{shadecolor}{rgb}{.97, .97, .97}
\definecolor{messagecolor}{rgb}{0, 0, 0}
\definecolor{warningcolor}{rgb}{1, 0, 1}
\definecolor{errorcolor}{rgb}{1, 0, 0}
\usepackage{float}
\usepackage{filecontents}
\usepackage{alltt}
\usepackage{lmodern}

\usepackage[T1]{fontenc}
\usepackage[latin9]{inputenc}
\usepackage[letterpaper]{geometry}
\usepackage{color}
\usepackage{babel}
\usepackage{breakcites}
\usepackage{amsmath}
\usepackage{amsthm}
\usepackage{amssymb}
\allowdisplaybreaks
\usepackage{tabularx}
\usepackage{bm}

\usepackage[pagewise]{lineno}
\makeatletter
\def\makeLineNumberLeft{%
  \linenumberfont\llap{\hb@xt@\linenumberwidth{\LineNumber\hss}\hskip\linenumbersep}
  \hskip\columnwidth
  \rlap{\hskip\linenumbersep\hb@xt@\linenumberwidth{\hss\LineNumber}}\hss}
\leftlinenumbers
\makeatother
\linenumbers
\AtBeginDocument{%
  \setlength{\linenumbersep}{\dimexpr\oddsidemargin+0.5in-\linenumberwidth\relax}%
}

\usepackage{pifont}
\newcommand{\xmark}{\text{\ding{55}}}
\newcommand{\cmark}{\text{\ding{51}}}
\usepackage{graphicx}
\usepackage{array}
\usepackage{tikz}
\usepackage[title]{appendix}
\newcolumntype{C}[1]{>{\centering\arraybackslash}m{#1}}
\usepackage{setspace}
\AtBeginDocument{%
  \addtolength\abovedisplayskip{-0.2\baselineskip}%
  \addtolength\belowdisplayskip{-0.2\baselineskip}%
}
 \usepackage{xr-hyper}

\usepackage[unicode=true,pdfusetitle,
 bookmarks=true,bookmarksnumbered=false,bookmarksopen=false,
 breaklinks=true,pdfborder={0 0 1},backref=false,colorlinks=true]
 {hyperref}
 \hypersetup{urlcolor=blue}
\hypersetup{citecolor=blue}
\hypersetup{filecolor=red}
\usepackage{multirow}
\makeatletter
\theoremstyle{plain}
\newtheorem{assumption}{\protect\assumptionname}[section]
\theoremstyle{plain}
\newtheorem{thm}{\protect\theoremname}[section]
  \theoremstyle{remark}
  \newtheorem{rem}{\protect\remarkname}[section]
  \theoremstyle{plain}
  \newtheorem{lem}{\protect\lemmaname}[section]
  \theoremstyle{plain}
  \newtheorem{cor}{\protect\corollaryname}[section]
    \theoremstyle{plain}
    \newtheorem{example}{\protect\examplename}[section]
     \theoremstyle{plain}
    \newtheorem{prop}{\protect\propositionname}[section]
    \theoremstyle{definition}
\newtheorem{definition}{Definition}[section]
\newtheorem{innercustomgeneric}{\customgenericname}
\providecommand{\customgenericname}{}
\newcommand{\newcustomtheorem}[2]{%
  \newenvironment{#1}[1]
  {%
   \renewcommand\customgenericname{#2}%
   \renewcommand\theinnercustomgeneric{##1}%
   \innercustomgeneric
  }
  {\endinnercustomgeneric}
}
\numberwithin{equation}{section}
\newcustomtheorem{customthm}{Theorem}
\newcustomtheorem{customlem}{Lemma}
\newcustomtheorem{customassumption}{Assumption}
\newcustomtheorem{customprop}{Proposition}
\newcustomtheorem{customexample}{Example}
\newcustomtheorem{customdef}{Definition}
\newcustomtheorem{customcor}{Corollary}
\newcustomtheorem{customrem}{Remark}

\@ifundefined{date}{}{\date{}}
\newcommand\independent{\protect\mathpalette{\protect\independenT}{\perp}}
\def\independenT#1#2{\mathrel{\rlap{$#1#2$}\mkern2mu{#1#2}}}

\usepackage{bbm}
\usepackage{chngcntr}

\usepackage{indentfirst}

\makeatother

  \providecommand{\lemmaname}{Lemma}
  \providecommand{\remarkname}{Remark}
\providecommand{\corollaryname}{Corollary}
\providecommand{\theoremname}{Theorem}
\providecommand{\examplename}{Example}
\providecommand{\assumptionname}{Assumption}
\providecommand{\propositionname}{Proposition}

\IfFileExists{upquote.sty}{\usepackage{upquote}}{}

\usepackage{natbib}
\usepackage{caption}
\makeatletter

\newcommand*{\addFileDependency}[1]{
  \typeout{(#1)}
  \@addtofilelist{#1}
  \IfFileExists{#1}{}{\typeout{No file #1.}}
}
\makeatother
 
\usepackage{pdfpages}
\begin{document}
\nolinenumbers
\title{{\vspace{-2cm}}Matching Points: Supplementing Instruments with Covariates in Triangular Models  \thanks{This paper is based on the first chapter of my doctoral dissertation at Columbia. I thank Jushan Bai, Sokbae (Simon) Lee and Bernard Salani\'e, who were gracious with their advice, support and feedback. I have also greatly benefited from comments and discussions with Karun Adusumili, Isaiah Andrews, Andres Aradillas-Lopez, Sandra Black, Ivan Canay, Songnian Chen, Xiaohong Chen, Eun Yi Chung, Leonard Goff, Florian Gunsilius, Han Hong, Roger Klein, Jessie Li, Xun Lu, Elena Manresa, Jos\'e Luis Montiel Olea, Ulrich M\"{u}ller, Whitney Newey, Serena Ng, Junhui Qian, Geert Ridder, Christoph Rothe, Zhentao Shi, Suyong Song, J\"{o}rg Stoye, Matt Taddy, Alexander Torgovitsky, Quang Vuong, Yulong Wang, Kaspar Wuthrich and the participants of the 2019 Econometric Society Asian Meeting in Xiamen and the econometrics seminars at Columbia, SJTU (Antai), CUHK Shenzhen, CUHK, HKUST, U of Iowa, Rutgers, USC, and UIUC. I also thank Research Connections for providing the data of the Head Start Impact Study. All errors are my own.}
}
\author{Junlong Feng\thanks{Department of Economics, the Hong Kong University of Science and Technology; \href{mailto:jlfeng@ust.hk}{jlfeng@ust.hk}.}}
\date{\today}
\maketitle
\vspace{-0.45in}
\begin{abstract}
Models with a discrete endogenous variable are typically underidentified when the instrument takes on too few values. This paper presents a new method that matches pairs of covariates and instruments to restore point identification in this scenario in a triangular model. The model consists of a structural function for a continuous outcome and a selection model for the discrete endogenous variable. The structural outcome function must be continuous and monotonic in a scalar disturbance, but it can be nonseparable. The selection model allows for unrestricted heterogeneity. Global identification is obtained under weak conditions. The paper also provides estimators of the structural outcome function. Two empirical examples of the return to education and selection into Head Start illustrate the value and limitations of the method. 
\vspace{0in}
\\
\vspace{0in}\\
\end{abstract}
\setcounter{page}{0}
\thispagestyle{empty}

\section{Introduction}
This paper considers identification and estimation of the structural outcome function $\boldsymbol{g}^{*}\equiv (g^{*}_{d})_{d}$ in a triangular model:
\begin{linenomath*}\begin{align}
Y&=\sum_{d\in S(D)}\mathbbm{1}(D=d)\cdot g^{*}_{d}(\boldsymbol{X},U_{d})\label{eq1.1}\\
D&=h(\boldsymbol{X},Z,\boldsymbol{V})\label{eq1.2}
\end{align}\end{linenomath*}
where both the endogenous variable $D\in S(D)$ and the instrumental variable $Z\in S(Z)$ are discrete, $\boldsymbol{X}$ is a vector of covariates, and the disturbances $(U_{d})_{d}$ and $\boldsymbol{V}$ are correlated (see different versions of the model in \cite{newey1999nonparametric}, \cite{chesher2003identification}, \cite{imbens2009identification}, etc.).

In many applications, the instrument takes on fewer values than the endogenous variable, i.e., the cardinality of their support sets satisfy $|S(Z)|<|S(D)|$. The outcome function $\boldsymbol{g}^{*}$ is then in general underidentified. For instance, under exogeneity of $Z$ and certain shape restrictions or separability of $\bm{g}^{*}(\boldsymbol{X},\cdot)$, the $|S(D)|$-vector of the unknown parameters $\bm{g}^{*}(\boldsymbol{x}_{0},u)$ for some $\bm{X}=\bm{x}_{0}$ and $u$ may satisfy moment equations conditional on $Z$ and $\bm{X}=\bm{x}_{0}$ (e.g. \cite{newey2003instrumental} and \cite{chernozhukov2005iv}). Conditioning on each value of $Z$ generates one moment equation. The number of the equations is thus $|S(Z)|$, smaller than the number of the unknowns ($|S(D)|$). The classical order condition fails and so does identification of $\bm{g}^{*}(\boldsymbol{x}_{0},u)$.

This paper develops a novel approach to obtain point identification of $\boldsymbol{g}^{*}$ in this scenario. For a value of interest $\bm{X}=\bm{x}_{0}$, identification is achieved by finding special values $\bm{X}=\bm{x}_{m}$, the \textit{matching points}, such that the function $\boldsymbol{g}^{*}(\boldsymbol{x}_{m},\cdot)$ can be expressed as a known mapping of $\boldsymbol{g}^{*}(\boldsymbol{x}_{0},\cdot)$. As a consequence, both variation in $Z$ and local variation in $\bm{X}$ across $\bm{x}_{0}$ and $\bm{x}_{m}$ have identification power for $\boldsymbol{g}^{*}(\boldsymbol{x}_{0},\cdot)$. Specifically, by substituting the mapping into the moment equations conditional on $Z$ and $\bm{X}=\bm{x}_{m}$ for a chosen $u$, the unknowns in these equations become $\boldsymbol{g}^{*}(\boldsymbol{x}_{0},u)$. Together with the moment equations conditional on $Z$ and $\bm{X}=\bm{x}_{0}$, the total number of the equations for $\boldsymbol{g}^{*}(\boldsymbol{x}_{0},u)$ is increased, while the number of the unknowns is unchanged. Even though the covariates are not excluded from the structural function, the special value $\bm{x}_{m}$ facilitates identification in a way an additional instrument value does. The effective support set of the instrument is thus enlarged, making identification possible.

Two key restrictions are needed in addition to the exogeneity of the instrument for the mapping of $\boldsymbol{g}^{*}(\boldsymbol{x}_{0},\cdot)$ to $\boldsymbol{g}^{*}(\boldsymbol{x}_{m},\cdot)$ to be traced out before themselves are. First, the matching point $\bm{x}_{m}$ and the value of interest $\bm{x}_{0}$ need to generate the same selection patterns when paired with appropriate instrument values. That is, for some $z,z'\in S(Z)$, $\bm{x}_{m}$ needs to satisfy $h(\bm{x}_{m},z',\cdot)=h(\bm{x}_{0},z,\cdot)$. Second, each outcome disturbance $U_{d}$ is a scalar and $g^{*}_{d}(\bm{X},\cdot)$ is continuous and strictly increasing for all $d\in S(D)$ almost surely. 

A consequence of the first restriction is that the distributions of the outcome disturbance $U_{D}\equiv\sum_{d\in S(D)}U_{d}$ conditional on $D$ given $(\bm{X},Z)=(\bm{x}_{0},z)$ and $(\bm{x}_{m},z')$ are equal, under some other assumptions. Intuitively, $D$ has the same degree of endogeneity under $(\bm{X},Z)=(\bm{x}_{m},z')$ and $(\bm{x}_{0},z)$. Continuity and monotonicity of $g^{*}_{d}(\bm{X},\cdot)$ imposed by the second key restriction then transform equality between the conditional distributions of $U_{D}$ into equality between the conditional distributions of $Y$ evaluated at $g_{d}^{*}(\bm{x}_{0},\cdot)$ and $g^{*}_{d}(\bm{x}_{m},\cdot)$ for each $d$. The mapping from $\bm{g}^{*}(\bm{x}_{0},\cdot)$ to $\bm{g}^{*}(\bm{x}_{m},\cdot)$ can then be traced out by inverting these observable distributions. 

Two approaches are available to find the matching points $\bm{x}_{m}$ satisfying the first restriction. If $h$ is known and identified, one can obtain the matching points by searching for an $\bm{x}_{m}$ that matches $h$ with $\bm{x}_{0},z,z'$ fixed. If $h$ is unknown, I show that a statistical implication of the first key restriction is that the generalized propensity scores conditional on $(\bm{X},Z)=(\bm{x}_{0},z)$ and $(\bm{x}_{m},z')$ are equal. Therefore, one can recover the matching points robustly by searching for an $\bm{x}_{m}$ to match the generalized propensity scores without knowing $h$. A sufficient condition for all such $\bm{x}_{m}$ to be matching points is that $(\bm{X},Z)$ enters $h$ only via the generalized propensity scores. Usually, this requires the selection model to have an index structure such as discrete choice models (see \cite{heckman2005structural}, for example). For selection models not satisfying these conditions, one can still use the solutions to the generalized propensity score matching as candidates of matching points, and their validity is testable. Therefore, the underlying selection model can be very general with no restrictions on the dimensionality or separability in the selection heterogeneity $\bm{V}$.

After the order condition is fulfilled using the instrument and the matching points, continuity and monotonicity of $\bm{g}^{*}(\bm{X},\cdot)$ also simplify the sufficient conditions for global identification. I show that the outcome function $\boldsymbol{g}^{*}(\bm{x}_{0},\cdot)$ is globally identified among monotonic functions if $\boldsymbol{g}^{*}(\boldsymbol{x}_{0},u)$ is only locally identified for all $u$. Hence, this new result only relies on local invertibility conditions for nonlinear functions, which are much weaker than the conditions in global inverse theorems widely adopted for global identification. This result also applies to the standard nonparametric quantile IV approach when the instrument has large support, and may be of independent interest. Based on the identification strategy, I construct a sieve estimator and derive its asymptotic properties under simple low level conditions.

An important special case of a continuous and strictly increasing structural outcome function is that it is additively separable in the disturbance. I show that under separability, the outcome function at a given $\bm{X}=\bm{x}_{0}$ solves a system of linear equations, preserving a similar structure as in the nonparametric IV approach with rich instruments (e.g. \cite{newey2003instrumental} and \cite{das2005instrumental}). I construct a closed form estimator that is easy to implement in practice. I apply it to a return to education application using the same extract from the 1979 National Longitudinal Surveys (NLS) as in \cite{card1995using}. Adopting the binary proximity-to-college instrument, the returns of three levels of education, high school, some college, and college and above, are nonparametrically underidentified by the existing approaches. To apply my approach, I use the average of parents' years of schooling as a covariate to generate matching points. Identification is restored using the matching points. I then estimate the returns and find that they are increasing in the level of education and heterogeneous in parents' years of schooling.

It is worth noting that my approach hinges on the covariates' ability to offset the effect of the instrument on selection. For applications where the instrument has a dominant effect, matching points may not exist. As an illustration, I consider another empirical example of the choice of preschool programs. I use the Head Start Impact Study dataset following \cite{kline2016evaluating}. A randomly assigned lottery granting access to Head Start serves as the instrument, while the endogenous variable is the multivalued preschool program choice. I find that the instrument has a much larger effect on selection than covariates such as the baseline test scores and family income. No matching points exist. 

I defer a detailed discussion of the relation of my approach to the literature until Section \ref{1.sec8}. Here let me only briefly highlight some major differences. Methods that circumvent the problem of a small-support instrument include imposing homogeneity between adjacent levels of $D$ when $D$ is ordered, or specifying a parametric form for $\boldsymbol{g}^{*}$ and using interactions between $Z$ and $\boldsymbol{X}$ as extra instruments by assuming $\boldsymbol{X}$ is exogenous. These methods would fail in a fully nonparametric model, as studied in this paper. \cite{torgovitsky2015identification,torgovitsky2017minimum} and \cite{d2015identification} show that a binary instrument can identify nonseparable models with a continuous $D$. Different from my approach, continuity in $D$ is indispensable, and they require $\bm{V}$ to be a scalar and the selection function $h$ to be strictly increasing in it. \cite{caetano2016identifying} use covariates to identify models when the instruments do not have enough variation. Their approach does not rely on a selection model, but they need the covariates used for identification purposes separable from the model. In contrast, the covariates in my approach can enter the model in an arbitrary way. The idea in my approach of using shifts in some observables to compensate for a shift in a target variable to facilitate identification can also be seen in \cite{ichimura2000direct}, \cite{vytlacil2007dummy} and \cite{chen2020identification}. They focus on different parameters than this paper, and the shifting variables and the target variables are also different.
\cite{vuong2017counterfactual} and \cite{feng2019estimation} study the individual treatment effect of a binary $D$ and develop a concept called the counterfactual mapping. It is also an identifiable mapping linking two outcome functions, but at \textit{different} values of $D$ and the \textit{same} value of $\bm{X}$.

The rest of the paper is organized as follows. In Section \ref{1.sec2}, I introduce the matching points and show how to use them to generate new moment equations and restore the order condition. How to find the matching points is also discussed. Given the fulfilled order condition, Section \ref{1.sec3} provides sufficient conditions for the global identification of the nonseparable model. Section \ref{1.secSP} discusses identification of a separable model as a special case under relaxed conditions. Section \ref{1.sec4} sketches estimation of the matching points and the separable model. Section \ref{1.sec6} shows Monte Carlo simulation results to illustrate the estimator's finite sample performance. Section \ref{1.sec7} presents two empirical applications.  Section \ref{1.sec8} discusses the relation of my approach to the literature. Section \ref{1.sec9} concludes. Some additional results and the proofs of identification are in Appendix. In the Supplemental Material, I provide an estimator of the general nonseparable model, proofs of its asymptotic properties, and additional simulation results.

\subsubsection*{Notation} 
A vector valued function is said to be (strictly) increasing or monotonic if every component in it is (strictly) increasing or monotonic. Random variables are denoted by upper-case Latin letters and  their realizations by the corresponding lower-cases. Bold Latin letters denote vectors or matrices. For two generic random variables $A$ and $B$, denote the conditional expectation of $A$ given $B=b$ by $\mathbb{E}_{A|B}(b)$, with similar notation for conditional cumulative distribution functions $F_{A|B}(a|b)$, densities $f_{A|B}(a|b)$ and variances $\mathbb{V}_{A|B}(b)$. Denote the support set of $A$ by $S(A)$, and the support of $A$ given $B=b$ by $S(A|B=b)$, or simply $S(A|b)$ when it does not cause confusion. For a finite set $G$, $|G|$ denotes its cardinality.
Throughout, I assume all the random variables are in a common probability space with the measure function $\mathbb{P}$. Almost surely (a.s.) and measurable are always with respect to $\mathbb{P}$.
\section{Matching Points and the Order Condition}\label{1.sec2}
To highlight the key features of my approach, I focus on a simple case where a single endogenous $D$ takes on three values ($|S(D)|=3$) and the instrument $Z$ is binary $(|S(Z)|=2)$. The discrete $D$ can be either ordered or unordered. Without loss of generality, let $S(D)=\{1,2,3\}$ and $S(Z)=\{0,1\}$. The general cases of arbitrary $|S(D)|>|S(Z)|$ and of multiple endogenous variables will be discussed in Appendix \ref{1.appA}. Besides the endogenous variable and the instrument, a continuous outcome variable $Y$ and a vector of covariates $\bm{X}$ are observable. The outcome is determined by the structural equation \eqref{eq1.1}. The function $\bm{g}^{*}$ will be referred to as the outcome function subsequently.

For a fixed value $\bm{x}_{0}\in S(\bm{X})$, $\bm{g}^{*}(\bm{x}_{0},\cdot)$ is underidentified by the standard IV approach (for example \cite{chernozhukov2005iv}). This is because for any $u$, there can be only two moment equations for $\bm{g}^{*}(\bm{x}_{0},u)$ by conditioning on $(\bm{X},Z)=(\bm{x}_{0},0)$ and $(\bm{x}_{0},1)$. But $\bm{g}^{*}(\bm{x}_{0},u)$ contains three elements. The classical order condition thus fails. The idea of this paper is to generate more moment equations by conditioning on some other special values of the covariates, the matching points, denoted by $\bm{x}_{m}$. Consider the moment equations conditional on $(\bm{X},Z)=(\bm{x}_{m},0)$ and $(\bm{x}_{m},1)$. Although the unknowns in these equations are $\bm{g}^{*}(\bm{x}_{m},u)$, if $\bm{g}^{*}(\bm{x}_{m},u)$ can be rewritten as an known function of $\bm{g}^{*}(\bm{x}_{0},u)$, then substituting the function into these moment equations provide new equations for $\bm{g}^{*}(\bm{x}_{0},u)$.  

To achieve this goal, a selection model for $D$ is needed. By the discreteness of $D$, rewrite the selection model \eqref{eq1.2} as follows:
\begin{linenomath*}\begin{equation}\label{1.sl}
D=d\text{ if and only if } h_{d}(\boldsymbol{X},Z,\boldsymbol{V})=1
\end{equation}\end{linenomath*}
where for all $d\in S(D)$, the function $h_{d}(\boldsymbol{X},Z,\boldsymbol{V})\in\{0,1\}$ and $\sum_{d\in S(D)} h_{d}(\boldsymbol{X},Z,\boldsymbol{V})=1$ a.s. The random vector $\boldsymbol{V}$ contains unobservables that are correlated with the outcome disturbances $\{U_{d}\}$. I assume that for every $(\boldsymbol{x},z)\in S(\boldsymbol{X},Z)$, the indicator function $h_{d}(\boldsymbol{x},z,\cdot)$ is measurable on $S(\boldsymbol{V})$. Note that the covariates in the selection model are the same as those in the outcome function for simplicity. Additional covariates in the outcome function are indeed allowed. In that case, all the analysis in this paper can be viewed as conditional on them. In the rest of this section, I introduce the matching points of an arbitrary $\bm{x}_{0}\in S(\bm{X})$ and related concepts, and  show how the mappings between the outcome functions at these points are identified, and how to use them to generate moment equations.


\subsection{Matching Points and the M-Connected Set}
A matching point $\bm{x}_{m}$ of a given point $\bm{x}_{0}\in S(\bm{X})$ is defined as follows:
\begin{definition}[Matching Point and Matching Pair]\label{1.defMP} A point $\boldsymbol{x}_{m}\in S(\boldsymbol{X})$ is a matching point of $\boldsymbol{x}_{0}\in S(\boldsymbol{X})$ if there exist $z, z'\in S(Z)$ such that the following equations hold for all $d\in S(D)$, $u\in S(U)$ and $\bm{v}\in S(\bm{V})$:
\begin{linenomath*}\begin{align}
h_{d}(\boldsymbol{x}_{m},z',\bm{v})&=h_{d}(\boldsymbol{x}_{0},z,\bm{v}),\label{1.eq3}\\
F_{U_{d}\boldsymbol{V}|\bm{X}Z}(u,\bm{v}|\boldsymbol{x}_{m},z')&=F_{U_{d}\boldsymbol{V}|\bm{X}Z}(u,\bm{v}|\boldsymbol{x}_{0},z)\label{1.eq2.5}
\end{align}\end{linenomath*}
The covariates-instrument combinations $(\boldsymbol{x}_{0},z)$  and $(\boldsymbol{x}_{m},z')$ form a matching pair.
\end{definition}

Note that equation \eqref{1.eq2.5} is satisfied for any $(\bm{x}_{0},z)$ and $(\bm{x}_{m},z')$ if $(\bm{X},Z)$ are jointly independent of $(U_{d},\bm{V})$ for all $d\in S(D)$. The independence of $Z$ is needed in this paper and will be introduced in the following subsections. The independence of the covariates is unnecessary, though commonly assumed in practice. It is worth noting that the conditions only need to be satisfied by the covariates used to generate the matching points. Finally, all these restrictions are indirectly testable under overidentification of $\bm{g}^{*}(\bm{x}_{0},\cdot)$. 

Denote the generalized propensity score $\mathbb{P}(D=d|\boldsymbol{X}=\boldsymbol{x},Z=z)$ at an arbitrary point $(\bm{x},z)$ by $p_{d}(\boldsymbol{x},z)$. The following lemma is a direct consequence of equations \eqref{1.eq3} and \eqref{1.eq2.5}.

\begin{lem}\label{1.thmMEQ}
Suppose $(\bm{x}_{0},z)$ and $(\bm{x}_{m},z')$ are a matching pair. For all $d\in S(D)$ and $u\in S(U)$, 
\begin{linenomath*}\begin{align}
F_{U_{d}|D\boldsymbol{X}Z}(u|d,\boldsymbol{x}_{m},z')&=F_{U_{d}|D\boldsymbol{X}Z}(u|d,\boldsymbol{x}_{0},z)\label{1.eq5}\\
p_{d}(\boldsymbol{x}_{m},z')&=p_{d}(\boldsymbol{x}_{0},z)\label{eq2.4a}
\end{align}\end{linenomath*}
\end{lem}
\begin{proof}
See Appendix \ref{1.appB}.
\end{proof}
Lemma \ref{1.thmMEQ} does not require the instrument or the covariates to be exogenous. Also, no structures for the outcome functions are imposed so far. So the lemma may have a broader application than studied in this paper.  In Section \ref{sec2.1}, I will show under what conditions the mapping from $\bm{g}^{*}(\bm{x}_{0},\cdot)$ to $\bm{g}^{*}(\bm{x}_{m},\cdot)$ can be traced out from equation \eqref{1.eq5}, while equation \eqref{eq2.4a} provides a model-free way to find the matching points, discussed in Section \ref{1.sec3.1}. 

For each $\bm{x}_{0}$, one may find different matching points by pairing $\bm{x}_{0}$ with different values of $Z$; for some $\bm{x}_{m1}\neq \bm{x}_{m2}$, points $(\bm{x}_{0},0)$ and $(\bm{x}_{m1},1)$ and $(\bm{x}_{0},1)$ and $(\bm{x}_{m2},0)$ may form two different matching pairs. Furthermore, a matching point of $\bm{x}_{0}$ may have its own matching points besides $\bm{x}_{0}$. See the following example for an illustration.


\begin{example}[Ordered Choice]\label{1.exOC3}
Suppose $D$ is ordered and there is only one covariate $X$. Let $h_{1}(X,Z,V)=\mathbbm{1}(V<\kappa_{1}+\beta X+\alpha Z)$, $h_{3}(X,Z,V)=\mathbbm{1}(V\geq \kappa_{2}+\beta X+\alpha Z)$, and $h_{2}=1-h_{1}-h_{3}$. Assume $\alpha\cdot\beta \neq 0$, $\kappa_{1}<\kappa_{2}$, and $(X,Z)\independent (U_{d},V)$ for all $d$ where $V$ is continuously distributed on $\mathbb{R}$. For any fixed $x_{0}\in S(X)$, it has the following two matching points by equation \eqref{1.eq3} if both are in $S(X)$:
\begin{linenomath*}\begin{align}
(z=0,z'=1):\beta x_{m1}+\alpha\cdot 1=\beta x_{0}+\alpha\cdot 0\implies x_{m1}=x_{0}-\frac{\alpha}{\beta}\label{1.eqOC1}\\
(z=1,z'=0):\beta x_{m2}+\alpha\cdot 0=\beta x_{0}+\alpha\cdot 1\implies x_{m2}=x_{0}+\frac{\alpha}{\beta}\label{1.eqOC2}
\end{align}\end{linenomath*}
Similarly, each of $x_{m1}$ and $x_{m2}$ also has two matching points: One is $x_{0}$, and the other is $x_{0}-2\frac{\alpha}{\beta}$ and $x_{0}+2\frac{\alpha}{\beta}$ respectively. This process can be continued until the boundaries of $S(X)$ are reached, illustrated in the following figure:
 \begin{figure}[H]
\centering
\includegraphics[width=1\linewidth,trim={2cm 9cm 0 5cm},clip]{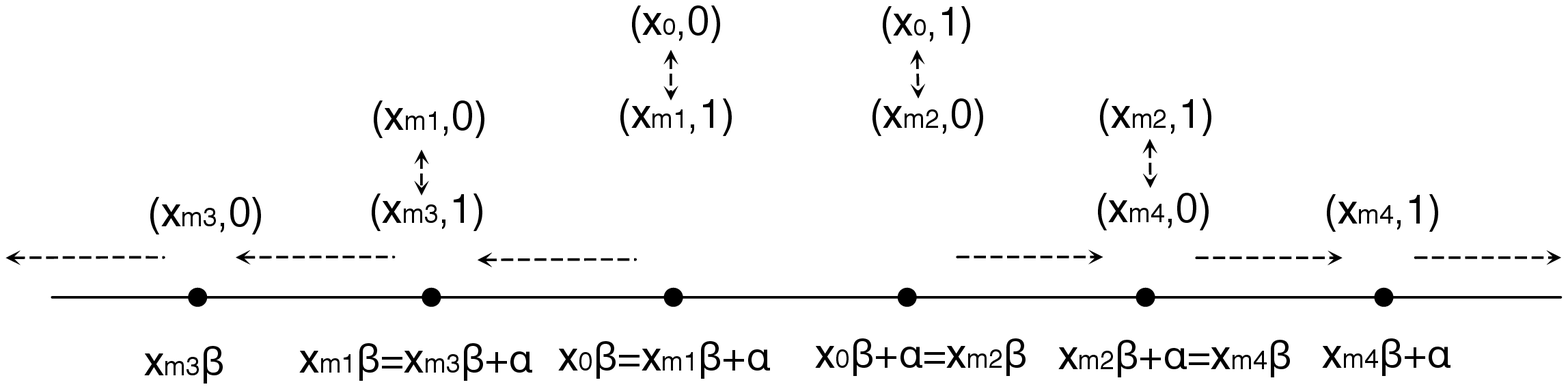}
\caption{The Pyramid of Matching Points}\label{1.fig1}
\end{figure}
\noindent The horizontal axis in Figure \ref{1.fig1} is the value of the single index $x\beta+z\alpha$. Starting from $(x_{0},0)$ and $(x_{0},1)$, $x_{m1}$ and $x_{m2}$ are obtained by solving the equations below the horizontal axis. Repeat this procedure to match $(x_{m1},0)$ with $(x_{m3},1)$ and $(x_{m2},1)$ with $(x_{m4},0)$. Continuing the process, one can expect to see that the dotted points on this axis extend to both directions, until they reach the boundaries of $S(X)$.
\end{example}

In Example \ref{1.exOC3}, $x_{0}$ has two matching points: $x_{m1}$ and $x_{m2}$. Each of them has one extra descendent matching point $x_{m3}$ and $x_{m4}$. These points, in turn, have their matching points. Although they are no longer matching points of $\boldsymbol{x}_{0}$, the conditional distributions of $U_{d}$ at these descendants and at $\boldsymbol{x}_{0}$ paired with appropriate values of the instrument are still linked by repeatedly applying Lemma \ref{1.thmMEQ}. Motivated by this observation, let me introduce a concept that is more general than the matching point.

\begin{definition}[M-Connected Set]\label{1.defMPMC} 
A set $\mathcal{X}_{MC}(\boldsymbol{x}_{0})\subseteq S(\bm{X})$ is called the m-connected set of $\boldsymbol{x}_{0}$ if $\boldsymbol{x}_{0}\in \mathcal{X}_{MC}(\boldsymbol{x}_{0})$ and for any $\boldsymbol{x}\in \mathcal{X}_{MC}(\boldsymbol{x}_{0})$, there exists $\boldsymbol{x}_{1},\boldsymbol{x}_{2},...,\boldsymbol{x}_{k(\boldsymbol{x})}\in  \mathcal{X}_{MC}(\boldsymbol{x}_{0})$ such that $\boldsymbol{x}_{j}$ is a matching point of $\boldsymbol{x}_{j-1}$, $j=1,...,k(\boldsymbol{x})$, and $\boldsymbol{x}$ is a matching point of $\boldsymbol{x}_{k(\boldsymbol{x})}$. Any two points in the m-connected set are said to be \textit{m-connected}.
\end{definition}

In Example \ref{1.exOC3}, any two points in the form of $x_{0}+c\frac{\alpha}{\beta},c\in\mathbb{Z}$, are m-connected provided that both are in $S(X)$. The m-connected set of $x_{0}$ is $\mathcal{X}_{MC}(x_{0})=\{x_{0}+c\frac{\alpha}{\beta},c\in\mathbb{Z}\}\cap S(X)$.

By construction, the m-connected set is the largest subset of $S(\bm{X})$ such that the relationship between the conditional distributions $U_{d}$ at any two elements in it can be established by recursively applying Lemma \ref{1.thmMEQ}. 

\subsection{The Fulfillment of the Order Condition}\label{sec2.1}
Now it is ready to present how to use Lemma \ref{1.thmMEQ} to recover the mapping of $\bm{g}^{*}(\bm{x}_{0},\cdot)$ to $\bm{g}^{*}(\bm{x},\cdot)$ for any $\bm{x}\in\mathcal{X}_{MC}(\bm{x}_{0})$, and how in turn to use the m-connected points to supplement the instrument.
\subsubsection{The Mapping of $\bm{g}^{*}(\bm{x}_{0},\cdot)$ to $\bm{g}^{*}(\bm{x},\cdot)$}
Matching $\bm{g}^{*}(\bm{x},\cdot)$ with $\bm{g}^{*}(\bm{x}_{0},\cdot)$ by Lemma \ref{1.thmMEQ} needs the following assumptions.
\begin{assumption}[Continuity and Monotonicity]\label{1.assCM} For all $\boldsymbol{x}\in S(\boldsymbol{X})$, $\bm{g}^{*}(\boldsymbol{x},\cdot)$ is continuous and strictly increasing.
\end{assumption}

\begin{assumption}[Normalization, Exogeneity, and Full Support]\label{1.assENSP}
For all $\bm{x}\in\mathcal{X}_{MC}(\bm{x}_{0})$ and $d\in S(D)$, conditional on $\bm{X}=\bm{x}$ , i) $U_{d}\sim\textrm{Unif}[0,1]$, ii) $(U_{d},\boldsymbol{V})$ is continuously distributed with $(U_{d},\boldsymbol{V})\independent Z$, and iii) $S(U_{d}|\boldsymbol{V})=S(U_{d})$ .
\end{assumption}

Assumption \ref{1.assCM} regulates the behavior of $\boldsymbol{g}^{*}(\boldsymbol{x},\cdot)$. Continuity and strict monotonicity are two standard requirements in the literature of nonseparable models with a scalar unobservable (e.g. \cite{matzkin2003nonparametric,matzkin2007nonparametric}, \cite{chernozhukov2005iv}, etc.). In addition to constructing moment conditions, I will show that they also deliver nice results for the global uniqueness of the solution to systems of nonlinear equations in Section \ref{1.sec3}.

Assumption \ref{1.assENSP} i) normalizes the distribution of $U_{d}$ for each $d\in S(D)$ conditional on an arbitrary point $\bm{x}$ in $\mathcal{X}_{MC}(\bm{x}_{0})$. Together with Assumption \ref{1.assCM}, the normalization gives $g^{*}_{d}(\bm{x},\cdot)$ an interpretation of the counterfactual quantile function. 

In Assumption \ref{1.assENSP} ii), the requirement of joint independence between $Z$ and $(U_{d},\bm{V})$ is common for triangular models (e.g. \cite{imbens2009identification}), yet it is usually made only conditional on the value of interest $\bm{X}=\bm{x}_{0}$. The assumption here is stronger since joint independence needs to hold conditional on every point in $\mathcal{X}_{MC}(\bm{x}_{0})$. The need arises from using the moment conditions conditional on the m-connected points instead of only on $\bm{X}=\bm{x}_{0}$, as will be seen later in this section. It is noteworthy that when the moment conditions obtained from the m-connected points are more than needed, joint independence can be relaxed to hold only conditional on points in a subset of $\mathcal{X}_{MC}(\bm{x}_{0})$. 

In Assumption \ref{1.assENSP} iii), the full support requirement is useful to identify $\bm{g}^{*}(\bm{x}_{0},\cdot)$ on the entire domain $[0,1]$. It can also be found in related work with a similar focus, for instance \cite{d2015identification}, \cite{torgovitsky2015identification} and \cite{vuong2017counterfactual}.

Assumptions \ref{1.assCM} and \ref{1.assENSP} guarantee that $S(Y|d,\boldsymbol{x},z)=S(Y|d,\bm{x})$, that $S(Y|d,\bm{x})$ is compact, and that the range of $g^{*}_{d}(\boldsymbol{x},\cdot)$ on $[0,1]$ is equal to $S(Y|d,\boldsymbol{x})$ for all $d\in S(D)$, $z\in S(Z)$ and $\bm{x}\in\mathcal{X}_{MC}(\bm{x}_{0})$.  

Under Assumption \ref{1.assCM}, $F_{Y|D\bm{X}Z}\left(g^{*}_{d}(\bm{x},u)|d,\bm{x},z\right)=F_{U_{d}|D\bm{X}Z}\left(u|d,\bm{x},z\right)$ for any $(u,d,\bm{x},z)\in S(U_{D},D,\bm{X},Z)$. Equation \eqref{1.eq5} in Lemma \ref{1.thmMEQ} thus implies the following equation for a matching pair $(\bm{x}_{0},z)$ and $(\bm{x}_{m},z')$ for all $d\in S(D)$ and $u\in S(U_{d}|d,\bm{x}_{0},z)$:
\begin{linenomath*}\begin{equation}
F_{Y|D\bm{X}Z}\left(g^{*}_{d}(\bm{x}_{m},u)|d,\bm{x}_{m},z'\right)=F_{Y|D\bm{X}Z}\left(g^{*}_{d}(\bm{x}_{0},u)|d,\bm{x}_{0},z\right),\label{eq2.9}
\end{equation}\end{linenomath*}
Let $F^{-1}_{Y|DXZ}:[0,1]\mapsto S(Y|D,X,Z)$ be the inverse of the cumulative conditional distribution function of $Y$ whenever it is strictly increasing on the support. Let $q_{d,(\bm{x},z),(\bm{x'},z')}(y)=F^{-1}_{Y|D\boldsymbol{X}Z}\big(F_{Y|D\boldsymbol{X}Z}(y|d,\boldsymbol{x},z)\big|d,\boldsymbol{x}',z'\big)$ for $y\in \mathbb{R}$. Under exogeneity of $Z$ and the full support condition in Assumption \ref{1.assENSP}, $F_{Y|D\bm{X}Z}\left(g^{*}_{d}(\bm{x}_{m},u)|d,\bm{x}_{m},z'\right)$ is strictly increasing in $g^{*}_{d}(\bm{x}_{m},u)$ for all $u\in [0,1]$. Therefore, equation \eqref{eq2.9} implies
\begin{linenomath*}\begin{equation}\label{1.eq15}
g^{*}_{d}(\boldsymbol{x}_{m},u)=q_{d,(\bm{x}_{0},z),(\bm{x}_{m},z')}\left(g^{*}_{d}(\boldsymbol{x}_{0},u)\right)
\end{equation}\end{linenomath*}
for all $d\in S(D)$ and $u\in [0,1]$. The mapping $q_{d,(\bm{x}_{0},z),(\bm{x}_{m},z')}$ is directly identifiable from the population, continuous and increasing on $\mathbb{R}$ and strictly increasing on $S(Y|d,\bm{x}_{0})$. 

More generally, let $\varphi_{d}(g^{*}_{d}(\boldsymbol{x}_{0},u);\bm{x})$ denote the mapping from $g^{*}_{d}(\boldsymbol{x}_{0},u)$ to $g^{*}_{d}(\boldsymbol{x},u)$ for an arbitrary $\bm{x}\in\mathcal{X}_{MC}(\bm{x}_{0})$. This mapping is also identified. Specifically, suppose $\bm{x}$ is m-connected with $\bm{x}_{0}$ such that $(\bm{x}_{0},z_{0})$ and $(\bm{x}_{1},z_{1})$, $(\bm{x}_{1},z_{1}')$ and $(\bm{x}_{2},z_{2})$, ..., $(\bm{x}_{k},z_{k}')$ and $(\bm{x},z)$ are matching pairs, where $z_{0},z_{1},z_{1}',z_{2},...,z_{k}',z\in S(Z)$, then for all $d\in S(D)$ and $u\in [0,1]$,
\begin{linenomath*}\begin{equation}
\varphi_{d}(g^{*}_{d}(\boldsymbol{x}_{0},u);\bm{x})=q_{d,(\bm{x}_{k},z_{k}'),(\bm{x},z)}\circ\cdots\circ q_{d,(\bm{x}_{1},z_{1}'),(\bm{x}_{2},z_{2})}\circ q_{d,(\bm{x}_{0},z_{0}),(\bm{x}_{1},z_{1})}(g_{d}^{*}(\bm{x}_{0},u))
\end{equation}\end{linenomath*}
where $\circ$ denotes function composition. 

The mapping $\varphi_{d}(\cdot;\bm{x})$ has some useful properties. First, it is continuous and increasing on $\mathbb{R}$, and strictly increasing on $S(Y|d,\bm{x}_{0})$. Second, for a matching pair $(\bm{x}_{m},z')$ and $(\bm{x}_{0},z)$, $\varphi_{d}(g^{*}_{d}(\boldsymbol{x}_{0},u);\bm{x}_{m})=q_{d,(\bm{x}_{0},z),(\bm{x}_{m},z')}(g^{*}_{d}(\boldsymbol{x}_{0},u))$. Third, since $(\bm{x}_{0},z)$ and itself form a matching pair by definition, $\varphi_{d}(g^{*}_{d}(\bm{x}_{0},u);\bm{x}_{0})=q_{d,(\bm{x}_{0},z),(\bm{x}_{0},z)}(g^{*}_{d}(\boldsymbol{x}_{0},u))=g^{*}_{d}(\boldsymbol{x}_{0},u)$. 
\subsubsection{The Moment Conditions}
With the traced out mapping from $g^{*}_{d}(\bm{x}_{0},\cdot)$ to $g^{*}_{d}(\bm{x},\cdot)$ for each $d\in S(D)$ and $\bm{x}\in \mathcal{X}_{MC}(\bm{x}_{0})$, moment conditions can be constructed.

\begin{assumption}[Rank Similarity]\label{1.assRSS} Conditional on any $\bm{x}\in\mathcal{X}_{MC}(\bm{x}_{0})$, $\{U_{d}\}$ are identically distributed conditional on $\bm{V}$.
\end{assumption}

Assumption \ref{1.assRSS} is adopted from \cite{chernozhukov2005iv} to handle $d$-dependent outcome disturbances. It is worth emphasizing that Assumption \ref{1.assRSS} is not needed to identify the mapping $\varphi_{d}$.
\begin{prop}[Moment Condition]\label{1.prop2}  Under Assumptions \ref{1.assCM} to \ref{1.assRSS}, the following equation holds for all $z\in\{0,1\}$, $\bm{x}\in \mathcal{X}_{MC}(\bm{x}_{0})$ and $u\in[0,1]$,
\begin{linenomath*}\begin{equation}\label{1.eq2}
\sum_{d\in S(D)}p_{d}(\boldsymbol{x},z)\cdot F_{Y|D\boldsymbol{X}Z}(\varphi_{d}(g^{*}_{d}(\boldsymbol{x}_{0},u);\bm{x})|d,\boldsymbol{x},z)=u
\end{equation}\end{linenomath*}
\end{prop} 
\begin{proof}
See Appendix \ref{1.appB}.
\end{proof}

Proposition \ref{1.prop2} generalizes Theorem 1 in \cite{chernozhukov2005iv}. To identify $\bm{g}^{*}(\bm{x}_{0},u)$, they only condition on $\bm{X}=\bm{x}_{0}$. In equation \eqref{1.eq2}, this is the case when $\bm{x}=\bm{x}_{0}$ because $\varphi_{d}(g^{*}_{d}(\boldsymbol{x}_{0},u);\bm{x}_{0})=g^{*}_{d}(\boldsymbol{x}_{0},u)$, and the equation becomes
\begin{linenomath*}\begin{equation*}
\sum_{d\in S(D)}p_{d}(\boldsymbol{x}_{0},z)\cdot F_{Y|D\boldsymbol{X}Z}(g^{*}_{d}(\boldsymbol{x}_{0},u)|d,\boldsymbol{x}_{0},z)=u
\end{equation*}\end{linenomath*}
As the generalized propensity scores and the conditional cumulative distribution functions are directly identified, two moment conditions are available for each $u\in [0,1]$ by setting $z=0$ and $1$, but there are three unknowns. 

Now with the matching points and more generally, the m-connected points, equation \eqref{1.eq2} induces a larger system of equations for $\bm{g}^{*}(\bm{x}_{0},u)$. The number of the equations is determined by the sizes of $S(Z)$ and $\mathcal{X}_{MC}(\bm{x}_{0})$. As $(\bm{X},Z)=(\bm{x}_{0},0)$ and $(\bm{x}_{0},1)$ already provide two moment equations, the order condition is fulfilled as long as one nontrivial matching point of $\bm{x}_{0}$ exists. More discussion on the number of the matching points and the size difference of $S(D)$ and $S(Z)$ can be found in Appendix \ref{1.appA}. With more matching points and the m-connected points, we may have overidentification, and the rank condition introduced in Section \ref{1.sec3} for global identification is more likely to be satisfied.

It is worth noting that although each m-connected point (including $\bm{x}_{0}$ by definition) can induce two moment equations with $z=0,1$, in the four equations generated by two adjacently m-connected points (i.e., one is a matching point of the other), one equation is redundant. For example, suppose $(\bm{x}_{0},0)$ and $(\bm{x}_{m},1)$ are a matching pair. Equation \eqref{1.eq2} holds for $(\bm{x},z)=(\bm{x}_{0},0),(\bm{x}_{0},1),(\bm{x}_{m},0),(\bm{x}_{m},1)$, but it can be checked that the equations at $(\bm{x}_{0},0)$ and $(\bm{x}_{m},1)$ are identical. So in Figure \ref{1.fig1} (Example \ref{1.exOC3}), ten covariate-instrument combinations are available to build moment equations for $\bm{g}^{*}(\boldsymbol{x}_{0},u)$, but four of them are redundant (one in each pair of the arrow-connected points). 

Finally, utilizing the matching points is useful even when $S(D)=S(Z)$; the outcome function is then overidentified so that the validity of the instrument is testable together with the validity of the matching points.

\subsection{Finding the Matching Points}\label{1.sec3.1}
So far, I assumed that the matching points existed and were known. In this subsection, I discuss the existence of the matching points and how to find them.

First, equation \eqref{eq2.4a} in Lemma \ref{1.thmMEQ} provides a statistical implication for a matching point. For $\bm{x}_{m}\in S(\bm{X})$ to be a matching point of $\bm{x}_{0}$, there must exist $z,z'\in S(Z)$ such that
\begin{linenomath*}\begin{equation}\label{1.eq11}
\big(p_{1}(\boldsymbol{x}_{m},z')-p_{1}(\boldsymbol{x}_{0},z)\big)^{2}+\big(p_{2}(\boldsymbol{x}_{m},z')-p_{2}(\boldsymbol{x}_{0},z)\big)^{2}=0
\end{equation}\end{linenomath*}
The generalized propensity scores at $d=3$ are not included because they are matched automatically under \eqref{1.eq11}. Hence, for a given $\bm{x}_{0}$, if there does not exist $\bm{x}_{m},z$ and $z'$ that satisfy equation \eqref{1.eq11}, no matching point exists. 

The existence of a solution to equation \eqref{1.eq11} depends on how much $\boldsymbol{X}$ can affect the generalized propensity scores at $z'$ and how large $S(\bm{X})$ is. For instance, if the generalized propensity scores $(p_{1}(\bm{X},z'),p_{2}(\bm{X},z'))$ have full support, i.e., $(p_{1}(\cdot,z'),p_{2}(\cdot,z')):S(\boldsymbol{X})\mapsto [0,1]\times [0,1]$ is surjective, then a solution always exists. As another example, in Example \ref{1.exOC3}, $x_{0}\pm \alpha/\beta$ are two matching points if they are in $S(X)$. So $S(X)$ needs to be sufficiently large and/or $\alpha/\beta$ is small. The latter implies that the effects of $X$ on the generalized propensity scores are large compared to $Z$.

Conversely, not all the solutions to equation \eqref{1.eq11} are necessarily to be matching points. Sufficient conditions for both equations \eqref{1.eq3} and \eqref{1.eq2.5} in Definition \ref{1.defMP} to hold for all the solutions to \eqref{1.eq11} are that i) $(\bm{X},Z)$ enters the selection model only via the generalized propensity scores, and ii) $(\bm{X},Z)\independent (U_{d},\bm{V})$ for all $d\in S(D)$. 

Condition i) is often imposed in the literature on local instrumental variable (LIV) and marginal treatment effect (MTE) \citep{heckman1999local, heckman2001policy,heckman2005structural,heckman2006understanding}. Typically only selection models that are separable in $\bm{V}$ satisfy the condition. Appendices \ref{1.appA} and \ref{1.appC} provide some examples.

Note that the purpose of condition i) in this paper is different from that in the MTE and LIV literature. Here it guarantees that the selection model can be matched via the generalized propensity scores. Yet in the LIV and MTE literature, the instruments are continuous and the condition is imposed to obtain \textit{index sufficiency}, that is, for any $d\in S(D)$ and $y\in S(Y|d,\bm{X})$, $F_{Y|D\bm{X}Z}(y|d,\bm{X},Z)=F_{Y|D\bm{Xp}}\left(y|d,\bm{X},\bm{p}(\bm{X},Z)\right)$ a.s. (\cite{heckman2005structural}, p.678). Index sufficiency holds trivially when $Z$ is binary as in this paper's setup, provided that $\bm{p}(\bm{X},z)\neq \bm{p}(\bm{X},z')$ a.s. This is because in this case, $\bm{p}(\bm{X},Z)$ is a one-to-one function of $Z$ given $\bm{X}$. Hence, index sufficiency itself does not have identification power here.

When condition i) or ii) does not hold, for instance $\bm{X}$ is correlated with $(U_{d},\bm{V})$ for some $d\in S(D)$ or the selection model is nonseparable in $\bm{V}$, matching points may still exist because both requirements in Definition \ref{1.defMP} are local, yet conditions i) and ii) impose global restrictions on the selection model and the dependence of $(U_{d},\bm{V})$ on $(\bm{X},Z)$. Since the matching points form a subset of the solutions to equation \eqref{1.eq11}, one can treat the solutions as candidates for the matching points and obtain candidates for the m-connected points similarly. As long as the number of the moment equations from the instrument and such candidates is greater than $|S(D)|$, whether the candidates are truly matching points or m-connected points is testable.

\section{Identification}\label{1.sec3}
The fulfilled order condition makes identification of the outcome function $\bm{g}^{*}(\bm{x}_{0},\cdot)$ possible. This section provides a new result on the uniqueness of the solution to the system of nonlinear equations characterized by Propositions \ref{1.prop2}. The result relies on weaker conditions than those commonly used in the Hadamard-type global inverse function theorems. Since the moment conditions obtained by the nonparametric quantile IV approach with rich instruments are special cases of the ones obtained in this paper, this new result also applies there.

Proposition \ref{1.prop2} shows that for each $u\in [0,1]$, $\bm{g}^{*}(\bm{x}_{0},u)$ solves a system of nonlinear equations. Unlike identification of nonseparable models with a continuous $D$  \citep{chernozhukov2007instrumental,chen2014local}, here we do not face the ill-posed problem due to the discreteness of $D$. Nonetheless, establishing global identification of $\bm{g}^{*}(\bm{x}_{0},u)$ is still demanding. The Jacobian matrix of the nonlinear equation system being full rank at $\boldsymbol{g}^{*}(\boldsymbol{x}_{0},u)$ only implies local identification of $\boldsymbol{g}^{*}(\boldsymbol{x}_{0},u)$, and stronger high level conditions are required for global identification \citep{chernozhukov2005iv}.

However, the above approach with $u$ fixed does not exploit the structures of $\bm{g}^{*}(\bm{x}_{0},\cdot)$ as a function as well as the structures of the moment equations. By construction, $\bm{g}^{*}(\bm{x}_{0},\cdot)$ and the function $F_{Y|D\bm{X}Z}(\varphi_{d}(\cdot;\bm{x})|d,\bm{x},z)$ are continuous and strictly increasing on $[0,1]$ and on $S(Y|d,\bm{x}_{0})$ respectively. In this section, I show that with these properties, local identification of $\bm{g}^{*}(\bm{x}_{0},u)$ at every $u\in [0,1]$ guarantees global identification of $\boldsymbol{g}^{*}(\boldsymbol{x}_{0},\cdot)$ in the class of monotonic functions. This notion of identification is from a solution path perspective, common in differential equations and defined as follows.
\begin{definition}[Solution Path]
For an interval $\mathcal{U}\subseteq\mathbb{R}$ and a system of equations $\boldsymbol{M}(\boldsymbol{y},u)=\boldsymbol{0}$ where $\boldsymbol{y}$ is a real vector and $u\in\mathcal{U}$, a solution path $\boldsymbol{y}^{*}(\cdot)$ is a function on $\mathcal{U}$ such that $\boldsymbol{M}(\boldsymbol{y}^{*}(u),u)=\boldsymbol{0}$ for all $u\in\mathcal{U}$. 
\end{definition}
\begin{lem}\label{1.lemUnq}
Let $\mathcal{Y}\subseteq \mathbb{R}^{K}$, $\mathcal{U}\subseteq\mathbb{R}$ be a compact interval,  and $\bm{M}(\cdot,\cdot):\mathcal{Y}\times \mathcal{U}\mapsto \mathbb{R}^{L}$ be continuously differentiable. Suppose there exists a continuous and weakly increasing function $\bm{y}^{*}:\mathcal{U}\mapsto\mathcal{Y}$ such that $\bm{M}(\bm{y}^{*}(u),u)=\bm{0}$ for all $u\in\mathcal{U}$ and $\bm{y}^{*}(u^{*})=\bm{c}$ for some $u^{*}\in \mathcal{U}$ and $\bm{c}\in\mathcal{Y}$. If $\bm{M}(\cdot,u)$ is strictly increasing in each argument and its Jacobian matrix at $\bm{y}^{*}(u)$, $\nabla\bm{M}(\bm{y}^{*}(u),u)$, is full rank for all $u\in\mathcal{U}$, then $\bm{y}^{*}$ is the unique weakly increasing solution path that passes through $(u^{*},\bm{c})$.
\end{lem}
\begin{proof}
See Appendix \ref{1.appB}.
\end{proof}
\begin{rem} The lemma also holds for a system of equations with a decreasing solution path: If $\bm{y}$ is decreasing, let $\tilde{\bm{M}}((-\bm{y}),u)\equiv-\bm{M}(-(-\bm{y}),u)$, and thus $-\bm{y}$ is increasing and $\tilde{\bm{M}}(\cdot,u)$ as a function of $-\bm{y}$ is strictly increasing in every argument. Lemma \ref{1.lemUnq} thus applies. Similarly, $\bm{M}(\cdot,u)$ can be strictly decreasing as well.
\end{rem}

The proof of Lemma \ref{1.lemUnq} is in Appendix \ref{1.appB}. Here let me provide some heuristics to highlight the key roles played by monotonicity and continuity. Suppose $u^{*}$ is in the interior of $\mathcal{U}$ and there exists another increasing solution path $\tilde{\bm{y}}$ with $\tilde{\bm{y}}(u^{*})=\bm{c}$ but $\tilde{\bm{y}}(u)\neq \bm{y}^{*}(u)$ for all $u$ in some interval right to $u^{*}$. Then $\tilde{\bm{y}}(\cdot)$ cannot be continuous at $u^{*}$, otherwise there must exist some $u>u^{*}$ such that $\tilde{\bm{y}}(u)$ is close enough to $\bm{y}^{*}(u)$ that violates the local uniqueness of the solution implied by the full rank and continuous Jacobian. Consequently, $\tilde{\bm{y}}$ must jump up at $u^{*}$ since it is increasing. However, this is again not possible because otherwise, by continuity and monotonicity of $\bm{M}(\cdot,u^{*})$, $\bm{M}$ would jump up as well and thus the equation cannot hold at $\tilde{\bm{y}}(u^{*})$.

Although the uniqueness only holds among functions passing through the same point, this condition can be trivially satisfied in some special cases. For instance, let $\underline{u}$ be the lower boundary of $\mathcal{U}$. Suppose $\mathcal{Y}$ is the product of compact intervals and each component in $\bm{y}^{*}(\underline{u})$ equals the lower boundary of the corresponding interval, then all possible increasing solution paths $\tilde{\bm{y}}$ must satisfy $\tilde{\bm{y}}(\underline{u})=\bm{y}^{*}(\underline{u})$ because $\tilde{\bm{y}}(\underline{u})$ must be no smaller than $\bm{y}^{*}(\underline{u})$ to be in $\mathcal{Y}$, while if its greater than $\bm{y}^{*}(\underline{u})$, the system of equations cannot hold at $(\tilde{\bm{y}}(\underline{u}),\underline{u})$ because $\bm{M}(\cdot,\underline{u})$ is strictly increasing. This is indeed the case in this paper, as will be seen later in this section.

Lemma \ref{1.lemUnq} shows that monotonicity and continuity simplify the sufficient conditions usually required for the global uniqueness of a solution at a fixed $u$  (see \cite{ambrosetti1995primer} for variants of Hadamard's theorem). Here, the Jacobian matrix is just required to be full rank along the unique solution path, which only guarantees the local uniqueness of the solution at each fixed $u$. The lemma thus says that the local uniqueness of the solution pointwise in $u$ implies the global uniqueness of a monotonic solution path. Finally, the result holds among a class of functions where discontinuous functions are allowed. This is crucial to obtain the other results in this section.

Now let us turn to global identification of $\bm{g}^{*}(\bm{x}_{0},\cdot)$. Let $\mathcal{Z}(\boldsymbol{x}_{0})\equiv \mathcal{X}_{MC}(\boldsymbol{x}_{0})\times S(Z)$. For any three points $\tilde{\boldsymbol{z}}_{1},\tilde{\boldsymbol{z}}_{2},\tilde{\boldsymbol{z}}_{3}\in\mathcal{Z}(\bm{x}_{0})$, let $\Psi(\boldsymbol{g}^{*}(\boldsymbol{x}_{0},u);\tilde{\boldsymbol{z}}_{1},\tilde{\boldsymbol{z}}_{2},\tilde{\boldsymbol{z}}_{3}))$ denote the $3\times 1$ vector by stacking the left hand side of equation \eqref{1.eq2} evaluated at theses points respectively. For example, the $k$-th component in $\Psi$ is $\sum_{d=1}^{3}p_{d}(\tilde{\boldsymbol{z}}_{k})\cdot F_{Y|D\boldsymbol{X}Z}(\varphi_{d}(\cdot)|d,\tilde{\boldsymbol{z}}_{k})$. Denote the vector $(u,u,u)'$ by $\bm{u}$. Then $\boldsymbol{g}^{*}(\boldsymbol{x}_{0},\cdot)$ is one solution path to $\boldsymbol{M}(\boldsymbol{y},u)\equiv \Psi(\boldsymbol{y};\tilde{\boldsymbol{z}}_{1},\tilde{\boldsymbol{z}}_{2},\tilde{\boldsymbol{z}}_{3})-\boldsymbol{u}=\bm{0}$ on $[0,1]$. Let $\mathcal{G}$ be the set of all increasing functions defined on $[0,1]$:
\begin{linenomath*}\begin{equation}
\mathcal{G}\equiv\{\boldsymbol{g}:[0,1]\mapsto \mathbb{R}^{3}\text{ and is weakly increasing}\}\label{eqG}
\end{equation}\end{linenomath*}
The following theorem provides sufficient conditions that guarantee global identification of $\boldsymbol{g}^{*}(\bm{x}_{0},\cdot)$ in $\mathcal{G}$. 
\begin{thm}[Global Identification of $\bm{g}^{*}(\bm{x}_{0},\cdot)$]\label{1.thmIDNSP}
Under Assumptions \ref{1.assCM} to \ref{1.assRSS}, if there exist $\tilde{\boldsymbol{z}}_{1},\tilde{\boldsymbol{z}}_{2},\tilde{\boldsymbol{z}}_{3}\in \mathcal{Z}(\boldsymbol{x}_{0})$ such that $\Psi(\cdot;\tilde{\boldsymbol{z}}_{1},\tilde{\boldsymbol{z}}_{2},\tilde{\boldsymbol{z}}_{3})$ is continuously differentiable on $\prod_{d=1}^{3}S(Y|d,\boldsymbol{x}_{0})$, and that its Jacobian matrix at $\boldsymbol{g}^{*}(\boldsymbol{x}_{0},u)$ is full rank for all $u\in[0,1]$, then $\boldsymbol{g}^{*}(\boldsymbol{x}_{0},\cdot)$ is the unique solution path (up to $u=0,1$) to $\Psi(\cdot;\tilde{\boldsymbol{z}}_{1},\tilde{\boldsymbol{z}}_{2},\tilde{\boldsymbol{z}}_{3})-\boldsymbol{u}=0$ in $\mathcal{G}$. 
\end{thm}
\begin{proof}
See Appendix \ref{1.appB}.
\end{proof}
\begin{rem} The conditioning points $\tilde{\bm{z}}_{1},\tilde{\bm{z}}_{2}$ and $\tilde{\bm{z}}_{3}$ do not necessarily include $(\boldsymbol{x}_{0},z)$ and $(\boldsymbol{x}_{0},z')$. It is possible to use any points in $\mathcal{Z}(\bm{x}_{0})$ to achieve full rankness. Meanwhile, once $\bm{g}^{*}(\bm{x}_{0},\cdot)$ is identified, $\bm{g}^{*}(\bm{x},\cdot)$ is identified for all $\bm{x}\in\mathcal{X}_{MC}(\bm{x}_{0})$ by $\varphi_{d}(g^{*}_{d}(\bm{x}_{0},u);\bm{x})$.
\end{rem}

The proof of Theorem \ref{1.thmIDNSP} consists of two steps. In the first step, I invoke Lemma \ref{1.lemUnq} to show that the uniqueness holds in a smaller space $\mathcal{G}^{*}\equiv \{\boldsymbol{g}:[0,1]\mapsto \prod_{d=1}^{3}S(Y|d,\boldsymbol{x}_{0})\text{ and is weakly increasing}\}$, a subset of $\mathcal{G}$ defined in equation \eqref{eqG}. This step is straightforward by treating $[0,1]$ as $\mathcal{U}$ and $\prod_{d=1}^{3}S(Y|d,\boldsymbol{x}_{0})$ as $\mathcal{Y}$ in Lemma \ref{1.lemUnq}. Since $\bm{g}^{*}(\bm{x}_{0},\cdot)$ is continuous on the closed interval $[0,1]$, $\prod_{d=1}^{3}S(Y|d,\boldsymbol{x}_{0})$ is compact by Assumption \ref{1.assENSP} and $g^{*}_{d}(\bm{x}_{0},0)$ and $g_{d}^{*}(\bm{x}_{0},1)$ are equal to the lower and the upper boundaries of $S(Y|d,\bm{x}_{0})$ for all $d\in S(D)$. Since $F_{Y|DXZ}\circ\varphi_{d}$ in $\Psi$ is strictly increasing, all possible solutions in $\mathcal{G}^{*}$ at $u=0$ and $1$ must also equal these boundaries to make the system of equations hold at these $u$. All the conditions in Lemma \ref{1.lemUnq} are then satisfied.

In the second step, I show that the uniqueness indeed holds in the larger space $\mathcal{G}$ by exploiting the properties of the cumulative distribution functions in $\Psi$. Outside $S(Y|d,\boldsymbol{x}_{0})$, $F_{Y|DXZ}(\cdot|d,\bm{x}_{0},z)$ is either $0$ or $1$ and equal to the value at the corresponding boundary of $S(Y|d,\boldsymbol{x}_{0})$. Therefore, if there exists a second solution path $\bm{g}(u)$ which may take on values outside $\prod_{d=1}^{3}S(Y|d,\boldsymbol{x}_{0})$, there must also exist a function only taking on values within it (including the boundaries) which yields the same $\Psi$. That function is then in $\mathcal{G}^{*}$, so it is necessarily equal to $\bm{g}^{*}(\bm{x}_{0},\cdot)$. Therefore, $\bm{g}(u)$ has to be equal to $\bm{g}^{*}(\bm{x}_{0},u)$ for all $u\in (0,1)$ and can only take on values outside $\prod_{d=1}^{3}S(Y|d,\boldsymbol{x}_{0})$ at $u=0$ or $1$. See the proof in Appendix \ref{1.appB} for more details.

Allowing the parameter space to contain functions taking on values outside the conditional support $\prod_{d=1}^{3}S(Y|d,\boldsymbol{x}_{0})$ is useful in estimation because then one does not need to accurately estimate $\prod_{d=1}^{3}S(Y|d,\boldsymbol{x}_{0})$ to obtain a consistent estimator of $\bm{g}^{*}(\bm{x}_{0},\cdot)$. For instance, one can focus on the following parameter space 
\begin{linenomath*}\begin{equation}
\mathcal{G}_{0}\equiv \{\boldsymbol{g}:[0,1]\mapsto \prod_{d=1}^{3}S(Y|d)\text{ and is weakly increasing}\}\label{eqG0}
\end{equation}\end{linenomath*}
where $S(Y|d)$ is easier to estimate than $S(Y|d,\bm{x}_{0})$ with an estimator converging much faster due to the discreteness of $D$.

Define
\begin{linenomath*}\begin{equation}
Q_{NSP}(\boldsymbol{g},u)\equiv \left(\Psi\left(\boldsymbol{g}(u);\tilde{\bm{z}}_{1},\tilde{\bm{z}}_{2},\tilde{\bm{z}}_{3}\right)-\boldsymbol{u}\big)'\boldsymbol{W}_{NSP}(u)\big(\Psi\left(\boldsymbol{g}(u);\tilde{\bm{z}}_{1},\tilde{\bm{z}}_{2},\tilde{\bm{z}}_{3}\right)-\boldsymbol{u}\right)\label{eq3.1}
\end{equation}\end{linenomath*}
 where $\boldsymbol{W}_{NSP}(u)$ is positive definite uniformly in $u\in [0,1]$. Theorem \ref{1.thmIDNSP} implies that $\bm{g}^{*}(\bm{x}_{0},\cdot)$ is the unique increasing function such that $\int_{0}^{1}Q_{NSP}(\cdot,u)du=0$. Beyond that, it is important to know if $\int_{0}^{1}Q_{NSP}(\boldsymbol{g}(u),u)du$ is well separated from $0$ when $\boldsymbol{g}(\cdot)$ is well separated from $\boldsymbol{g}^{*}(\bm{x}_{0},\cdot)$, for instance, whether the following inequality holds for any $\delta>0$ and any closed interval $\mathcal{U}_{0}$ in the interior of $[0,1]$:
\begin{linenomath*}\begin{equation}\label{1.eq42}
\inf_{\substack{\boldsymbol{g}\in\mathcal{G}_{0}\\\sup_{u\in \mathcal{U}_{0}}|\boldsymbol{g}(u)-\boldsymbol{g}^{*}(x_{0},u)|\geq \delta}}\int_{0}^{1}Q_{NSP}(\boldsymbol{g}(u),u)du>0
\end{equation}\end{linenomath*}

It can be verified that the infinite dimensional space $\mathcal{G}_{0}$ is not compact under the sup-metric. In general, inequality \eqref{1.eq42} does not necessarily hold when the parameter space is noncompact even under the global uniqueness of $\bm{g}^{*}(\bm{x}_{0},\cdot)$ \citep{chen2007large,chen2012estimation}.  However, the following corollary shows that this is not a concern here. Again, monotonicity and continuity of $\bm{g}^{*}(\bm{x}_{0},\cdot)$ play the central role in it. 

\begin{cor}\label{1.corIDNSP}
Under the conditions in Theorem \ref{1.thmIDNSP}, inequality \eqref{1.eq42} is true.
\end{cor}
\begin{proof}
See Appendix \ref{1.appB}.
\end{proof}


It is noteworthy that since $\varphi_{d}(g^{*}_{d}(\boldsymbol{x}_{0},\cdot);\boldsymbol{x}_{0})=g^{*}_{d}(\boldsymbol{x}_{0},\cdot)$, Theorem \ref{1.thmIDNSP} and Corollary \ref{1.corIDNSP} also apply to the standard nonparametric quantile IV approach when $D$ is discrete with $|S(Z)|\geq |S(D)|$ (for example \cite{chernozhukov2005iv}).

Before closing this section, let me emphasize that global identification in terms of the solution path does not rule out the possibility that at some $u$, the solution to $\Psi(\cdot)=\boldsymbol{u}$ is not unique. This is expected because the conditions required here are much weaker than the sufficient conditions for global invertibility of $\Psi(\cdot)$ on $\Pi_{d=1}^{3} S(Y|d,\boldsymbol{x}_{0})$. Under this weaker notion of identification, one cannot estimate $\bm{g}^{*}(\bm{x}_{0},u)$ for a fixed $u$. In Appendix \ref{secSA2} in the Supplemental Material, I provide an estimator that minimizes the sample analogue of $Q_{NSP}(\cdot,u)$ jointly at multiple nodes of $u$ in $[0,1]$ under a monotonicity constraint. The number of the nodes needs to grow to infinity slowly with the sample size.

\section{A Special Case: The Separable Model}\label{1.secSP}
In some applications, the outcome disturbance $U_{d}$ may be additively separable. In this special case, identification results can be obtained under weaker conditions. Formally, suppose for each $d\in S(D)$, $g^{*}_{d}(\bm{X},U_{d})=m^{*}_{d}(\bm{X})+U_{d}$. The outcome equation \eqref{eq1.1} can then be rewritten as
\begin{linenomath*}\begin{equation}\label{1.sp}
Y=\sum_{d\in S(D)}\mathbbm{1}(D=d)\cdot\left(m_{d}^{*}(\boldsymbol{X})+U_{d}\right)\
\end{equation}\end{linenomath*}

Under separability, some requirements for the matching points and exogeneity of $Z$ can be relaxed. First, for $\bm{x}_{m}$ to be a matching point of $\bm{x}_{0}$, equation \eqref{1.eq2.5} in Definition \ref{1.defMP} can be weakened such that for all $d\in S(D)$,
\begin{linenomath*}\begin{equation}\label{1.eq2.4}
\mathbb{E}_{U_{d}|\boldsymbol{V}\boldsymbol{X}Z}(\bm{v},\boldsymbol{x}_{m},z')=\mathbb{E}_{U_{d}|\boldsymbol{V}\boldsymbol{X}Z}(\bm{v},\boldsymbol{x}_{0},z)\text{ and }F_{\boldsymbol{V}|\bm{X}Z}(\bm{v}|\boldsymbol{x}_{m},z')=F_{\boldsymbol{V}|\bm{X}Z}(\bm{v}|\boldsymbol{x}_{0},z)
\end{equation}\end{linenomath*}
Essentially, for $U_{d}$, only mean dependence of $U_{d}$ on $\bm{V}$ are required to be the same conditional on $(\bm{X},Z)=(\bm{x}_{m},z')$ and on $(\bm{x}_{0},z)$.

Then silmilar to Lemma \ref{1.thmMEQ}, the generalized propensity scores at $(\bm{x}_{m},z')$ and $(\bm{x}_{0},z)$ are equal, and the following equation holds for all $d\in S(D)$:
\begin{linenomath*}\begin{equation}
\mathbb{E}_{U_{d}|D\boldsymbol{X}Z}(d,\boldsymbol{x}_{m},z')=\mathbb{E}_{U_{d}|D\boldsymbol{X}Z}(d,\boldsymbol{x}_{0},z)\label{1.eq4}
\end{equation}\end{linenomath*}

With equation \eqref{1.eq4}, one can trace out the mapping from $m_{d}^{*}(\bm{x}_{0})$ to $m_{d}^{*}(\bm{x}_{m})$ for all $d\in S(D)$: Take expectations on both sides of equation \eqref{1.sp} conditional on $(D,\bm{X},Z)=(d,\bm{x},z)$: 
\begin{linenomath*}\begin{equation*}
m_{d}^{*}(\boldsymbol{x})+\mathbb{E}_{U_{d}|D\boldsymbol{X}Z}(d,\boldsymbol{x},z)=\mathbb{E}_{Y|D\boldsymbol{X}Z}(d,\boldsymbol{x},z)
\end{equation*}\end{linenomath*}
Evaluate the above equation at $(\bm{x}_{m},z')$ and $(\bm{x}_{0},z)$ and subtract one from the other. The conditional expectations of $U_{d}$ are canceled out by equation \eqref{1.eq4}. For any $\bm{x},\bm{x}'\in S(\bm{X})$, let $\delta_{d,(\bm{x},z),(\bm{x}',z')}=\mathbb{E}_{Y|D\boldsymbol{X}Z}(d,\boldsymbol{x}',z')-\mathbb{E}_{Y|D\boldsymbol{X}Z}(d,\boldsymbol{x},z)$. Then we have
\begin{linenomath*}\begin{equation}
m_{d}^{*}(\boldsymbol{x}_{m})=m_{d}^{*}(\boldsymbol{x}_{0})+\delta_{d,(\bm{x}_{0},z),(\bm{x}_{m},z')}\label{eq6}
\end{equation}\end{linenomath*}
The term $\delta_{d,(\bm{x}_{0},z),(\bm{x}_{m},z')}$ is directly identified from the population.

More generally, suppose $\bm{x}$ is m-connected with $\bm{x}_{0}$ such that $(\bm{x}_{0},z_{0})$ and $(\bm{x}_{1},z_{1})$, $(\bm{x}_{1},z_{1}')$ and $(\bm{x}_{2},z_{2})$, ..., $(\bm{x}_{k},z_{k}')$ and $(\bm{x},z)$ are matching pairs, where $z_{0},z_{1},z_{1}',...,z_{k}',z\in S(Z)$. Let 
\begin{linenomath*}\begin{equation*}
\Delta_{d}(\bm{x}_{0},\bm{x})=\delta_{d,(\bm{x}_{0},z_{0}),(\bm{x}_{1},z_{1})}+\delta_{d,(\bm{x}_{1},z_{1}'),(\bm{x}_{2},z_{2})}\cdots+\delta_{d,(\bm{x}_{k},z_{k}'),(\bm{x},z)},
\end{equation*}\end{linenomath*}
then we have $m^{*}_{d}(\bm{x})=m^{*}_{d}(\bm{x}_{0})+\Delta_{d}(\bm{x}_{0},\bm{x})$ for all $d\in S(D)$ and $\bm{x}\in\mathcal{X}_{MC}(\bm{x}_{0})$. By construction, $\Delta_{d}(\bm{x}_{0},\bm{x})$ is directly identified. Meanwhile, for a matching pair $(\bm{x}_{0},z)$ and $(\bm{x}_{m},z')$, $\Delta_{d}(\bm{x}_{0},\bm{x}_{m})=\delta_{d,(\bm{x}_{0},z),(\bm{x}_{m},z')}$. In particular, $\Delta_{d}(\bm{x}_{0},\bm{x}_{0})=0$.

Finally, the exogeneity assumption of $Z$ and rank similarity can be relaxed as follows due to separability: For all $\bm{x}\in\mathcal{X}_{MC}(\bm{x}_{0})$ and all $d\in S(D)$,
\begin{assumption}[Normalization and Exogeneity]\label{1.assESP} $\mathbb{E}_{U_{d}|\boldsymbol{X}}(\boldsymbol{x})=0$, $\mathbb{E}_{U_{d}|\boldsymbol{V}\boldsymbol{X}Z}(\boldsymbol{V},\boldsymbol{x},Z)=\mathbb{E}_{U_{d}|\boldsymbol{V}\boldsymbol{X}}(\boldsymbol{V},\boldsymbol{x})$ a.s., and $Z\independent\boldsymbol{V}$ conditional on $\boldsymbol{X}=\boldsymbol{x}$.
\end{assumption}
\begin{assumption}[Mean Similarity]\label{1.assMS}
Conditional on $\bm{X}=\bm{x}$, $\{U_{d}\}$ have the same expectation conditional on $\bm{V}$.
\end{assumption}
Assumption \ref{1.assESP} imposes a location normalization and joint (mean) independence of $Z$ and $(U_{d},\bm{V})$ for each $d\in S(D)$ conditional on any m-connected point. Joint (mean) independence is standard in the literature on triangular models with a separable outcome functions \citep{newey1999nonparametric}. 

Assumption \ref{1.assMS} relaxes rank similarity (Assumption \ref{1.assRSS}) in the nonseparable case. Instead of identical conditional distributions of $\{U_{d}\}$, only the conditional expectations are required to be identical. 

The following proposition characterizes the moment condition for $\bm{m}^{*}(\bm{x}_{0})\equiv (m^{*}_{d}(\bm{x}_{0}))_{d}$. 

\begin{prop}[Moment Condition]\label{1.prop1} Under Assumptions \ref{1.assESP} and \ref{1.assMS}, the following equation holds for all $z\in S(Z)$ and $\bm{x}\in\mathcal{X}_{MC}(\bm{x}_{0})$,
\begin{linenomath*}\begin{equation}\label{1.eq1}
\sum_{d\in S(D)}p_{d}(\boldsymbol{x},z)\cdot m_{d}^{*}(\boldsymbol{x}_{0})=\sum_{d\in S(D)}p_{d}(\boldsymbol{x},z)\cdot \left(\mathbb{E}_{Y|D\boldsymbol{X}Z}(d,\boldsymbol{x},z)-\Delta_{d}(\bm{x}_{0},\bm{x})\right)
\end{equation}\end{linenomath*}
\end{prop} 
\begin{proof}
See Appendix \ref{1.appB}.
\end{proof}
Similar to the nonseparable case, when $\bm{x}=\bm{x}_{0}$, $\Delta_{d}(\bm{x}_{0},\bm{x})=0$ for all $d\in S(D)$ and then equation \eqref{1.eq1} is back to equation (2.2) in \cite{newey2003instrumental} or equation (2.5) in \cite{das2005instrumental}. 
 
Global identification of $\bm{m}^{*}(\bm{x}_{0})$ is straightforward to establish due to linearity of equation \eqref{1.eq1}. Recall the augmented set of the conditioning points $\mathcal{Z}(\boldsymbol{x}_{0})\equiv \mathcal{X}_{MC}(\boldsymbol{x}_{0})\times S(Z)$. Evaluating equation \eqref{1.eq1} at any three points $\tilde{\boldsymbol{z}}_{1},\tilde{\boldsymbol{z}}_{2},\tilde{\boldsymbol{z}}_{3}\in\mathcal{Z}(\boldsymbol{x}_{0})$ yields a system of linear equations of $\bm{m}^{*}(\bm{x}_{0})$ with the coefficient matrix equal to $\Pi_{SP}(\tilde{\boldsymbol{z}}_{1},\tilde{\boldsymbol{z}}_{2},\tilde{\boldsymbol{z}}_{3})\equiv (\bm{p}(\tilde{z}_{1}),\bm{p}(\tilde{z}_{2}),\bm{p}(\tilde{z}_{3}))'$, where $\bm{p}$ is the column vector of the three generalized propensity scores $(p_{1},p_{2},p_{3})'$. Then $\bm{m}^{*}(\bm{x}_{0})$ is globally identified on $\mathbb{R}^{3}$ if 
 \begin{linenomath*}\begin{equation}
 \Pi_{SP}(\tilde{\boldsymbol{z}}_{1},\tilde{\boldsymbol{z}}_{2},\tilde{\boldsymbol{z}}_{3})\text{ is full rank.}
 \end{equation}\end{linenomath*}



Finally, let us discuss the sufficient and necessary conditions for the full rankness of $\Pi_{SP}$. For concreteness, let $(\boldsymbol{x}_{m},z')$ and $(\boldsymbol{x}_{0},z)$ be a matching pair. Let $\tilde{\bm{z}}_{1}=(\bm{x}_{0},z),\tilde{\bm{z}}_{2}=(\bm{x}_{0},z')$ and $\tilde{\bm{z}}_{3}=(\bm{x}_{m},z)$. The moment equation at $(\bm{x}_{m},z')$ is not included as it is identical with the equation at $(\boldsymbol{x}_{0},z)$.

Since the sum of the three columns in $\Pi_{SP}$ is always equal to $\boldsymbol{1}_{3\times 1}$, it can be shown that $\Pi_{SP}$ is full rank if and only if 
\begin{linenomath*}\begin{align}
&\left[p_{1}(\boldsymbol{x}_{m},z)-p_{1}(\boldsymbol{x}_{0},z)\right]\cdot\left[p_{3}(\boldsymbol{x}_{0},z)-p_{3}(\boldsymbol{x}_{0},z')\right]\notag\\
\neq& \left[p_{1}(\boldsymbol{x}_{0},z)-p_{1}(\boldsymbol{x}_{0},z')\right]\cdot \left[p_{3}(\boldsymbol{x}_{m},z)-p_{3}(\boldsymbol{x}_{0},z)\right]\label{1.eq13}
\end{align}\end{linenomath*}
Inequality \eqref{1.eq13} does not hold if both sides are simultaneously zero. This is the case when $Z$ has no effect on $\boldsymbol{p}$ at $\boldsymbol{X}=\boldsymbol{x}_{0}$ or $\boldsymbol{X}\in\{\bm{x}_{0},\bm{x}_{m}\}$ has no effect on $\boldsymbol{p}$ at $Z=z$. Both can be ruled out by a local relevance condition saying that $\boldsymbol{X}$ and $Z$ have nonzero effects on the propensity scores at $(\boldsymbol{x}_{0},z)$. 

Now suppose neither side is $0$. By $\boldsymbol{p}(\boldsymbol{x}_{0},z)=\boldsymbol{p}(\boldsymbol{x}_{m},z')$, inequality \eqref{1.eq13} can be rewritten as
\begin{linenomath*}\begin{equation}\label{1.eq14}
\frac{p_{1}(\boldsymbol{x}_{m},z)-p_{1}(\boldsymbol{x}_{0},z)}{p_{3}(\boldsymbol{x}_{m},z)-p_{3}(\boldsymbol{x}_{0},z)}\neq
\frac{p_{1}(\boldsymbol{x}_{m},z')-p_{1}(\boldsymbol{x}_{0},z')}{p_{3}(\boldsymbol{x}_{m},z')-p_{3}(\boldsymbol{x}_{0},z')}
\end{equation}\end{linenomath*}

Inequality \eqref{1.eq14} generally holds unless the propensity score differences are locally uniform. For example, it can be verified that the inequality holds in the ordered choice model in Example \ref{1.exOC3} for almost all $x_{0}\in S(X)$ and their matching points unless $V$ is (locally) uniformly distributed. In particular, it holds in widely used logit and probit models. 

\section{Estimation}\label{1.sec4}
In this section, I illustrate how to estimate the matching points and the separable model given an i.i.d. sample $(Y_{i},D_{i},\boldsymbol{X}_{i},Z_{i})_{i=1}^{n}$. The estimator for the nonseparable model and its properties are introduced and proved in Appendix \ref{secSA2} in the Supplemental Material. For illustrative purposes, I focus on the following benchmark case to highlight the key features of the estimation procedure: i) $\boldsymbol{X}$ is one dimensional, denoted by $X$, and ii) all the solutions to the generalized propensity score matching equation \eqref{1.eq11} are the matching points. Let $(x_{0},0)$ and $(x_{m1},1)$, and $(x_{0},1)$ and $(x_{m2},0)$ be two matching pairs. This benchmark case is the simplest scenario where both the matching points and the outcome functions are overidentified, allowing me to introduce the overidentification tests. Extending a scalar $X$ to the multivariate case is straightforward. At the same time, whether the generalized propensity score matching is successful and whether a solution to the matching is a matching point are testable by the overidentification tests provided in this section.

\subsection{Estimating the Matching Points}\label{sec4.1}
Let $S_{0}(X)$ be a compact subset in the interior of $S(X)$. For $z\in S(Z)$, let $\hat{\boldsymbol{p}}(\cdot,z)$ be a consistent estimator of the $3\times 1$ vector of the generalized propensity scores $\boldsymbol{p}(\cdot,z)$ uniformly on $S_{0}(X)$. For concreteness, I consider the following Nadaraya-Watson estimator for each component in it. Other common nonparametric estimators of conditional probability would work as well.
\begin{linenomath*}\begin{equation}\label{1.eq23}
\hat{p}_{d}(x,z)=\frac{\sum_{i=1}^{n}\mathbbm{1}(D_{i}=d)K(\frac{X_{i}-x}{h_{x}})\mathbbm{1}(Z_{i}=z)}{\sum_{i=1}^{n} K(\frac{X_{i}-x}{h_{x}})\mathbbm{1}(Z_{i}=z)}
\end{equation}\end{linenomath*}
where $K(\cdot)$ is a kernel function and $h_{x}$ is the bandwidth converging to $0$. 

Assume both matching points are in $S_{0}(X)$. The matching points can be estimated by the sample analogue of equation \eqref{1.eq11}. Denote $\Delta \hat{\boldsymbol{p}}(x_{1},x_{2})\equiv \big(\hat{p}_{1}(x_{1},1)-\hat{p}_{1}(x_{0},0),\hat{p}_{2}(x_{1},1)-\hat{p}_{2}(x_{0},0),\hat{p}_{1}(x_{2},0)-\hat{p}_{1}(x_{0},1),\hat{p}_{2}(x_{2},0)-\hat{p}_{2}(x_{0},1)\big)'$ and its propability limit by $\Delta \boldsymbol{p}(x_{1},x_{2})$. For some weighting matrix $\boldsymbol{W}_{xn}$ with a positive definite probability limit, let $\hat{Q}_{x}(x_{1},x_{2})\equiv \Delta \hat{\boldsymbol{p}}(x_{1},x_{2})'\boldsymbol{W}_{xn}\Delta\hat{\boldsymbol{p}}(x_{1},x_{2})$. Define the estimator $(\hat{x}_{m1},\hat{x}_{m2})$ as any point in $S_{0}^{2}(X)$ such that
for some $a_{n}=o(1)$, 
\begin{linenomath*}\begin{equation}\label{1.eq22}
\hat{Q}_{x}(\hat{x}_{m1},\hat{x}_{m2})\leq \inf_{S_{0}^{2}(X)} \hat{Q}_{x}(x_{1},x_{2})+a_{n}^{2}
\end{equation}\end{linenomath*}
This estimator is adapted from \cite{chernozhukov2007estimation} for partially identified parameters. It is applicable here because the solution to the generalized propensity score matching may not be unique. For simplicity, in this section I focus on the case where $(x_{m1},x_{m2})$ is unique and let $a_{n}=0$. Consistency and asymptotic normality of $(\hat{x}_{m1},\hat{x}_{m2})$ then follow from the standard arguments for (local) GMM estimators under Assumption \ref{1.assREGMP}. The proofs are omitted. The general case with multiple solutions to the matching and $a_{n}>0$ is discussed in Appendix \ref{1.appA2} in the Supplemental Material. 

\begin{assumption}\label{1.assREGMP}
Let $f_{XZ}\equiv \mathbb{P}_{Z|X}\cdot f_{X}$ and $f_{DXZ}\equiv \mathbb{P}_{D|XZ}\cdot f_{XZ}$. For every $d\in S(D)$ and $z\in S(Z)$, $f_{DXZ}(d,\cdot,z)$ and $f_{XZ}(\cdot,z)$ are three times continuously differentiable on $S(X)$ with bounded derivatives, and are bounded away from zero on $S_{0}(X)$. The kernel $K(\cdot)$ is positive, symmetric at $0$, continuously differentiable on $\mathbb{R}$ with bounded derivative, and supported on $[-1,1]$.
\end{assumption}
Under Assumption \ref{1.assREGMP}, it can be shown that $||(\hat{x}_{m1},\hat{x}_{m2})-(x_{m1},x_{m2})||=o_{p}(1)$. For the asymptotic distribution, denote the Jacobian matrix of $\Delta\boldsymbol{p}(x_{1},x_{2})$ evaluated at $x_{m1}$ and $x_{m2}$ by $\partial_{x'}\Delta\boldsymbol{p}(x_{m1},x_{m2})$. Let $\tilde{\boldsymbol{z}}_{1},...,\tilde{\boldsymbol{z}}_{4}$ be $(x_{0},0)$, $(x_{0},1)$, $(x_{m1},1)$ and $(x_{m2},0)$ respectively. Let $\kappa\equiv \int K(v)^{2}dv$ and $\Sigma_{x}=\kappa\begin{pmatrix}
\Sigma_{x1} & \boldsymbol{0}\\
\boldsymbol{0}    & \Sigma_{x2}
\end{pmatrix}$, where for $k=1,2$, 
\[\Sigma_{xk}=\begin{pmatrix}
\frac{p_{1}(\tilde{\boldsymbol{z}}_{k})(1-p_{1}(\tilde{\boldsymbol{z}}_{k}))}{f_{XZ}(\tilde{\boldsymbol{z}}_{k})}+\frac{p_{1}(\tilde{\boldsymbol{z}}_{k+2})(1-p_{1}(\tilde{\boldsymbol{z}}_{k+2}))}{f_{XZ}(\tilde{\boldsymbol{z}}_{k+2})}&
-\frac{p_{1}(\tilde{\boldsymbol{z}}_{k})p_{2}(\tilde{\boldsymbol{z}}_{k})}{f_{XZ}(\tilde{\boldsymbol{z}}_{k})}-\frac{p_{1}(\tilde{\boldsymbol{z}}_{k+2})p_{2}(\tilde{\boldsymbol{z}}_{k+2})}{f_{XZ}(\tilde{\boldsymbol{z}}_{k+2})}\\
& &\\
-\frac{p_{1}(\tilde{\boldsymbol{z}}_{k})p_{2}(\tilde{\boldsymbol{z}}_{k})}{f_{XZ}(\tilde{\boldsymbol{z}}_{k})}-\frac{p_{1}(\tilde{\boldsymbol{z}}_{k+2})p_{2}(\tilde{\boldsymbol{z}}_{k+2})}{f_{XZ}(\tilde{\boldsymbol{z}}_{k+2})} &
\frac{p_{2}(\tilde{\boldsymbol{z}}_{k})(1-p_{2}(\tilde{\boldsymbol{z}}_{k}))}{f_{XZ}(\tilde{\boldsymbol{z}}_{3})}+\frac{p_{2}(\tilde{\boldsymbol{z}}_{k+2})(1-p_{2}(\tilde{\boldsymbol{z}}_{k+2}))}{f_{XZ}(\tilde{\boldsymbol{z}}_{k+2})}\end{pmatrix}.\] 
Using the standard two step GMM procedure, let $\hat{\Sigma}_{x}^{-1}$ be the feasible optimal weighting matrix obtained by a consistent first step estimator (for instance using the identity matrix as the weighting matrix). If $(x_{m1},x_{m2})$ lies in the interior of $S^{2}_{0}(X)$ and $\partial_{x'}\Delta\boldsymbol{p}(x_{m1},x_{m2})$ is nonsingular, $(\hat{x}_{m1},\hat{x}_{m2})$ has the following asymptotic distribution under undersmoothing $h_{x}^{2}\cdot \sqrt{nh_{x}}=o(1)$ and $\sqrt{nh_{x}^{3}}\to\infty$ (guaranteeing the derivatives of $\hat{\bm{p}}(z,\cdot)$ are also uniformly consistent):
\begin{linenomath*}\begin{equation}\label{1.eq31}
\sqrt{nh_{x}}\begin{pmatrix}
\hat{x}_{m1}-x_{m1}\\
\hat{x}_{m2}-x_{m2}
\end{pmatrix}\overset{d}{\to} \mathcal{N}\left(0,\left[\partial_{x'}\Delta\boldsymbol{p}(x_{m1},x_{m2})\Sigma_{x}^{-1}\partial_{x}\Delta\boldsymbol{p}(x_{m1},x_{m2})\right]^{-1}\right)
\end{equation}\end{linenomath*}

Since only one covariate is in the model, each matching point is a scalar that matches two generalized propensity scores. So $(x_{m1},x_{m2})$ is overidentified. The null hypothesis $\mathbb{H}_{0}: \Delta\boldsymbol{p}(x_{m1},x_{m2})=\boldsymbol{0}$ can be tested by the J test 
\begin{linenomath*}
\[\mathcal{J}_{x}=nh_{x} \Delta\hat{\boldsymbol{p}}(\hat{x}_{m1},\hat{x}_{m2})'\widehat{\Sigma}_{x}^{-1}\Delta\hat{\boldsymbol{p}}(\hat{x}_{m1},\hat{x}_{m2})\]
\end{linenomath*} Under the null, $\mathcal{J}_{x}\overset{d}{\to} \chi_{2}^{2}$. In addition to jointly testing whether $(x_{m1},x_{m2})$ solves the propensity score matching equation, one can separately test either one of them when needed. By block-diagonality of the asymptotic variance, $\hat{x}_{m1}$ and $\hat{x}_{m2}$ are asymptotically independent, and thus it is equivalent to estimate the two matching points separately. In each separate problem, the matching point is still overidentified. Let $\mathcal{J}_{x1}=nh_{x}\Delta\hat{\boldsymbol{p}}(\hat{x}_{m1})'(\kappa\widehat{\Sigma}_{x1})^{-1}\Delta\hat{\boldsymbol{p}}(\hat{x}_{m1})$ and
$\mathcal{J}_{x2}=nh_{x}\Delta\hat{\boldsymbol{p}}(\hat{x}_{m2})'(\kappa\widehat{\Sigma}_{x2})^{-1}\Delta\hat{\boldsymbol{p}}(\hat{x}_{m2})$, where $\Delta\hat{\boldsymbol{p}}(\hat{x}_{m1})$ and $\Delta\hat{\boldsymbol{p}}(\hat{x}_{m2})$ are subvectors of $\Delta\hat{\boldsymbol{p}}(\hat{x}_{m1},\hat{x}_{m2})$ containing its first and last two elements respectively. Each test statistic converges in distribution to $\chi_{1}^{2}$ under the null.

\subsection{Estimating the Separable Model}
For the separable model, linearity of the moment conditions \eqref{1.eq1} yields a closed form estimator. Assume $\Pi_{SP}\equiv \left(\bm{p}(x_{0},0),\bm{p}(x_{0},1),\bm{p}(x_{m1},0),\bm{p}(x_{m2},1)\right)'$ is full rank. Let $\hat{\delta}_{(x_{0},0),(\hat{x}_{m1},1)}=\widehat{\mathbb{E}}_{Y|DXZ}(d,\hat{x}_{m1},1)-\widehat{\mathbb{E}}_{Y|DXZ}(d,x_{0},0)$ and $\hat{\delta}_{(x_{0},1),(\hat{x}_{m2},0)}$ be defined similarly. Then with a weighting matrix $\bm{W}_{mn}$ that has a positive definite probability limit, let
\begin{linenomath*}\begin{equation}\label{1.eq24}
\hat{\boldsymbol{m}}(x_{0})
=(\widehat{\Pi}_{SP}'\boldsymbol{W}_{mn}\widehat{\Pi}_{SP})^{-1}\cdot \widehat{\Pi}_{SP}'\boldsymbol{W}_{mn}\widehat{\Phi}
\end{equation}\end{linenomath*}
where $\widehat{\Pi}_{SP}=\left(\hat{\bm{p}}(x_{0},0),\hat{\bm{p}}(x_{0},1),\hat{\bm{p}}(\hat{x}_{m1},0),\hat{\bm{p}}(\hat{x}_{m2},1)\right)'$ and
 \[
\widehat{\Phi}=
\begin{pmatrix}
\sum_{d=1}^{3}\widehat{\mathbb{E}}_{Y|DXZ}(d,x_{0},0)\hat{p}_{d}(x_{0},0)\\
\sum_{d=1}^{3}\widehat{\mathbb{E}}_{Y|DXZ}(d,x_{0},1)\hat{p}_{d}(x_{0},1)\\
\sum_{d=1}^{3}\big[\widehat{\mathbb{E}}_{Y|DXZ}(d,\hat{x}_{m1},0)-\hat{\delta}_{(x_{0},0),(\hat{x}_{m1},1)}\big]\hat{p}_{d}(\hat{x}_{m1},0)\\
\sum_{d=1}^{3}\big[\widehat{\mathbb{E}}_{Y|DXZ}(d,\hat{x}_{m2},1)-\hat{\delta}_{(x_{0},1),(\hat{x}_{m2},0)}\big]\hat{p}_{d}(\hat{x}_{m2},1)
\end{pmatrix}.
\] 
The propensity score estimators are as equation \eqref{1.eq23}. The conditional expectations can be estimated by the Nadaraya-Watson estimator: 
\begin{linenomath*}\begin{equation}\label{1.eq25}
\widehat{\mathbb{E}}_{Y|DXZ}(d,x,z)=\frac{\sum_{i=1}^{n}Y_{i}\mathbbm{1}(D_{i}=d)K(\frac{X_{i}-x}{h_{m}})\mathbbm{1}(Z_{i}=z)}{\sum_{i=1}^{n}\mathbbm{1}(D_{i}=d) K(\frac{X_{i}-x}{h_{m}})\mathbbm{1}(Z_{i}=z)}
\end{equation}\end{linenomath*} 

\begin{assumption}\label{1.assREGSP}
The conditional variance of $Y$, $\mathbb{V}_{Y|DXZ}(d,\cdot,z)$, is finite and continuous on $S(X)$ for each $(d,z)\in S(D,Z)$. The conditional expectation $\mathbb{E}_{Y|DXZ}(d,\cdot,z)$ is three times continuously differentiable on $S(X)$ with bounded derivatives.
\end{assumption}
Under Assumptions \ref{1.assREGMP} and \ref{1.assREGSP}, every component in the right hand side of equation \eqref{1.eq24} is uniformly consistent on $S_{0}(X)$. Under consistency of $(\hat{x}_{m1},\hat{x}_{m2})$, $||\hat{\boldsymbol{m}}(x_{0})-\boldsymbol{m}^{*}(x_{0})||=o_{p}(1)$.

For the asymptotic distribution, I let $h_{m}/h_{x}\to 0$ so that the impacts of estimating $(x_{m1},x_{m2})$ and the generalized propensity scores are negligible. Let $\tilde{\boldsymbol{z}}_{1},...,\tilde{\boldsymbol{z}}_{6}$ be $(x_{0},0)$, $(x_{m1},0)$, $(x_{m1},1)$, $(x_{0},1)$, $(x_{m2},1)$ and $(x_{m2},0)$. Let $\Sigma_{SP}=\kappa(\Sigma_{SP,1}+\Sigma_{SP,2}+\Sigma_{SP,3})$ and $\Sigma_{SP,d}$ ($d=1,2,3$) equal
\begin{linenomath*}\begin{equation*}\small
\setlength\arraycolsep{0pt}
\begin{pmatrix}
\frac{p_{d}(\tilde{\boldsymbol{z}}_{1})^{2}\mathbb{V}_{Y|DXZ}(d,\tilde{\boldsymbol{z}}_{1})}{f_{DXZ}(d,\tilde{\boldsymbol{z}}_{1})}&0& \frac{p_{d}(\tilde{\boldsymbol{z}}_{1})p_{d}(\tilde{\boldsymbol{z}}_{2})\mathbb{V}_{Y|DXZ}(d,\tilde{\boldsymbol{z}}_{1})}{f_{DXZ}(d,\tilde{\boldsymbol{z}}_{1})}&0\\
0 &\frac{p_{d}(\tilde{\boldsymbol{z}}_{4})^{2}\mathbb{V}_{Y|DXZ}(d,\tilde{\boldsymbol{z}}_{4})}{f_{DXZ}(d,\tilde{\boldsymbol{z}}_{4})}&0&\frac{p_{d}(\tilde{\boldsymbol{z}}_{4})p_{d}(\tilde{\boldsymbol{z}}_{5})\mathbb{V}_{Y|DXZ}(d,\tilde{\boldsymbol{z}}_{4})}{f_{DXZ}(d,\tilde{\boldsymbol{z}}_{4})}\\
\frac{p_{d}(\tilde{\boldsymbol{z}}_{1})p_{d}(\tilde{\boldsymbol{z}}_{2})\mathbb{V}_{Y|DXZ}(d,\tilde{\boldsymbol{z}}_{1})}{f_{DXZ}(d,\tilde{\boldsymbol{z}}_{1})}&0 & \sum_{k=1}^{3}\frac{p_{d}(\tilde{\boldsymbol{z}}_{2})^{2}\mathbb{V}_{Y|DXZ}(d,\tilde{\boldsymbol{z}}_{k})}{f_{DXZ}(d,\tilde{\boldsymbol{z}}_{k})}&0\\
0& \frac{p_{d}(\tilde{\boldsymbol{z}}_{4})p_{d}(\tilde{\boldsymbol{z}}_{5})\mathbb{V}_{Y|DXZ}(d,\tilde{\boldsymbol{z}}_{4})}{f_{DXZ}(d,\tilde{\boldsymbol{z}}_{4})}& 0 & \sum_{k=4}^{6}\frac{p_{d}(\tilde{\boldsymbol{z}}_{5})^{2}\mathbb{V}_{Y|DXZ}(d,\tilde{\boldsymbol{z}}_{k})}{f_{DXZ}(d,\tilde{\boldsymbol{z}}_{k})}
\end{pmatrix}.
\end{equation*}\end{linenomath*}\normalsize Let $W_{mn}=\hat{\Sigma}_{SP}^{-1}$. Then if $h_{m}^{2}\cdot \sqrt{nh_{m}}\to 0$ and $\sqrt{nh_{m}^{3}}\to\infty$,

%
\begin{linenomath*}\begin{equation}\label{1.eq39}
\sqrt{nh_{m}}\left(\hat{\boldsymbol{m}}(x_{0})-\boldsymbol{m}^{*}(x_{0})\right) \overset{d}{\to} \mathcal{N}\big(0,(\Pi_{SP}'\Sigma_{SP}^{-1}\Pi_{SP})^{-1}\big)
\end{equation}\end{linenomath*}

Since there are four moment equations, $\bm{m}^{*}(\bm{x}_{0})$ is overidentified. Let the overidentification test statistic be $\mathcal{J}_{SP}=nh_{m}\big(\widehat{\Pi}_{SP}\hat{\boldsymbol{m}}(x_{0})-\widehat{\Phi}_{SP}\big)'\widehat{\Sigma}_{SP}^{-1}\big(\widehat{\Pi}_{SP}\hat{\boldsymbol{m}}(x_{0})-\widehat{\Phi}_{SP}\big)$. The null hypothesis is that all of the four moment conditions hold. Validity of the moment conditions jointly depend on the exogeneity of the instrument and validity of the matching points obtained by generalized propensity score matching. Under the null, it can be verified that $\mathcal{J}_{SP}\overset{d}{\to} \chi_{1}^{2}$ following the standard argument in the GMM framework.

\section{Monte Carlo Simulations}\label{1.sec6}
This section illustrates the finite sample performance of the estimator. The endogenous variable $D$ follows the ordered choice model in Example \ref{1.exOC3}. The instrument $Z$ is binary. Let the outcome variable $Y$ be determined by the following model:
\begin{linenomath*}\begin{align*}
Y=[\gamma_{1}\mathbbm{1}(D=1)+&\gamma_{2}\mathbbm{1}(D=2)+\gamma_{3}\mathbbm{1}(D=3)]\cdot (X+1)+U
\end{align*}\end{linenomath*}
where $X$ is drawn from $\text{Unif}[-3,3]$, $Z$ from a Bernoulli distribution with parameter $0.5$, $[U,V]$ from $\mathcal{N}\left(0,\begin{pmatrix}
1&\rho\\\rho&1
\end{pmatrix} \right)$, and $X\independent Z\independent (U,V)$. 

For the parameters, I set $(\gamma_{1},\gamma_{2},\gamma_{3},\kappa_{1},\kappa_{2})=(1.5,3,3.5,-0.7,0.1)$. The parameters $(\alpha,\beta)$ govern the strength of the instrument and the covariate. In this section I present the results under $(\alpha,\beta)=(0.8,0.4)$. The parameter values are selected for two reasons: i) all the generalized propensity scores are away from 0 so that in the simulated sample, there are a sufficient number of observations to estimate the conditional expectation and the generalized propensity score for each $d$, and ii) $X$ and $Z$ have large effects on the generalized propensity scores. Finally, I set $\rho=0.5$ and $x_{0}=0$. Additional simulation results for small $(\alpha,\beta)$, different $\rho$ and different $x_{0}$ are provided in Appendix \ref{secSB} in the Supplemental Material. 

\begin{table}[t]
\centering
\caption{$x_{0}=0$. $\boldsymbol{m}^{*}(0)=(1.5,3,3.5)$.}\label{1.tab1}
\begin{tabular}{C{1cm} C{1cm} C{1.3cm} C{1.3cm} C{1.3cm} C{1.3cm} C{1.3cm} C{1.3cm} C{1.3cm}}
\hline
\hline
&$n$ & Average & $\text{Bias}^{2}$& Variance & MSE & 90\% & 95\% & 99\% \\
\hline
\multirow {3}{*}{$\hat{m}_{1}(0)$} & 1000 &1.49 &$2\cdot 10^{-4}$&0.12&0.12&90.2\%& 95.4\% & 99\%  \tabularnewline
                                                     &2000 &1.51 & $3\cdot 10^{-5}$  &0.06&0.06& 91.6\%& 96\%&99\%   \tabularnewline
                                                     &3000 &1.49 & $4\cdot 10^{-5}$ &0.04&0.04& 88.4\%&94.4\%&99\% \tabularnewline
\hline 
\multirow {3}{*}{$\hat{m}_{2}(0)$} & 1000 &2.89&0.01&0.78&0.79&93.2\%& 96.2\% & 99\% \tabularnewline
                                                     &2000 &2.88 &0.01&0.37&0.39& 89.6\%& 95\%  &99\%  \tabularnewline
                                                     &3000 &2.92 & 0.01&0.25&0.26& 89.2\%&95.6\% &99.8\%\tabularnewline
\hline 
\multirow {3}{*}{$\hat{m}_{3}(0)$} & 1000 &3.47&0.001&0.22&0.22&92.8\% & 97\%& 98.6\%\tabularnewline
                                                  &2000 &3.49 &$2\cdot 10^{-4}$&0.12&0.12 & 92.2\%& 97.2\%  &98.8\%\tabularnewline
                                                     &3000 &3.49 & $10^{-4}$&0.07&0.07& 92.6\%&97.2\% &99.4\%\tabularnewline
\hline 

\multirow {3}{*}{$\mathcal{J}_{x}$} &1000&             &           &            &            &  90\%  &   95\%   &99.2\%\tabularnewline
                                                     &2000&             &           &            &            &  91.6\%  &   95.8\%   &99.6\%\tabularnewline
                                                     &3000&             &           &            &            &  92.6\%  &   96.8\%   &99.2\%\tabularnewline
\hline 
\multirow {3}{*}{$\mathcal{J}_{SP}$} &1000&             &           &            &  &  91.6\%  &   94.8\%   &98.2\%\tabularnewline
                                                     &2000&             &           &            &            &  93.6\%  &   96.4\%   &98.4\%\tabularnewline
                                                     &3000&             &           &            &            &  91.6\%  &   96.2\%   &98.8\%\tabularnewline
\hline 
\end{tabular}
\end{table}

Table \ref{1.tab1} contains the results for sample size $n=1000$, 2000 and 3000. The number of simulation replications is 500. In each replication, I estimate $(x_{m1},x_{m2})$ using grid search with $500$ grid nodes. The propensity scores and the conditional expectations are estimated as proposed in Section \ref{1.sec4} with the biweight kernel. A smaller bandwidth is chosen when estimating the outcome function than the one used to estimate the matching points. The actual coverage probabilities of the confidence intervals for $\boldsymbol{m}^{*}(x_{0})$ are computed using the asymptotic variance estimator.  The coverage probabilities of the overidentification tests for $(x_{m1},x_{m2})$ and for $\boldsymbol{m}^{*}(x_{0})$ are also reported.

As is shown in Table \ref{1.tab1}, the variance of the estimator dominates in mean squared error (MSE) due to undersmoothing. The actual coverage probabilities are close to the nominal values for both the outcome function and the overidentification tests. 

\section{Empirical Applications}\label{1.sec7}
In this section, I use two empirical examples to illustrate the value and limitations of my approach. Section \ref{sec6.1} presents an application of the return to education. Section \ref{sec6.2} uses an example of the preschool program choice to illustrate when a matching point does not exist.
\subsection{The Return to Education}\label{sec6.1}
In this application, I use the same extract from the 1979 NLS dataset as in \cite{card1995using} and adopt the same instrument. The instrument equals $1$ if an individual grew up near an accredited four year college, and equals $0$ otherwise. The outcome variable $Y$ is the log wage and is assumed to be determined by a separable model. I use the average of parents' years of schooling as the matching covariate $X$, and $x_{0}$ is set equal to $10,11$ and $12$. Finally, I drop the individuals who were still enrolled in a school at the time of the survey. The remaining sample size is 2000.

I construct $D$ by dividing individuals' (or, children's) years of schooling into either two or three categories. In the case of a binary $D$, both my approach and the standard nonparametric IV approach can identify and estimate the outcome function at each level of $D$. Comparing the results obtained by these two approaches, they yield similar point estimates, but my approach attains smaller variances. When $D$ takes on three values, no existing method would work without imposing additional structures on the outcome function. My estimates are coherent with the empirical literature. 

\subsubsection{A Binary $D$}\label{1.sec7.1}
In this subsection, I assume that the latent selection mechanism only yields two outcomes: $D=1$ if an individual's years of schooling is greater than 12 and $D=0$ otherwise. As $|S(D)|=|S(Z)|=2$, at a fixed $x_{0}$, the outcome function $\bm{m}^{*}(x_{0})\equiv ({m}_{0}^{*}(x_{0}),{m}_{1}^{*}(x_{0}))'$ is just-identified by the standard nonparametric IV approach, and overidentified by my approach if a matching point exists.

For each value of $x_{0}$, I try to find two matching points $x_{m1}$ and $x_{m2}$ such that $(x_{0},0)$ and $(x_{m1},1)$ are a matching pair, and $(x_{0},1)$ and $(x_{m2},0)$ are another. As $D$ only takes on two values, each matching point only needs to match one propensity score. The propensity score is estimated as proposed in Section \ref{sec4.1}. The kernel and the bandwidth follow those used in Section \ref{1.sec6}. 

Figure \ref{1.fig4} illustrates the propensity score matching for $x_{0}=12$. The red solid curves in the left and the right panels are $\hat{p}_{0}(x,1)-\hat{p}_{0}(12,0)$ and $\hat{p}_{0}(x,0)-\hat{p}_{0}(12,1)$ respectively. The intersection points of these propensity score differences with zero are the estimated matching points. The patterns for $x_{0}=10$ and $11$ are similar to Figure \ref{1.fig4} thus omitted. Figure \ref{1.fig4} implies that individuals whose parents have more years of schooling are more likely to attain post-high school education. From the values of the matching points, living close to a four year college and parents' education are substitutes for an individual's educational attainment. At $X=12$, an increase of about half a year in parents' education compensates for living far from a college.

 \begin{figure}[t]
\centering
\includegraphics[width=0.65\linewidth,trim={0cm 1cm 0 0cm},clip]{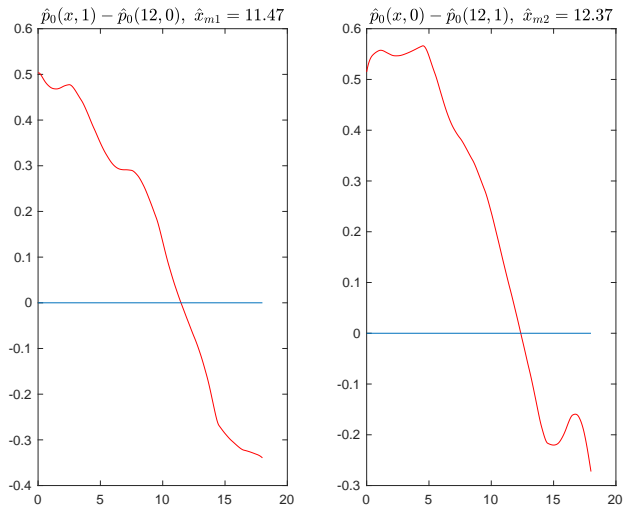}
\caption{Propensity Score Differences: $x_{0}=12$, $|S(D)|=2$}\label{1.fig4}
\end{figure}

Now the outcome function can be estimated using two approaches. The results are shown in Table \ref{1.tab8}. The second row \textit{Matching} indicates whether the matching points are estimated and used. When not using the matching points, $\bm{m}^{*}(x_{0})$ is estimated by the nonparametric IV approach by solving the sample analogue of equation \eqref{1.eq1} with $(x,z)=(x_{0},0)$ and $(x_{0},1)$, and $\Delta_{d}(x_{0},x_{0})=0$. The standard errors in parentheses are computed using the asymptotic variance estimators. The $p$-values of the overidentification test for the outcome function are reported in the last row when applicable.

\begin{table}[h]
\centering
\caption{Estimates of $(m^{*}_{0}(x_{0}),m^{*}_{1}(x_{0}))$}\label{1.tab8}
\begin{tabular}{C{3cm} C{1.5cm} C{1.5cm}C{1.5cm}C{1.5cm}C{1.5cm}C{1.5cm}}
\hline
\hline
& \multicolumn{2}{c}{$x_{0}=10$} & \multicolumn{2}{c}{$x_{0}=11$} & \multicolumn{2}{c}{$x_{0}=12$} \\
\hline
Matching: & \xmark  & \cmark & \xmark & \cmark & \xmark & \cmark\\
\hline

$\hat{m}_{0}(x_{0})$&$5.63$ &$5.64$ &$5.59$&$5.56$& $5.35$ & $5.37$\\
&$(0.28)$ & $(0.17)$ &$(0.33)$&$(0.20)$& $(0.60)$ & $(0.33)$ \\

$\hat{m}_{1}(x_{0})$& $7.15$ & $7.13$ &$6.90$&$6.92$& $6.90$ & $6.89$ \\ 
 &$(0.31)$ & $(0.19)$ &$(0.25)$&$(0.15)$& $(0.32)$ & $(0.18)$\\

\hline
Over-Id $p$ value  & N.A. & $0.98$ & N.A. &$0.99$ &N.A. & $0.80$\\
\hline
\end{tabular}
\end{table}

From Table \ref{1.tab8}, we can make three observations. First, the estimates using the two approaches are very close. It provides evidence that the additional moment conditions brought in by the matching points are valid. The insignificant overidentification tests also suggest that both the instrument and the matching points are valid. Note that the overidentification test is unavailable when using the nonparametric IV approach. Second, the variances are lower using the new approach. Variance reduction is due to the use of more moment conditions. Consequently, the estimated effects can be more significant. For instance, though not reported here, $\hat{m}_{1}(12)-\hat{m}_{0}(12)$ is significant at 10\% level using the IV approach but is significant at 1\% level using my approach. Third, the return is increasing in the level of own education and heterogeneous in parents' education.
\subsubsection{A Three-Valued $D$}\label{1.sec7.2}
Now assume the selection model yields three outcomes. The baseline level of $D$ is still at most high school but recoded by $D=1$. Post-high school education is further divided into two groups: $D=2$ if $12<\text{years of schooling}\leq 15$ (some college), and $D=3$ if $\text{years of schooling}> 15$ (college and above). In this case, no existing method can identify and estimate $\boldsymbol{m}^{*}(x_{0})$ without imposing additional assumptions on it. 

Figure \ref{1.fig5} illustrates the matching points for $x_{0}=12$. Again, the plots for the other values of $x_{0}$ are omitted as the patterns are similar. The solid red curves are $\hat{p}_{1}(x,1)-\hat{p}_{1}(x_{0},0)$ and $\hat{p}_{1}(x,0)-\hat{p}_{1}(x_{0},1)$, while the dashed blue curves are $\hat{p}_{3}(x,1)-\hat{p}_{3}(x_{0},0)$ and $\hat{p}_{3}(x,0)-\hat{p}_{3}(x_{0},1)$. In theory, matching is successful if the solid and the dashed curves intersect with the horizontal line of zero at the same point. From the figure, the intersection points are indeed very close in both panels. The overidentification tests also support successful matching; $\mathcal{J}_{x1}$ and $\mathcal{J}_{x2}$ reported on top of the plots are insignificant in both cases. Finally, since the baseline level here (years of schooling $\leq 12$) is defined in the same way as in the case of a binary $D$, its propensity score is also equal to the previous case. Since this propensity score has to be matched in both cases, the matching points in these cases should be identical. Here the estimates are 11.54 and 12.34, indeed very close to those when $D$ is binary (11.47 and 12.37).
 \begin{figure}[t]
\centering
\includegraphics[width=0.65\linewidth,trim={0cm 1cm 0 0cm},clip]{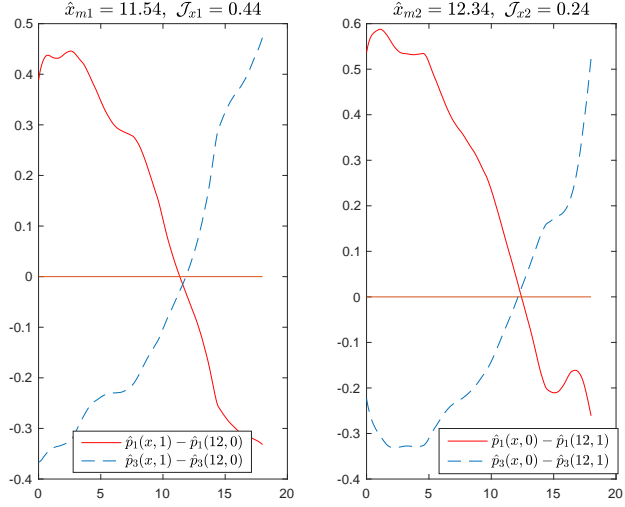}
\caption{Generalized Propensity Score Differences: $x_{0}=12$, $|S(D)|=3$}\label{1.fig5}
\end{figure}

Next, let us turn to $\hat{\bm{m}}(x_{0})$ shown in Table \ref{1.tab9}. In this case both the outcome function and the matching points are overidentified. The $p$-value for each overidentification test is presented in the bottom panel. First, we can see that none of the overidentification tests for $\bm{m}^{*}(x_{0})$ is significant at any reasonable level, similar to Table \ref{1.tab8} for the binary case. Meanwhile, the joint overidentification tests for the matching points are also insignificant, confirming that the single covariate matches all the generalized propensity scores. Second, the return to education is monotonic in the level of own education and heterogeneous in parents' years of schooling, while the difference in returns across adjacent own education levels is decreasing.
\begin{table}[H]
\centering
\caption{Estimates of $(m^{*}_{1}(x_{0}),m^{*}_{2}(x_{0}),m^{*}_{3}(x_{0}))$}\label{1.tab9}
\begin{tabular}{C{3cm}C{2.5cm} C{2.5cm} C{2.5cm}}
\hline
\hline
& $x_{0}=10$& $x_{0}=11$& $x_{0}=12$\\
\hline
$\hat{m}_{1}(x_{0})$  & $5.62$ &$5.56$& $5.33$ \\
& $(0.23)$ &$(0.22)$& $(0.38) $\\
$\hat{m}_{2}(x_{0})$  & $7.03$ &$6.47$ &$6.39$ \\
& $(3.18)$& $(1.39)$& $(1.45)$ \\
$\hat{m}_{3}(x_{0})$ & $7.28$ &$7.32$& $7.31$\\
 & $(2.72)$ & $(1.09)$&$(1.00)$\\
\hline
\multicolumn{4}{c}{Over-Id $p$-value}\\
$\bm{m}^{*}(x_{0})$ & $0.89$ & $0.87$& $0.58$ \\
$(x_{m1},x_{m2})$ & $0.54$ & $0.36$& $0.41$\\
\hline
\end{tabular}
\end{table}

\subsection{When Does the Matching Fail?}\label{sec6.2}
As shown in Lemma \ref{1.thmMEQ}, a matching pair necessarily matches all the generalized propensity scores. Matching may fail if the instrument has dominant effects on the generalized propensity scores such that shifting the covariates cannot compensate for those effects. As my identification strategy treats local variation in the covariates like instruments, one would hope that both such variation and the original instrument have comparably strong effects on selection. When the former is too weak, it loses identification power and this paper's approach would fail.

For illustration, I consider an application on preschool program selection, following \cite{kline2016evaluating} using the Head Start Impact Study (HSIS) dataset. The endogenous variable $D$ takes on three values: participating in Head Start ($h$), participating in another competing preschool program ($c$), and not participating in any preschool programs ($n$). The binary instrument $Z$ is a random lottery granting access to Head Start. Candidates for the covariate $X$ are family income, baseline test score and the centers' quality index.

 \begin{figure}[t]
\centering
\includegraphics[width=0.65\linewidth,trim={0cm 1cm 0 0cm},clip]{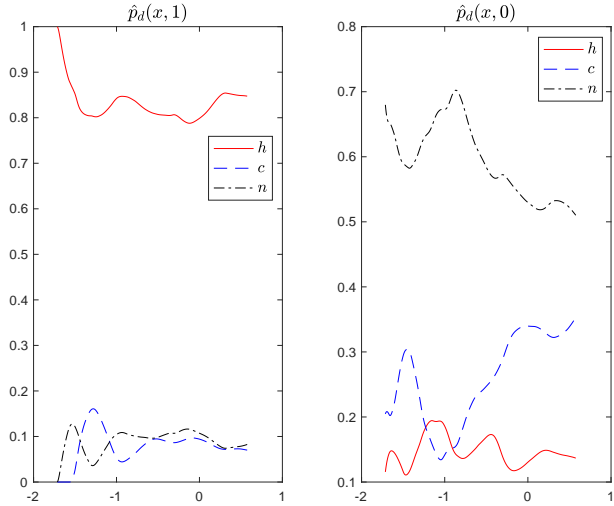}
\caption{Generalized Propensity Scores of Preschool Program Choices}\label{1.fig7}
\end{figure}

Figure \ref{1.fig7} shows the estimated generalized propensity scores using the baseline test score as $X$ and $x_{0}$ is equal to the sample median. Findings under other values of this covariate or using other covariates are similar. We see that if an individual wins the lottery, the probability of attending Head Start is very high, and not much affected by the baseline test score. On the contrary, when not winning the lottery, the individual would most likely not participate in any program, and in particular, the probability of attending Head Start is lowest for almost any baseline test score. A matching point does not exist in this example because shifting $X$ never offsets the dominant effects of $Z$ on the generalized propensity scores.

\section{Relation to the Existing Literature}\label{1.sec8}
\subsection{Triangular Models}\label{1.sec8.1}
Techniques that achieve point identification of triangular models often require the endogenous variable $D$ to be continuous. Different approaches are developed depending on whether $Z$ is also continuous or discrete.

Continuous $D$ and $Z$. The widely used control function approach usually needs a continuous instrument (e.g. \cite{newey1999nonparametric}, \cite{chesher2003identification}, \cite{florens2008identification}, \cite{imbens2009identification}, etc.). This approach allows the outcome heterogeneity to be multidimensional, but the first stage heterogeneity needs to be a scalar with the selection function $h$ strictly increasing in it. 

A continuous $D$ and a binary $Z$. \cite{d2015identification} and \cite{torgovitsky2015identification} show that identification of a nonseparable outcome function increasing in the scalar disturbance can be achieved with a binary $Z$. \cite{gunsilius2018point} extends the model to allow for multidimensional heterogeneity. Compared with my approach, I need to use covariates as an additional source of variation, but I allow for a discrete $D$ and more flexible selection heterogeneity.

Multidimensional $D$ (with a continuous component) and a binary $Z$. \cite{huang2019identification} consider a triangular model with a separable outcome function where there are two endogenous variables and a single instrument that can be binary. In this respect, their paper focuses on a similar problem to mine since in my setup, a discrete $D$ is equivalent to multiple dummy endogenous variables while there is only one binary instrument.  However, like the other papers discussed, they also need one of the endogenous variables to be continuous at least on a subset of its support, with a first stage where the unobservable must be a scalar.

\subsection{Single Equation Approaches}\label{1.sec8.2}
Single equation approaches refer to methods that achieve identification without explicitly relying on a selection model. The typical IV approach for nonparametric identification falls into this category (e.g. \cite{newey2003instrumental}, \cite{das2005instrumental}, \cite{chernozhukov2005iv}, \cite{chernozhukov2007instrumental}, \cite{chen2014local}, ect.). This approach requires $Z$ to have large support. 

\cite{caetano2016identifying} develop a strategy that achieves identification using small-support $Z$ when $D$ is multivalued. Their method does not rely on selection models. Similar to this paper, they utilize the covariates for identification, but the covariates need to be structurally separable from the model. Taking the nonseparable model with a discrete $D$ as an example, they essentially impose a single index structure: $Y=g^{*}_{D}(U)$ and $U=\phi(\boldsymbol{X},U_{0})$, where both $U$ and $U_{0}$ are unobserved. The function $\phi$ is real valued and $g^{*}_{D}(\cdot)$ is strictly increasing.
In contrast, I allow all the covariates to enter the model in arbitrary ways, but a selection model, though very general, is needed. These two approaches are complementary.

\section{Concluding Remarks}\label{1.sec9}
In this paper, I develop a novel approach to use covariates to identify structural outcome functions in a triangular model when the discrete endogenous variable takes on more values than the instrument. This paper illustrates that information on endogenous selection has large identifying power. The generalized propensity scores can provide useful information on the degree of endogeneity indexed by covariates-instrument combination. Across such combinations at which the endogenous variable has the same degree of endogeneity, extrapolation can be made to supplement the insufficient information from the instrument and facilitate identification.

Moving forward, it would be of interest to apply the idea in this paper to other scenarios, such as models with limited dependent variables, extrapolation in regression discontinuity designs, etc. Another direction is to generalize the outcome function by allowing for multidimensional heterogeneity.
\begin{appendices}


\section{The General Case of $|S(D)|>|S(Z)|$ and Multiple Endogenous Variables}\label{1.appA}
Consider the general case of arbitrary $|S(D)|>|S(Z)|$. For a given $\boldsymbol{x}_{0}$, at least $|S(D)|-|S(Z)|$ m-connected points are needed for identification. The size difference is not as formidable as it appears: When $|S(Z)|$ increases, the number of the m-connected points may increase faster. Recall in Figure \ref{1.fig1}, each value of $Z$ induces a branch for m-connected points to grow. Then, for instance, if $|S(Z)|=3$, two matching points of $\bm{x}_{0}$ for each $z\in S(Z)$ may be obtained by solving the following equations:
\begin{linenomath*}\begin{equation*}
\boldsymbol{p}(\boldsymbol{x},z')=\boldsymbol{p}(\boldsymbol{x}_{0},z)\text{ and }\boldsymbol{p}(\boldsymbol{x},z'')=\boldsymbol{p}(\boldsymbol{x}_{0},z)
\end{equation*}\end{linenomath*}
for $z',z''\in S(Z)$. Now that $z$ takes on 3 values, up to 6 matching points may be obtained even if each propensity score matching equation has only one solution. By induction, the number of the matching points can be as many as $|S(Z)|\cdot (|S(Z)|-1)$. With the variation in $Z$ itself, a discrete instrument taking on $|S(Z)|$ values may identify a nonparametric model with $|S(D)|=|S(Z)|^{2}$, instead of $|S(D)|=|S(Z)|$ when using the standard IV approach. Note this is only the number of the matching points. For the m-connected points, the number can be even larger. 
%

The approach in this paper can also be applied to the case of multiple discrete endogenous variables. Suppose there are $M$ discrete endogenous variables $D_{1},...,D_{M}$ in a model. It is equivalent to recode them as one single endogenous variable $D_{0}$. For instance, if $S(D_{1},...,D_{M})=S(D_{1})\times\cdots\times S(D_{M})$, then let $S(D_{0})=\{1,2,...,\Pi_{m=1}^{M}|S(D_{m})|\}$. The matching points may still be found by matching the generalized propensity scores. The following example illustrates it.
\begin{example}[Two Endogenous Variables]\label{egA1}
Let $D_{1},D_{2}\in\{0,1\}$ be two endogenous variables. Suppose they are determined by the following model:
\begin{linenomath*}\begin{align*}
D_{1}&=\mathbbm{1}(\gamma_{1}(\bm{X},Z)\leq V_{1})\\
D_{2}&=\mathbbm{1}(\gamma_{2}(\bm{X},Z)\leq V_{2})
\end{align*}\end{linenomath*}
where the vector of the unobservables $(V_{1},V_{2})$ is continuously distributed on $\mathbb{R}^{2}$. Assume $(\bm{X},Z)\independent (V_{1},V_{2})$. Let $D_{0}=1,2,3,4$, corresponding to $(D_{1},D_{2})=(0,0),(0,1),\newline (1,0),(1,1)$ respectively. Then the selection model can be rewritten as:
\begin{linenomath*}\begin{align*}
h_{1}(\bm{X},Z,V_{1},V_{2})&=\big(1-\mathbbm{1}(\gamma_{1}(\bm{X},Z)\leq V_{1})\big)\cdot \big(1-\mathbbm{1}(\gamma_{2}(\bm{X},Z)\leq V_{2})\big)\\
h_{2}(\bm{X},Z,V_{1},V_{2})&=\big(1-\mathbbm{1}(\gamma_{1}(\bm{X},Z)\leq V_{1})\big)\cdot \mathbbm{1}(\gamma_{2}(\bm{X},Z)\leq V_{2})\\
h_{3}(\bm{X},Z,V_{1},V_{2})&=\mathbbm{1}(\gamma_{1}(\bm{X},Z)\leq V_{1})\cdot \big(1-\mathbbm{1}(\gamma_{2}(\bm{X},Z)\leq V_{2})\big)\\
h_{4}(\bm{X},Z,V_{1},V_{2})&=\mathbbm{1}(\gamma_{1}(\bm{X},Z)\leq V_{1})\cdot \mathbbm{1}(\gamma_{2}(\bm{X},Z)\leq V_{2})
\end{align*}\end{linenomath*}

Now we can verify that whenever $\boldsymbol{p}(\bm{x}_{0},z)=\boldsymbol{p}(\bm{x}_{m},z')$, we have $h_{d}(\bm{x}_{0},z,\cdot,\cdot)=h_{d}(\bm{x}_{m},z',\cdot,\cdot)$ for all $d\in\{1,2,3,4\}$. First, by $p_{1}(\bm{x}_{0},z)=p_{1}(\bm{x}_{m},z')$, we have
\begin{linenomath*}\begin{equation*}
\mathbb{P}\big(V_{1}< \gamma_{1}(\bm{x}_{0},z),V_{2}<\gamma_{2}(\bm{x}_{0},z)\big)=\mathbb{P}\big(V_{1}< \gamma_{1}(\bm{x}_{m},z'),V_{2}<\gamma_{2}(\bm{x}_{m},z')\big)
\end{equation*}\end{linenomath*}
Suppose $\gamma_{1}(\bm{x}_{0},z)\neq \gamma_{1}(\bm{x}_{m},z')$. Without loss of generality, let $ \gamma_{1}(\bm{x}_{0},z)< \gamma_{1}(\bm{x}_{m},z')$. Then $\gamma_{2}(\bm{x}_{0},z)>\gamma_{2}(\bm{x}_{m},z')$ to make the above equation hold. Consequently,
\begin{linenomath*}\begin{equation*}
\mathbb{P}\big(V_{1}\geq  \gamma_{1}(\bm{x}_{0},z),V_{2}<\gamma_{2}(\bm{x}_{0},z)\big)>\mathbb{P}\big(V_{1}\geq  \gamma_{1}(\bm{x}_{m},z'),V_{2}<\gamma_{2}(\bm{x}_{m},z')\big).
\end{equation*}\end{linenomath*}
But this implies $p_{3}(\bm{x}_{0},z)>p_{3}(\bm{x}_{m},z')$, a contradiction. Using a similar argument, it can be shown that $\gamma_{2}(\bm{x}_{0},z)= \gamma_{2}(\bm{x}_{m},z')$ as well. 
\end{example}

\section{More  Examples}\label{1.appC}
More examples are presented where $\bm{p}(\bm{x}_{m},z')=\bm{p}(\bm{x}_{0},z)$ implies $h_{d}(\bm{x}_{m},z',\cdot)=h_{d}(\bm{x}_{0},z,\cdot)$ for all $d\in S(D)$.
\begin{example}[Ordered Choice with Stochastic Thresholds]\label{1.exOC}
This example generalizes the ordered choice model in Example \ref{1.exOC3}. Let $h_{1}(\boldsymbol{X},Z,\bm{V})=\mathbbm{1}(V_{1}\leq \gamma_{1}(\boldsymbol{X},Z))$, $h_{3}(\boldsymbol{X},Z,\bm{V})=\mathbbm{1}(V_{2}> \gamma_{2}(\boldsymbol{X},Z))$, and $h_{2}=1-h_{1}-h_{3}$ where $(V_{1},V_{2})$ are continuously distributed on $\mathbb{R}^{2}$ with $V_{1}<V_{2}$ a.s. Also, assume $\gamma_{1}(\boldsymbol{X},Z)<\gamma_{2}(\boldsymbol{X},Z)$ a.s. This model nests parametric ordered choice models and some nonparametric models, for instance the general ordered choice model in \cite{cunha2007identification}. 

Under $(\bm{X},Z)\independent (V_{1},V_{2})$, we have $p_{1}(\bm{x},z)=F_{V_{1}}(\gamma_{1}(\bm{x},z))$ and $p_{3}(\bm{x},z)=1-F_{V_{2}}(\gamma_{2}(\bm{x},z))$ for any $(\bm{x},z)\in S(\bm{X},Z)$. By strict monotonicity of $F_{V_{1}}$ and $F_{V_{2}}$, $p_{1}$ and $\gamma_{1}$ as well as $p_{2}$ and $\gamma_{2}$ are one-to-one. Hence, whenever $\bm{p}(\boldsymbol{x}_{0},z)=\bm{p}(\boldsymbol{x}_{m},z')$, $\gamma_{1}(\boldsymbol{x}_{0},z)=\gamma_{1}(\boldsymbol{x}_{m},z')$ and $\gamma_{2}(\boldsymbol{x}_{0},z)=\gamma_{2}(\boldsymbol{x}_{m},z')$. 
\end{example}

\begin{example}[Multinomial Choice]\label{1.exMC}
Consider a nonparametric multinomial choice model. Versions of it are considered in \cite{matzkin1993nonparametric}, \cite{heckman2008instrumental}, \cite{lee2018identifying} etc. Let $R_{d}(\boldsymbol{X},Z)+\tilde{V}_{d}$ be the indirect utility of alternative $d$ where $\tilde{V}_{d}$ is an unobserved continuous random variable. Alternative $d$ is selected if $R_{d}(\boldsymbol{X},Z)+\tilde{V}_{d}>R_{-d}(\boldsymbol{X},Z)+\tilde{V}_{-d}$ where the subscript $-d$ refers to any selection other than $d$. Reparameterize the model by letting $V_{1}=\tilde{V}_{2}-\tilde{V}_{1}$, $V_{2}=\tilde{V}_{3}-\tilde{V}_{1}$, $V_{3}=\tilde{V}_{3}-\tilde{V}_{2}$, $\gamma_{1}(\boldsymbol{X},Z)=R_{1}(\boldsymbol{X},Z)-R_{2}(\boldsymbol{X},Z)$ and $\gamma_{2}(\boldsymbol{X},Z)=R_{1}(\boldsymbol{X},Z)-R_{3}(\boldsymbol{X},Z)$. The model can be rewritten as
\begin{linenomath*}\begin{align*}
& D=1 \iff V_{1}<\gamma_{1}(\boldsymbol{X},Z),V_{2}<\gamma_{2}(\boldsymbol{X},Z)\\
& D=2 \iff V_{1}>\gamma_{1}(\boldsymbol{X},Z),V_{3}<\gamma_{2}(\boldsymbol{X},Z)-\gamma_{1}(\boldsymbol{X},Z)\\
& D=3 \iff V_{2}>\gamma_{2}(\boldsymbol{X},Z),V_{3}>\gamma_{2}(\boldsymbol{X},Z)-\gamma_{1}(\boldsymbol{X},Z)
\end{align*}\end{linenomath*}
If $0<p_{d}(\boldsymbol{x}_{0},z)<1$ for all $d$, it can be verified using a similar argument as in Example \ref{egA1} that $\bm{p}(\bm{x}_{m},z')=\bm{p}(\bm{x}_{0},z)$ implies $\gamma_{1}(\bm{x}_{m},z')=\gamma_{1}(\bm{x}_{0},z)$ and $\gamma_{2}(\bm{x}_{m},z')=\gamma_{2}(\bm{x}_{0},z)$.
\end{example}

\section{Proofs of Results in Sections \ref{1.sec2} and \ref{1.sec3}}\label{1.appB}
%

\begin{proof}[Proof of Lemma \ref{1.thmMEQ}]
To save space, I only show equation \eqref{1.eq5} as equation \eqref{eq2.4a} follows a similar argument. For any $d\in S(D)$,
\begin{linenomath*}\begin{align*}
F_{U_{d}|D\boldsymbol{X}Z}(u|d,\boldsymbol{x}_{0},z)=&F_{U_{d}|\boldsymbol{V}\boldsymbol{X}Z}(u|h_{d}(\boldsymbol{x}_{0},z,\boldsymbol{V})=1,\boldsymbol{x}_{0},z)\\
=&\frac{\int_{h_{d}(\boldsymbol{x}_{0},z,\boldsymbol{v})=1}F_{U_{d}|\boldsymbol{V}\boldsymbol{X}Z}(u|\boldsymbol{v},\boldsymbol{x}_{0},z)d\mathbb{P}(\boldsymbol{V}=\boldsymbol{v}|\boldsymbol{X}=\boldsymbol{x}_{0},Z=z)}{\int_{h_{d}(\boldsymbol{x}_{0},z,\boldsymbol{v})=1}d\mathbb{P}(\boldsymbol{V}=\boldsymbol{v}|\boldsymbol{X}=\boldsymbol{x}_{0},Z=z)}\\
=&\frac{\int_{h_{d}(\boldsymbol{x}_{m},z',\boldsymbol{v})=1}F_{U_{d}|\boldsymbol{V}\boldsymbol{X}Z}(u|\boldsymbol{v},\boldsymbol{x}_{m},z')d\mathbb{P}(\boldsymbol{V}=\boldsymbol{v}|\boldsymbol{X}=\boldsymbol{x}_{m},Z=z')}{\int_{h_{d}(\boldsymbol{x}_{m},z',\boldsymbol{v})=1}d\mathbb{P}(\boldsymbol{V}=\boldsymbol{v}|\boldsymbol{X}=\boldsymbol{x}_{m},Z=z')}\\
=&F_{U_{d}|\boldsymbol{V}\boldsymbol{X}Z}(u|h_{d}(\boldsymbol{x}_{m},z',\boldsymbol{V})=1,\boldsymbol{x}_{m},z')\\
=&F_{U_{d}|D\boldsymbol{X}Z}(u|d,\boldsymbol{x}_{m},z')
\end{align*}\end{linenomath*}
where the first and the last equalities follow equation \eqref{1.sl}, the second and the fourth equalities are by definition, and the third equality is by equations \eqref{1.eq3} and \eqref{1.eq2.5}.
\end{proof}
\begin{proof}[Proof of Propositions \ref{1.prop2} and \ref{1.prop1}]
For Proposition \ref{1.prop2}, by $g_{d}^{*}(\bm{x},u)=\varphi_{d}(g_{d}^{*}(\bm{x}_{0},u);\bm{x})$ and $\bm{x}\in\mathcal{X}_{MC}(\bm{x}_{0})$, Assumption \ref{1.assCM} implies that $F_{Y|D\bm{X}Z}(\varphi_{d}(g_{d}^{*}(\bm{x}_{0},u);\bm{x})|d,\bm{x},z)=F_{U_{d}|D\bm{X}Z}(u|d,\bm{x},z)$ for all $d\in S(D)$. By Assumptions \ref{1.assENSP} and \ref{1.assRSS}, for $d\neq d'\in S(D)$, $F_{U_{d}|\boldsymbol{V}\boldsymbol{X}Z}(u|\boldsymbol{v},\boldsymbol{x},z)=F_{U_{d'}|\boldsymbol{V}\boldsymbol{X}Z}(u|\boldsymbol{v},\boldsymbol{x},z)$. Therefore,
\begin{linenomath*}\begin{align*}
F_{U_{d}|D\boldsymbol{X}Z}(u|d,\boldsymbol{x},z)=&F_{U_{d}|\boldsymbol{V}\boldsymbol{X}Z}(u|h_{d}(\boldsymbol{x},z,\boldsymbol{V})=1,\boldsymbol{x},z)\\
=&\frac{\int_{h_{d}(\boldsymbol{x},z,\boldsymbol{v})=1}F_{U_{d}|\boldsymbol{V}\boldsymbol{X}Z}(u|\boldsymbol{v},\boldsymbol{x},z)d\mathbb{P}(\boldsymbol{V}=\boldsymbol{v}|\boldsymbol{X}=\boldsymbol{x},Z=z)}{\int_{h_{d}(\boldsymbol{x},z,\boldsymbol{v})=1}d\mathbb{P}(\boldsymbol{V}=\boldsymbol{v}|\boldsymbol{X}=\boldsymbol{x},Z=z)}\\
=&\frac{\int_{h_{d}(\boldsymbol{x},z,\boldsymbol{v})=1}F_{U_{d'}|\boldsymbol{V}\boldsymbol{X}Z}(u|\boldsymbol{v},\boldsymbol{x},z)d\mathbb{P}(\boldsymbol{V}=\boldsymbol{v}|\boldsymbol{X}=\boldsymbol{x},Z=z)}{\int_{h_{d}(\boldsymbol{x},z,\boldsymbol{v})=1}d\mathbb{P}(\boldsymbol{V}=\boldsymbol{v}|\boldsymbol{X}=\boldsymbol{x},Z=z)}\\
=&F_{U_{d'}|\boldsymbol{V}\boldsymbol{X}Z}(u|h_{d}(\boldsymbol{x},z,\boldsymbol{V})=1,\boldsymbol{x},z)\\
=&F_{U_{d'}|D\boldsymbol{X}Z}(u|d,\boldsymbol{x},z)
\end{align*}\end{linenomath*}
Hence, the left hand side of equation \eqref{1.eq2} is equal to $F_{U_{d}|\bm{X}Z}(u|\bm{x},z)$, which is equal to $u$ under Assumption \ref{1.assENSP}.

For Proposition \ref{1.prop1}, by $m_{d}^{*}(\bm{x})=m_{d}^{*}(\bm{x}_{0})+\Delta_{d}(\bm{x}_{0},\bm{x})$, equation \eqref{1.eq1} is equivalent as $\sum_{d\in S(D)}p_{d}(\boldsymbol{x},z)\cdot m_{d}^{*}(\boldsymbol{x})=\sum_{d\in S(D)}p_{d}(\boldsymbol{x},z)\cdot \mathbb{E}_{Y|D\boldsymbol{X}Z}(d,\boldsymbol{x},z)$, which holds if $\sum_{d\in S(D)}p_{d}(\boldsymbol{x},z)\cdot\mathbb{E}_{U_{d}|D\bm{X}Z}(d,\bm{x},z)=0$. Under Assumption \ref{1.assMS}, it can be verified that $\mathbb{E}_{U_{d}|D\bm{X}Z}(d,\bm{x},z)=\mathbb{E}_{U_{d'}|D\bm{X}Z}(d,\bm{x},z)$ for any $d\neq d'\in S(D)$ following a similar argument as in the proof of Proposition \ref{1.prop2}. Therefore, $\sum_{d\in S(D)}p_{d}(\boldsymbol{x},z)\cdot\mathbb{E}_{U_{d}|D\bm{X}Z}(d,\bm{x},z)=\mathbb{E}_{U_{d}|\bm{X}Z}(\bm{x},z)=0$ by Assumption \ref{1.assESP}.
\end{proof}

\begin{proof}[Proof of Lemma \ref{1.lemUnq}]

Suppose $\boldsymbol{y}^{*}$ is not the unique increasing solution path passing through $(u^{*},\bm{c})$. Let $\tilde{\boldsymbol{y}}$ be another increasing solution path such that $\tilde{\bm{y}}(u^{*})=\bm{c}$. Then there must exist $u_{1}\leq u^{*}\leq u_{2}$ such that $u_{1},u_{2}\in\mathcal{U}$ and $\tilde{\bm{y}}(u)=\bm{y}^{*}(u)$ for all $u\in [u_{1},u_{2}]$. Note that if $u_{1}=u_{2}$, $u=u^{*}$. Let $\bar{u}$ and $\underline{u}$ be the supremum and the infimum of the sets of such $u_{2}$ and $u_{1}$. Since $\mathcal{U}$ is compact, we have $\underline{u},\bar{u}\in\mathcal{U}$. Then there are the following two cases.

Case 1. $\tilde{\boldsymbol{y}}(\cdot)$ is continuous at both $\bar{u}$ and $\underline{u}$. By continuity of $\bm{y}^{*}$ and $\tilde{\bm{y}}$, $\bm{y}^{*}(\bar{u})=\tilde{\bm{y}}(\bar{u})$ and $\bm{y}^{*}(\underline{u})=\tilde{\bm{y}}(\underline{u})$. Therefore, if $[\underline{u},\bar{u}]=\mathcal{U}$, we are done. Otherwise, without loss of generality, suppose $\bar{u}$ is in the interior of $\mathcal{U}$. Since $\tilde{\boldsymbol{y}}$ is monotonic, it has at most countable discontinuities. Thus, there exists a $\bar{u}'>\bar{u}$ where $\bar{u}'$ is also in the interior of $\mathcal{U}$ such that $\tilde{\boldsymbol{y}}$ is continuous on $[\bar{u},\bar{u}')$ and $\tilde{\boldsymbol{y}}(u)\neq \boldsymbol{y}^{*}(u)$ for all $u\in (\bar{u},\bar{u}')$. 

Since the Jacobian matrix $\nabla\bm{M}$ is full rank and continuous at $(\bm{y}^{*}(\bar{u}),\bar{u})$, there exists a neighborhood of $(\boldsymbol{y}^{*}(\bar{u}),\bar{u})$, $\mathcal{N}$, on which $\bm{M}(\cdot,\cdot)$ is one-to-one. Since $\tilde{\bm{y}}$ is continuous on $[\bar{u},\bar{u}')$ and $\tilde{\bm{y}}(\bar{u})=\bm{y}^{*}(\bar{u})$, there must exist $u''\in (\bar{u},,\bar{u}')$ such that $(\tilde{\bm{y}}(u''),u'')\in\mathcal{N}$. Now $\tilde{\boldsymbol{y}}(u'')\neq \boldsymbol{y}^{*}(u'')$ but $\bm{M}(\tilde{\boldsymbol{y}}(u''),u'')=\bm{M}(\boldsymbol{y}^{*}(u''),u'')$, a contradiction. A similar argument can also be found in \cite{ortega1970iterative}\label{1.fn2}, pp. 133-134, \cite{ambrosetti1995primer}, pp. 48-49, and \cite{de2014global}, as an intermediate step to show variants of the Hadamard Theorem.


Case 2. $\tilde{\boldsymbol{y}}(\cdot)$ is discontinuous at $\bar{u}$ or $\underline{u}$. Again, without loss of generality, suppose discontinuity is at $\bar{u}$. Since $\tilde{\bm{y}}$ is increasing, it must be the case that $\lim_{u\searrow \bar{u}}\tilde{\bm{y}}(u)>\lim_{u\nearrow \bar{u}}\tilde{\bm{y}}(u)$, i.e., there is at least one component in $\tilde{\bm{y}}$ jumps up at $\bar{u}$. Since $\bm{M}(\cdot,\cdot)$ is continuous and $\bm{M}(\tilde{\bm{y}}(u),u)=\bm{0}$ for all $u\in\mathcal{U}$, 
\begin{linenomath*}\begin{align*}
\bm{0}=\lim_{u\searrow \bar{u}}\bm{M}(\tilde{\bm{y}}(u),u)=\bm{M}(\lim_{u\searrow \bar{u}}\tilde{\bm{y}}(u),\bar{u})>\bm{M}(\lim_{u\nearrow \bar{u}}\tilde{\bm{y}}(u),\bar{u})=\lim_{u\nearrow \bar{u}}\bm{M}(\tilde{\bm{y}}(u),u)=\bm{0}
\end{align*}\end{linenomath*}
where the inequality holds because $\bm{M}(\cdot,\bar{u})$ is strictly increasing on $\mathcal{Y}$. A contradiction.

Therefore, $\boldsymbol{y}^{*}$ is the unique increasing solution path passing through $(u^{*},\bm{c})$.
\end{proof}

\begin{proof}[Proof of Theorem \ref{1.thmIDNSP}]
Step 1. Uniqueness in $\mathcal{G}^{*}$. Recall that $\mathcal{G}^{*}\subseteq\mathcal{G}$ contains increasing functions whose ranges are contained in $\prod_{d=1}^{3} S(Y|d,\boldsymbol{x}_{0})$. Under Assumptions \ref{1.assCM} and \ref{1.assENSP}, $\Psi(\cdot;\tilde{\bm{z}}_{1},\tilde{\bm{z}}_{2},\tilde{\bm{z}}_{3})$ is strictly increasing and continuously differentiable on $\prod_{d=1}^{3} S(Y|d,\boldsymbol{x}_{0})$. Let the upper and the lower boundaries of $S(Y|d,\bm{x}_{0})$ be $\bar{y}_{d\bm{x}_{0}}$ and $\underline{y}_{d\bm{x}_{0}}$. For any candidate solution path $\bm{g}(\cdot)\in\mathcal{G}^{*}$, it must satisfy $\bm{g}(0)=(\underline{y}_{d\bm{x}_{0}})_{d}$ and $\bm{g}(1)=(\bar{y}_{d\bm{x}_{0}})_{d}$ by construction of $\Psi(\cdot;\tilde{\bm{z}}_{1},\tilde{\bm{z}}_{2},\tilde{\bm{z}}_{3})$. Therefore, conditions in Lemma \ref{1.lemUnq} are satisfied, implying that $\bm{g}^{*}(\bm{x}_{0},\cdot)$ is the unique solution path in $\mathcal{G}^{*}$. 

Step 2. Uniqueness in $\mathcal{G}$. Suppose there exists another solution path $\breve{\boldsymbol{g}}\in\mathcal{G}$. For each $d\in S(D)$, let 
\begin{linenomath*}\begin{equation*}
g^{\dagger}_{d}(u)=\left\{\begin{matrix}
\underline{y}_{d\boldsymbol{x}_{0}}, & \text{if } \breve{g}_{d}(u)<\underline{y}_{d\boldsymbol{x}_{0}}\\
\breve{g}_{d}(u), & \text{if }  \breve{g}_{d}(u)\in S(Y|d,\boldsymbol{x}_{0}) \\
\bar{y}_{d\boldsymbol{x}_{0}}, & \text{if } \breve{g}_{d}(u)>\bar{y}_{d\boldsymbol{x}_{0}}\\
\end{matrix}\right.
\end{equation*}\end{linenomath*}
Clearly, $\boldsymbol{g}^{\dagger}\in \mathcal{G}^{*}$ and is also a solution path. By the uniqueness of $\bm{g}^{*}(\bm{x}_{0},\cdot)$ in $\mathcal{G}^{*}$, $\boldsymbol{g}^{\dagger}(\cdot)=\boldsymbol{g}^{*}(\bm{x}_{0},\cdot)$. Since $\bm{g}^{*}(\bm{x}_{0},\cdot)$ is strictly increasing, $\underline{y}_{d\bm{x}_{0}}=g^{\dagger}_{d}(0)<g^{\dagger}_{d}(u)<g^{\dagger}_{d}(1)=\bar{y}_{d\bm{x}_{0}}$ for all $u\in (0,1)$ and $d\in S(D)$. Therefore, $\breve{\boldsymbol{g}}(u)$ necessarily equals $\boldsymbol{g}^{*}(\boldsymbol{x}_{0},u)$ for all $u\in (0,1)$. For $u=0$ or $1$, $g_{d}^{\dagger}(u)=\underline{y}_{d\bm{x}_{0}}$ or $\bar{y}_{d\bm{x}_{0}}$ for all $d\in S(D)$, so $\breve{g}_{d}(u)$ can take any value smaller than $\underline{y}_{d\boldsymbol{x}_{0}}$ or greater than $\bar{y}_{d\bm{x}_{0}}$ respectively.
\end{proof}

\begin{proof}[Proof of Corollary \ref{1.corIDNSP}]
Denote $\mathcal{G}_{0}^{-}= \{\boldsymbol{g}\in\mathcal{G}_{0}:\sup_{u\in \mathcal{U}_{0}}|\boldsymbol{g}(u)-\boldsymbol{g}^{*}(\bm{x}_{0},u)|\geq \delta \}$. Suppose inequality \eqref{1.eq42} does not hold. Then there exists a sequence $\{\boldsymbol{g}_{k}\}\subseteq\mathcal{G}_{0}^{-}$ such that
\begin{linenomath*}\begin{equation*}
\lim_{k\to\infty} \big(\int_{0}^{1}Q_{NSP}(\boldsymbol{g}_{k}(u),u)du-\int_{0}^{1}Q_{NSP}(\boldsymbol{g}^{*}(\bm{x}_{0},u),u)du\big)=0
\end{equation*}\end{linenomath*}

As functions in the sequence $\{\boldsymbol{g}_{k}\}$ are uniformly bounded monotonic functions on a compact interval, there exists a pointwise convergent subsequence $\tilde{\boldsymbol{g}}_{k_{l}}$ by Helly's Selection Theorem. Denote its pointwise limit by $\tilde{\boldsymbol{g}}$. Since pointwise convergence preserves monotonicity and the functions in $\{\boldsymbol{g}_{k_{l}}\}$ are uniformly bounded in $S(Y|d)$, $\tilde{\bm{g}}\in\mathcal{G}_{0}$. Note the equation above also holds for the subsequence $\tilde{\boldsymbol{g}}_{k_{l}}$. Then by the Dominated Convergence Theorem, we can change the order of the limit and the integral operators:
\begin{linenomath*}\begin{align*}
\int_{0}^{1}\lim_{k_{l}\to\infty}Q_{NSP}(\boldsymbol{g}_{k_{l}}(u),u)du&=\int_{0}^{1}Q_{NSP}(\boldsymbol{g}^{*}(\bm{x}_{0},u),u)du=0\\
\implies \int_{0}^{1}Q_{NSP}(\tilde{\boldsymbol{g}}(u),u)du&=\int_{0}^{1}Q_{NSP}(\boldsymbol{g}^{*}(\bm{x}_{0},u),u)du=0
\end{align*}\end{linenomath*}
where the last equation follows from continuity of $Q_{NSP}$. 

Since $\tilde{\bm{g}}\in\mathcal{G}_{0}$, by Theorem \ref{1.thmIDNSP}, $\tilde{\boldsymbol{g}}(\cdot)=\boldsymbol{g}^{*}(\bm{x}_{0},\cdot)$ on $\mathcal{U}_{0}$. Hence, $\bm{g}_{k_{l}}(u)$ converges to $\bm{g}^{*}(\bm{x}_{0},u)$ for all $u\in\mathcal{U}_{0}$. Since pointwise convergence of a sequence of monotonic functions on a compact domain implies uniform convergence if the limiting function is continuous, by continuity of $\boldsymbol{g}^{*}(\bm{x}_{0},\cdot)$,
\begin{linenomath*}\begin{equation*}
\lim_{k_{l}\to \infty}\sup_{u\in\mathcal{U}_{0}}||\boldsymbol{g}^{*}(\bm{x}_{0},u)-\boldsymbol{g}_{k_{l}}(u)||= 0,
\end{equation*}\end{linenomath*}
contradicting with $\{\boldsymbol{g}_{k_{l}}\}\subseteq\mathcal{G}_{0}^{-}$.
\end{proof}
\begin{rem}
From the last step in the proof of Corollary \ref{1.corIDNSP}, $\mathcal{U}_{0}$ in Corollary \ref{1.corIDNSP} can be replaced by $[0,1]$ if the uniqueness of $\boldsymbol{g}^{*}(\bm{x}_{0},\cdot)$ holds on the entire interval $[0,1]$. This would be the case if the infimum is taken over $\mathcal{G}^{*}$ instead of $\mathcal{G}_{0}$ in equation \eqref{1.eq42}.
\end{rem}
\end{appendices}
\makeatletter\@input{ns.tex}\makeatother
\nocite{united}
\bibliography{/Users/Ecthelion/Documents/Econometrics/NonparametricIdentification/Paper/references}
\end{document}


\nolinenumbers
\title{\vspace{-2cm} Supplement to "Matching Points: Supplementing Instruments with Covariates in Triangular Models"  
}
\author{Junlong Feng\thanks{Department of Economics, the Hong Kong University of Science and Technology; \href{mailto:jlfeng@ust.hk}{jlfeng@ust.hk}.}}
\maketitle
\vspace{-0.45in}
\begin{abstract}
Appendix \ref{secSA} discusses estimation of the matching points when multiple solutions exist to the generalized propensity score matching, derives an estimator for the nonseparable outcome function, and provides the proofs of the estimators' asymptotic properties. Appendix \ref{secSB} shows additional simulation results. 
\bigskip
\end{abstract}

\begin{appendices}
\renewcommand{\thesection}{S.A}
\section{More on Estimation}\label{secSA}
In this section, I still focus on the benchmark case specified in the beginning of Section \ref{1.sec4}, that is, only one covariate $X$ is in the model, and all the solutions to the generalized propensity score matching $\bm{p}(x,z')=\bm{p}(x_{0},z)$ are the matching points. Appendix \ref{1.appA2} discusses estimation of the set of the matching points when multiple solutions exist to $\bm{p}(x,z')=\bm{p}(x_{0},z)$. Appendix \ref{secSA2} provides an estimator of the nonseparable outcome function and presents its asymptotic properties. Appendix \ref{1.appE1} provides the proofs of these results. Proofs of the lemmas used in Appendix \ref{1.appE1} are in Appendix \ref{1.appE2}. 
\subsection{Estimating the Set of Matching Points}\label{1.appA2}
Let $\{x_{m1}\}$ be the set of the solutions to $\bm{p}(x,1)=\bm{p}(x_{0},0)$ and $\{x_{m2}\}$ be the solution set to $\bm{p}(x,0)=\bm{p}(x_{0},1)$. Recall the estimator for the general case is defined by equation \eqref{1.eq22} in Section \ref{1.sec4} with $a_{n}>0$:
\begin{linenomath*}\begin{equation*}
\hat{Q}_{x}(\hat{x}_{m1},\hat{x}_{m2})\leq \inf_{(x_{1},x_{2})\in S_{0}^{2}(X)} \hat{Q}_{x}(x_{1},x_{2})+a_{n}^{2}
\end{equation*}\end{linenomath*}

Let $\hat{\mathcal{X}}_{m}=\{(\hat{x}_{m1},\hat{x}_{m2})\}$ and $\mathcal{X}_{m}=\{(x_{m1},x_{m2})\}$. Let $Q_{x}$ be the probability limit of $\hat{Q}_{x}$. By construction, $Q_{x}(x_{m1},x_{m2})=0$. For any $(x_{1},x_{2})\in S^{2}(X)$, let $\rho((x_{1},x_{2}),\mathcal{X}_{m})\equiv \inf_{(x'_{1},x'_{2})\in\mathcal{X}_{m}}||(x_{1},x_{2})-(x'_{1},x'_{2})||$ be the distance of $(x_{1},x_{2})$ to the set $\mathcal{X}_{m}$. For two generic real numbers, let $a\land b\equiv \min\{a,b\}$ and $a\lor b\equiv \max\{a,b\}$. Following \cite{chernozhukov2007estimation}, the following assumption is made:
\begin{appxass}\label{assApp1}
There exist $C,\delta>0$ such that 
\begin{linenomath*}\begin{equation*}
Q_{x}(x_{1},x_{2})\geq C\cdot(\rho(x_{1},x_{2}),\mathcal{X}_{m})\land\delta)^{2}
\end{equation*}\end{linenomath*}
\end{appxass} 

Assumption \ref{assApp1} is satisfied when $\mathcal{X}_{m}$ is a singleton under the full rankness and continuity of the Jacobian of $Q_{x}$ near $(x_{m1},x_{m2})$. See \cite{chernozhukov2007estimation} for a detailed discussion. By verifying the conditions in their Theorem 3.1, the following theorem establishes consistency of $\hat{\mathcal{X}}_{m}$ in Hausdorff distance\footnote{The Hausdorff distance $\rho_{H}$ between two generic subsets $A$ and $B$ of a metric space endowed with metric $\rho$ is defined as $\rho_{H}(A,B)=\max\{\sup_{a\in A}\inf_{b\in B} \rho(a,b),\sup_{b\in B}\inf_{a\in A} \rho(a,b)\}$.}.
\begin{appxthm}\label{1.thmA1}
 Let $c_{n}=\frac{\sqrt{\log(n)}}{\sqrt{nh_{x}}}+h_{x}^{2}$ and $a_{n}=\sqrt{\log(n)}c_{n}$. Under Assumptions \ref{1.assREGMP} and \ref{assApp1}, $\rho_{H}(\hat{\mathcal{X}}_{m},\mathcal{X}_{m})=O_{p}\left(a_{n}\right)$.
\end{appxthm}
\begin{proof}
See Appendix \ref{1.appE1}.
\end{proof}

Once $\hat{\mathcal{X}}_{m}$ is obtained, one can select an element in it as the estimator of a matching point. Although the rate of convergence is slightly slower than the case when the solution to $Q_{x}(x_{1},x_{2})=0$ is unique ($O_{p}\left(\frac{1}{\sqrt{nh_{x}}}+h_{x}^{2}\right)$), estimation of the outcome function can stay unaffected by appropriately undersmoothing. 

In order to conduct the overidentification test for the matching points, I make the following assumption and propose a post-estimation procedure to estimate the isolated matching points that are locally unique.
\begin{appxass}\label{1.assISO}
Let $\mathcal{X}_{ISO}\subseteq\mathcal{X}_{m}$ be the set of the isolated matching points:
\begin{linenomath*}\begin{equation*}
\mathcal{X}_{ISO}=\{(x_{m1},x_{m2})\in \mathcal{X}_{m}:\inf_{(x_{1},x_{2})\in\mathcal{X}_{m}\backslash\{(x_{m1},x_{m2})\}}||(x_{1},x_{2})-(x_{m1},x_{m2})||>0\}.
\end{equation*}\end{linenomath*}
Suppose the following hold:

i) All the isolated matching points are in the interior of $S_{0}^{2}(X)$.

ii) The Jacobian matrix $\partial_{x'}\Delta \bm{p}(x_{m1},x_{m2})$ defined in Section \ref{sec4.1} is full rank for every $(x_{m1},x_{m2})\in \mathcal{X}_{ISO}$.
\end{appxass}

The following post-estimation procedure is designed to obtain the isolated matching points:
\begin{itemize}
\item Step 0. Obtain $\hat{\mathcal{X}}_{m}$ by equation \eqref{1.eq22} with $a_{n}$ specified in Theorem \ref{1.thmA1}.
\item Step 1 (Isolated matching points selection). Cover $\hat{\mathcal{X}}_{m}$ with squares $\{I_{k}\}$ of side $b_{n}=\log(n)a_{n}$. Select $I_{k}$ if it fully covers a cluster of solutions.
\item Step 2 (Reestimation). Minimize $\hat{Q}_{x}(x_{1},x_{2})$ again over each local parameter space $I_{k}$.
\end{itemize}
From Theorem \ref{1.thmA1} and Assumption \ref{1.assISO}, if $\mathcal{X}_{ISO}$ is nonempty, there are isolated clusters of solutions to equation \eqref{1.eq22}. These clusters are close to their corresponding isolated matching points with probability approaching one (w.p.a.1). The diameters of such clusters shrink to zero at the rate of $a_{n}$ as they converge to the isolated matching points in probability. Therefore, each cluster can be covered by a square $I_{k}$ w.p.a.1 since $b_{n}>a_{n}$ in order. For the reestimation step, again, as $b_{n}/a_{n}\to\infty$, the isolated matching point is in the interior of $I_{k}$ w.p.a.1. Under that event, the minimizer of $\hat{Q}_{x}$ satisfies the first order condition. The asymptotic distribution and the overidentification test derived in Section \ref{sec4.1} hold if $\sqrt{nh_{x}}h_{x}^{2}\to 0$ and $nh_{x}^{3}\to\infty$.

\subsection{Estimating the Nonseparable Outcome Function}\label{secSA2}
In the benchmark case, we have 
\[
\Psi(\boldsymbol{g}(u))=
\begin{pmatrix}
\sum_{d=1}^{3}p_{d}(x_{0},0)\cdot F_{Y|DXZ}(g_{d}(u)|d,x_{0},0)\\
\sum_{d=1}^{3}p_{d}(x_{0},1)\cdot F_{Y|DXZ}(g_{d}(u)|d,x_{0},1)\\
\sum_{d=1}^{3}p_{d}(x_{m1},0)\cdot F_{Y|DXZ}(\varphi_{d}(g_{d}(u);x_{m1})|d,x_{m1},0)\\
\sum_{d=1}^{3}p_{d}(x_{m2},1)\cdot F_{Y|DXZ}(\varphi_{d}(\tilde{g}_{d}(u);x_{m2})|d,x_{m2},1)
\end{pmatrix}.
\]
where the dependence of $\Psi$ on the conditioning points is omitted for the ease of notation. Recall that Theorem \ref{1.thmIDNSP} implies that $\boldsymbol{g}^{*}(\boldsymbol{x}_{0},\cdot)$ is the unique minimizer (up to the boundaries) to $\min_{\boldsymbol{g}\in\mathcal{G}_{0}} \int_{0}^{1}Q_{NSP}(\boldsymbol{g}(u),u)du$ where $Q_{NSP}$ and $\mathcal{G}_{0}$ are defined in equations \eqref{eq3.1} and \eqref{eqG0} in Section \ref{1.sec3}. 

Like the matching points and the separable outcome function, $\boldsymbol{g}^{*}(\boldsymbol{x}_{0},\cdot)$ can be estimated by the sample analogue of the population minimization problem. Construct $\hat{Q}_{NSP}(\boldsymbol{g}(u),u)$ by plugging in estimators of $\Psi$ and $\boldsymbol{W}_{g}(u)$. For $\Psi$, estimate the generalized propensity scores by equation \eqref{1.eq23}. For the conditional cumulative distribution functions, I adopt the following smoothed kernel estimator (e.g. \cite{hansen2004nonparametric} and \cite{li2008nonparametric}):
\begin{linenomath*}\begin{equation}\label{1.eq28}
\widehat{F}_{Y|DXZ}(y|d,x,z)=\frac{\sum_{i=1}^{n}L(\frac{y-Y_{i}}{h_{0}})\mathbbm{1}(D_{i}=d)K(\frac{X_{i}-x}{h_{g}})\mathbbm{1}(Z_{i}=z)}{\sum_{i=1}^{n}\mathbbm{1}(D_{i}=d) K(\frac{X_{i}-x}{h_{g}})\mathbbm{1}(Z_{i}=z)}
\end{equation}\end{linenomath*}
where $L(\cdot)$ is a smooth cumulative distribution function and $h_{0}$ converges to $0$ faster than $h_{g}$. Suppose $(x_{0},z)$ and $(x_{m},z')$ are a matching pair. The mapping $\varphi_{d}$ can be estimated by
\begin{linenomath*}\begin{equation}\label{1.eq29}
\hat{\varphi}_{d}(y;x_{m})=\arg\min_{y'\in S(Y|d)}\big(\widehat{F}_{Y|DXZ}(y'|d,\hat{x}_{m},z')-\widehat{F}_{Y|DXZ}(y|d,x_{0},z)\big)^{2}
\end{equation}\end{linenomath*}
where the support $S(Y|d)$ can be estimated faster than the nonparametric rate so is treated as known.

Let $u_{j}=\frac{j}{J}$ where $1\leq j\leq J$ and $J\to\infty$. Let $\underline{\boldsymbol{y}}\equiv(\underline{y}_{d})_{d}$ and $\bar{\boldsymbol{y}}\equiv(\bar{y}_{d})_{d}$ where $\underline{y}_{d}$ and $\bar{y}_{d}$ are the lower and the upper boundaries of $S(Y|d),d=1,2,3$. The estimator of $\boldsymbol{g}^{*}(\boldsymbol{x}_{0},\cdot)$ solves the following minimization problem:
\begin{linenomath*}\begin{align}
\hat{\bm{g}}\equiv\min_{\underline{\boldsymbol{y}}\leq \boldsymbol{g}(u_{1})\leq ...\leq \boldsymbol{g}(u_{J})\leq \bar{\boldsymbol{y}}} \frac{1}{J}\sum_{j=1}^{J}\hat{Q}_{NSP}(\boldsymbol{g}(u_{j}),u_{j})\label{1.eq27}
\end{align}\end{linenomath*}
%

The minimization problem \eqref{1.eq27} defines a sieve estimator. The constraint induces a finite dimensional function space $\widehat{\mathcal{G}}\subseteq\mathcal{G}_{0}$; each element in it is increasing and piecewise affine with line segments connecting adjacent nodes. By sending $J$ to infinity, elements in $\widehat{\mathcal{G}}$ are able to approximate any increasing function bounded between $\underline{\boldsymbol{y}}$ and $\bar{\boldsymbol{y}}$. One can also add a penalty term to smooth the estimator, for instance $\lambda \sum_{j=2}^{J}\left[\boldsymbol{g}(u_{j})-\boldsymbol{g}(u_{j-1})\right]'\left[\boldsymbol{g}(u_{j})-\boldsymbol{g}(u_{j-1})\right]$ where $\lambda\to 0$ sufficiently fast.

It is noteworthy that adopting a sieve estimator in this paper is not to solve an ill-posed problem \citep{chen2007large,chen2012estimation,chen2015sieve}, but to coordinate with the identification concept. In fact, if identification is established for a fixed $u$, one can simply minimize $\hat{Q}_{NSP}(\boldsymbol{g}(u),u)$ like in \cite{lewbel2007local} and the true parameter $\bm{g}^{*}(x_{0},u)$ is finite dimensional. Under identification of the solution path, however, the minimizer of $Q_{NSP}(\bm{g}(u),u)$ may not be unique at a given $u$. Joint estimation at multiple nodes under monotonicity is necessary to guarantee identification of $\bm{g}^{*}(\bm{x}_{0},\cdot)$.

\begin{appxass}\label{1.assREGNSP} 
Let $F_{YDXZ}\equiv F_{Y|DXZ}\cdot \mathbb{P}_{DZ|X}\cdot f_{X}$. The cumulative distribution functions $F_{YDXZ}$ and $F_{Y|DXZ}$ are three times continuously differentiable in the continuous arguments on $S(Y,X|D,Z)$ and $S(Y|D,X,Z)$ with bounded derivatives. The conditional density $f_{Y|DXZ}$ is uniformly bounded away from $0$ almost surely over $S(Y|D,X,Z)$. The derivative of $L(\cdot)$ is the kernel function $K(\cdot)$. 
\end{appxass}
Since the density is bounded away from zero on the support, the functions $F_{YDXZ}$ and $F_{Y|DXZ}$ are not differentiable at the boundaries of the support. However, outside the support, these functions are flat (equal to zero or one), so they are globally Lipschitz continuous. This property is useful in proving uniform consistency of the cumulative distribution function estimator over sets that are larger than the support. 

Recall $\mathcal{U}_{0}$ is a compact subset in the interior of $[0,1]$. Let $a_{n}$ be the rate of convergence of $(\hat{x}_{m1},\hat{x}_{m2})$. The following theorem establishes uniform consistency of $\hat{\bm{g}}(\cdot)$ over $\mathcal{U}_{0}$.

\begin{appxthm}\label{thmSUPCONS}
Suppose $\log(n)/\sqrt{nh_{g}}$, $h_{g}$, $h_{0}$, $a_{n}$ and $h_{0}/h_{g}$ are all $o(1)$. Under Assumptions \ref{1.assREGMP}, \ref{1.assREGNSP} and the conditions in Theorem \ref{1.thmIDNSP}, if $J\to\infty$, then
\begin{linenomath*}\begin{equation}\label{1.eq43}
\sup_{u\in\mathcal{U}_{0}}\big|\hat{\boldsymbol{g}}(u)-\boldsymbol{g}^{*}(x_{0},u)\big|=o_{p}(1).
\end{equation}\end{linenomath*}
\end{appxthm}

For the asymptotic distribution of $\hat{\boldsymbol{g}}(u_{0})-\boldsymbol{g}^{*}(x_{0},u_{0})$ for a fixed $u_{0}\in [u_{1},...,u_{J}]\cap\mathcal{U}_{0}$, if $J$ goes to infinity sufficiently slow and $|\boldsymbol{g}^{*}(x_{0},u)-\boldsymbol{g}^{*}(x_{0},u\pm 1/J)| =O(1/J)$ for all $u\in [0,1]$, then all the inequalities in the constraint in equation \eqref{1.eq27} can be nonbinding for $\hat{\bm{g}}$ under uniform consistency. The estimator $\hat{\boldsymbol{g}}$ at $u_{0}$ is thus asymptotically equivalent as the unconstrained pointwise estimator minimizing $\hat{Q}_{NSP}(\bm{g}(u_{0}),u_{0})$ under global identification at $u_{0}$. The asymptotic distribution of $\hat{\boldsymbol{g}}(u_{0})-\boldsymbol{g}^{*}(x_{0},u_{0})$ can then be obtained following the standard argument of local GMM estimators.

Let $\tilde{\boldsymbol{z}}_{1},...,\tilde{\boldsymbol{z}}_{6}$ be $(x_{0},0)$, $(x_{m1},0)$, $(x_{m1},1)$, $(x_{0},1)$, $(x_{m2},1)$ and $(x_{m2},0)$. Let $\phi_{d1}=\frac{f_{Y|DXZ}(g^{*}_{d}(x_{m1},u_{0})|d,x_{m1},0)}{f_{Y|DXZ}(g^{*}_{d}(x_{m1},u_{0})|d,x_{m1},1)}$, and $\phi_{d2}=\frac{f_{Y|DXZ}(g^{*}_{d}(x_{m2},u_{0})|d,x_{m2},1)}{f_{Y|DXZ}(g^{*}_{d}(x_{m2},u_{0})|d,x_{m2},0)}$. Let $\mathbb{V}_{\delta|DXZ}=F_{Y|DXZ}(1-F_{Y|DXZ})$. Let $\Sigma_{NSP}=\kappa\left(\Sigma_{NSP,1}+\Sigma_{NSP,2}+\Sigma_{NSP,3}\right)$ where $\kappa=\int K(\nu)^{2} d\nu$ and $\Sigma_{NSP,d}$ ($d=1,2,3$) equals\footnotesize
\begin{linenomath*}\begin{equation*}
\setlength\arraycolsep{0pt}
\begin{pmatrix}
\frac{p^{2}_{d}(\tilde{\boldsymbol{z}}_{1})\mathbb{V}_{\delta|DXZ}(d,\tilde{\boldsymbol{z}}_{1})}{f_{DXZ}(d,\tilde{\boldsymbol{z}}_{1})}&0& \frac{\phi_{d1}p_{d}(\tilde{\boldsymbol{z}}_{1})p_{d}(\tilde{\boldsymbol{z}}_{2})\mathbb{V}_{\delta|DXZ}(d,\tilde{\boldsymbol{z}}_{1})}{f_{DXZ}(d,\tilde{\boldsymbol{z}}_{1})}&0\\
0 &\frac{p^{2}_{d}(\tilde{\boldsymbol{z}}_{4})\mathbb{V}_{\delta|DXZ}(d,\tilde{\boldsymbol{z}}_{4})}{f_{DXZ}(d,\tilde{\boldsymbol{z}}_{4})}&0&\frac{\phi_{d2}p_{d}(\tilde{\boldsymbol{z}}_{4})p_{d}(\tilde{\boldsymbol{z}}_{5})\mathbb{V}_{\delta|DXZ}(d,\tilde{\boldsymbol{z}}_{4})}{f_{DXZ}(d,\tilde{\boldsymbol{z}}_{4})}\\
\frac{\phi_{d1}p_{d}(\tilde{\boldsymbol{z}}_{1})p_{d}(\tilde{\boldsymbol{z}}_{2})\mathbb{V}_{\delta|DXZ}(d,\tilde{\boldsymbol{z}}_{1})}{f_{DXZ}(d,\tilde{\boldsymbol{z}}_{1})}&0 & \sum_{k=1}^{3}\frac{\phi_{d1}^{2}p^{2}_{d}(\tilde{\boldsymbol{z}}_{2})\mathbb{V}_{\delta|DXZ}(d,\tilde{\boldsymbol{z}}_{k})}{f_{DXZ}(d,\tilde{\boldsymbol{z}}_{k})}&0\\
0& \frac{\phi_{d2}p_{d}(\tilde{\boldsymbol{z}}_{4})p_{d}(\tilde{\boldsymbol{z}}_{5})\mathbb{V}_{\delta|DXZ}(d,\tilde{\boldsymbol{z}}_{4})}{f_{DXZ}(d,\tilde{\boldsymbol{z}}_{4})}& 0 & \sum_{k=4}^{6}\frac{\phi_{d2}^{2}p^{2}_{d}(\tilde{\boldsymbol{z}}_{5})\mathbb{V}_{\delta|DXZ}(d,\tilde{\boldsymbol{z}}_{k})}{f_{DXZ}(d,\tilde{\boldsymbol{z}}_{k})}
\end{pmatrix}.
\end{equation*}\end{linenomath*}\normalsize
Let the feasible optimal weighting matrix $\hat{\Sigma}_{NSP}^{-1}$ be a uniformly consistent estimator of $\Sigma_{NSP}^{-1}$ obtained by a uniformly consistent first step estimator of $\bm{g}^{*}(x_{0},\cdot)$, for instance under the identify weighting matrix. Let $\Pi_{NSP}$ be the Jacobian matrix of $\Psi(\bm{g}^{*}(x_{0},u_{0}))$. Then we have the following theorem.
\begin{appxthm}\label{thmASYMP} 
Suppose $\sqrt{nh_{g}}\cdot h_{g}^{2}\to 0$, $h_{0}/h_{g}\to 0$, $\sqrt{nh_{g}}\cdot a_{n}\to 0$, $nh_{0}h_{g}\to \infty$, $nh_{g}^{3}\to \infty$ and $J^{3/2}\left(\sqrt{\log(n)/nh_{g}}+h_{g}\right)\to 0$. If for each $d\in S(D)$, $g_{d}^{*}(x_{0},\cdot)$ is differentiable on $[0,1]$ with the derivative bounded away from zero, then under all the conditions in Theorems \ref{1.thmIDNSP} and \ref{thmSUPCONS}, for any fixed node $u_{0}\in \{u_{1},...,u_{J}\}\cap \mathcal{U}_{0}$,
\begin{linenomath*}\begin{equation}
\sqrt{nh_{g}}(\hat{\boldsymbol{g}}(u_{0})-\boldsymbol{g}^{*}(x_{0},u_{0})) \overset{d}{\to}
 \mathcal{N}\left(0,\left[\Pi_{NSP}'\Sigma_{NSP}^{-1}\Pi_{NSP}\right]^{-1}\right)\label{1.eq47}
\end{equation}\end{linenomath*}
\end{appxthm}
In Theorem \ref{thmASYMP}, undersmoothing is made ($\sqrt{nh_{g}}\cdot h_{g}^{2}\to 0$) for simplicity. Similarly, I assume $h_{0}/h_{g}\to 0$ so that the smoothing bias introduced by $L$ is negligible. The requirement $\sqrt{nh_{g}}\cdot a_{n}\to 0$ guarantees that the estimation of the generalized propensity scores and the matching points do not affect the asymptotic distribution. The requirements $nh_{0}h_{g}\to \infty$ and $nh_{g}^{3}\to \infty$ ensure that the partial derivatives of $\hat{F}_{Y|DXZ}$ with respect to the continuous arguments are uniformly consistent over the interior of the support. Finally, the rate of $J$ guarantees that the constraints in the minimization problem \eqref{1.eq27} are nonbinding within $\mathcal{U}_{0}$. Based on the asymptotic distribution, an overidentification test can be constructed under the optimal weighting matrix by $\mathcal{J}_{NSP}=nh_{g}\hat{Q}_{NSP}(\hat{\boldsymbol{g}}(u_{0}),u_{0})$. Under the null that all the moment conditions hold jointly, $\mathcal{J}_{NSP}\overset{d}{\to} \chi_{1}^{2}$.

\subsection{Proofs of Results in Appendix \ref{secSA2}}\label{1.appE1}

\subsubsection{The Matching Points}
\begin{proof}[Proof of Theorem \ref{1.thmA1}]
I prove the theorem by verifying the conditions needed for Theorem 3.1 in \cite{chernozhukov2007estimation}. By a bit abuse of notation for simplicity, denote an arbitrary $(x_{1},x_{2})\in S^{2}(X)$ by $\bm{x}$. Two conditions need to be verified: i) $\sup_{\boldsymbol{x}\in S_{0}^{2}(X)}|\hat{Q}_{x}(\boldsymbol{x})-Q_{x}(\boldsymbol{x})|=O_{p}(c_{n}^{2})$ and ii) there exist positive $(\delta,\mu)$ such that for any $\varepsilon\in(0,1)$, there are $(\mu_{\varepsilon},n_{\varepsilon})$ such that for all $n>n_{\varepsilon}$, $\hat{Q}_{x}(\boldsymbol{x})\geq \mu\left[\rho(\boldsymbol{x},\mathcal{X}_{m})\land\delta\right]^{2}$ uniformly on $\Delta\equiv \{\boldsymbol{x}\in S_{0}^{2}(X):\rho(\boldsymbol{x},\mathcal{X}_{m})\geq \mu_{\varepsilon}c_{n}\}$ with probability at least $1-\varepsilon$.

For i), the following equation holds for any $z\in S(Z)$ following the standard argument for Nadaraya-Watson estimator under Assumption \ref{1.assREGMP} (see for example \cite{mack1982weak}, \cite{silverman2018density}, \cite{hardle1990bandwidth}, \cite{masry1996multivariate}, etc.):
\begin{linenomath*}\begin{align*}
\sup_{x\in S_{0}(X)}\big|\hat{p}_{d}(x,z)-p_{d}(x,z))\big|=O_{p}\left(\sqrt{\frac{\log(n)}{nh_{x}}}+h_{x}^{2}\right)\equiv O_{p}(c_{n})
\end{align*}\end{linenomath*}
Therefore, $\sup_{\boldsymbol{x}\in S_{0}^{2}(X)}|\hat{Q}_{x}(\boldsymbol{x})-Q_{x}(\boldsymbol{x})|=O_{p}(c_{n}^{2})$.
%
%

For ii), by Assumption \ref{assApp1} and $Q_{x},\hat{Q}_{x}\geq 0$,
\begin{linenomath*}\begin{align*}
\inf_{\boldsymbol{x}\in\Delta}\hat{Q}_{x}(\boldsymbol{x})=&\inf_{\boldsymbol{x}\in\Delta}|\hat{Q}_{x}(\boldsymbol{x})-Q_{x}(\boldsymbol{x})+Q_{x}(\boldsymbol{x})|\\
\geq &\inf_{\boldsymbol{x}\in\Delta}Q_{x}(\boldsymbol{x})-\sup_{\boldsymbol{x}\in\Delta}|\hat{Q}_{x}(\boldsymbol{x})-Q_{x}(\boldsymbol{x})|\\
\geq &C\left[\rho(\boldsymbol{x},\mathcal{X}_{m})\land\delta\right]^{2}-O_{p}\left(c_{n}^{2}\right)\\
\geq &C[\mu_{\varepsilon}^{2}c_{n}^{2}\land\delta^{2}]-O_{p}\left(c_{n}^{2}\right)
\end{align*}\end{linenomath*}
Therefore, we can choose large enough $(\mu_{\varepsilon},n_{\varepsilon})$ so that the desired inequality holds uniformly on $\Delta$ for some $\delta$ and $\mu$.
\end{proof}

\subsubsection{The Nonseparable Model}
For convenience, let us review some notation introduced earlier. For $d\in S(D)$ and $x\in S(X)$, let $\underline{y}_{dx}$ and $\bar{y}_{dx}$ denote the lower and upper boundaries of $S(Y|d,x)$. By continuity, $S(Y|d,x)=[\underline{y}_{dx},\bar{y}_{dx}]$. Similarly, let $S(Y|d)=[\underline{y}_{d},\bar{y}_{d}]$. Several parameter spaces are used in the proofs:
\begin{linenomath}\begin{align*}
&\mathcal{G}_{0}=\{\bm{g}:[0,1]\mapsto\prod_{d=1}^{3}S(Y|d)\text{ and is weakly increasing}\},\\
&\mathcal{G}^{*}=\{\bm{g}:[0,1]\mapsto\prod_{d=1}^{3}S(Y|d,x_{0})\text{ and is weakly increasing}\},\\
&\hat{\mathcal{G}}=\{\bm{g}:[0,1]\mapsto\prod_{d=1}^{3}S(Y|d)\text{ and is piecewise affine with $J$ increasing nodes at }\{u_{j}\}\}
\end{align*}\end{linenomath}

The following lemmas are needed to prove Theorems \ref{thmSUPCONS} and \ref{thmASYMP}. The proofs of these lemmas are in Appendix \ref{1.appE2}
\begin{appxlem}\label{1.lemA1}
Suppose $h_{0}/h_{g}\to 0,h_{g}\to 0$ and $\sqrt{nh_{g}}\to\infty$. Under Assumptions \ref{1.assREGMP} and \ref{1.assREGNSP}, the following holds for all $d\in S(D)$ and $z\in S(Z)$,
\begin{linenomath*}\begin{align}
&\sup_{\substack{y\in S(Y|d)\\x\in S_{0}(X)}}\big|\hat{F}_{Y|DXZ}(y|d,x,z)-F_{Y|DXZ}(y|d,x,z)\big|=O_{p}\left(\sqrt{\frac{\log(n)}{nh_{g}}}+h_{g}\right)\label{1.eqA5}
\end{align}\end{linenomath*}
\end{appxlem}
The term $h_{g}$ captures the order of the bias and is not the common $h_{g}^{2}$. This is because the supremum is taken over $S(Y|d)$ where $F_{Y|DXZ}$ is not differentiable everywhere. Taylor expansion is not available but since $F_{Y|DXZ}$ is globally Lipshitz continuous, an order of $O(h_{g})$ for the bias can still be achieved.  

Recall $a_{n}$ is the rate of convergence of the estimated matching points. Let $r_{n}=O_{p}\left(\sqrt{\frac{\log(n)}{nh_{g}}}+h_{g}\right)$. 
\begin{appxlem}\label{1.lemA2}
Suppose $(x_{0},z)$ and $(x_{m},z')$ are a matching pair. Under the conditions in Theorem \ref{thmSUPCONS}, the following holds for all $d\in S(D)$:
\begin{linenomath*}\begin{equation*}
\sup_{y\in [\underline{y}_{d},\bar{y}_{d}]}\big|\hat{F}_{Y|DXZ}(\hat{\varphi}_{d}(y;\hat{x}_{m})|d,\hat{x}_{m},z)-F_{Y|DXZ}(\varphi_{d}(y;x_{m})|d,x_{m},z)\big|=O_{p}(r_{n}+a_{n})
\end{equation*}\end{linenomath*}
\end{appxlem}

Under Lemma \ref{1.lemA2} and given uniform consistency of the estimator of the generalized propensity scores, it can be shown that
\begin{linenomath*}\begin{equation*}
\sup_{\boldsymbol{y}\in\prod_{d}S(Y|d)}|\hat{Q}_{NSP}(\boldsymbol{y},u)-Q_{NSP}(\boldsymbol{y},u)|=o_{p}(1),\forall u\in [0,1].
\end{equation*}\end{linenomath*}
Then we have the following lemma.
\begin{appxlem}\label{1.lemA3}
Under all the conditions in Theorem \ref{thmSUPCONS}, the following holds:
\begin{linenomath*}\begin{equation*}
\sup_{\boldsymbol{g}\in\hat{\mathcal{G}}}\big|\frac{1}{J}\sum_{j=1}^{J}\hat{Q}_{NSP}(\boldsymbol{g}(u_{j}),u_{j})-\int_{0}^{1}Q_{NSP}(\boldsymbol{g}(u),u)du\big|=o_{p}(1)
\end{equation*}\end{linenomath*}
\end{appxlem}
%

\begin{proof}[Proof of Theorem \ref{thmSUPCONS}]
Suppose $\sup_{u\in\mathcal{U}_{0}}|\hat{\boldsymbol{g}}(u)-\boldsymbol{g}^{*}(x_{0},u)|>\delta$ for some postive $\delta$. By construction, $\hat{\mathcal{G}}\subseteq\mathcal{G}_{0}$ and thus $\hat{\boldsymbol{g}}\in\mathcal{G}_{0}$. By Corollary \ref{1.corIDNSP}, there exists an $\varepsilon>0$ such that 
\begin{linenomath*}\begin{equation*}
\int_{0}^{1}Q_{NSP}(\hat{\boldsymbol{g}}(u),u)du-\int_{0}^{1}Q_{NSP}(\boldsymbol{g}^{*}(x_{0},u),u)du>\varepsilon
\end{equation*}\end{linenomath*}
Let $\bm{g}_{0}\in \hat{\mathcal{G}}$ such that $\bm{g}_{0}(0)=\bm{g}^{*}(x_{0},0)$, $\bm{g}_{0}(1)=\bm{g}^{*}(x_{0},1)$, and $\bm{g}_{0}(u_{j})=\bm{g}^{*}(x_{0},u_{j})$ for all $u_{j}$. Then by monotonicity and $J\to\infty$, $||\bm{g}_{0}(u)-\bm{g}^{*}(x_{0},u)||=o(1)$ for all $u\in [0,1]$. Therefore, by the Dominated Convergence Theorem, 
\begin{linenomath*}\begin{equation*}
\int_{0}^{1}Q_{NSP}(\boldsymbol{g}_{0}(u),u)du-\int_{0}^{1}Q_{NSP}(\boldsymbol{g}^{*}(x_{0},u),u)du=o(1).
\end{equation*}
Hence,
\end{linenomath*}
\begin{linenomath*}\begin{align*}
&\int_{0}^{1}Q_{NSP}(\hat{\boldsymbol{g}}(u),u)du-\int_{0}^{1}Q_{NSP}(\boldsymbol{g}^{*}(x_{0},u),u)du\\
= & \frac{1}{J}\sum_{j=1}^{J}\hat{Q}_{NSP}(\hat{\boldsymbol{g}}(u_{j}),u_{j})-\int_{0}^{1}Q_{NSP}(\boldsymbol{g}^{*}(x_{0},u),u)du\\
&+\int_{0}^{1}Q_{NSP}(\hat{\boldsymbol{g}}(u),u)du- \frac{1}{J}\sum_{j=1}^{J}\hat{Q}_{NSP}(\hat{\boldsymbol{g}}(u_{j}),u_{j})\\
\leq &\frac{1}{J}\sum_{j=1}^{J}\hat{Q}_{NSP}(\boldsymbol{g}_{0}(u_{j}),u_{j})-\int_{0}^{1}Q_{NSP}(\boldsymbol{g}_{0}(u),u)du\\
&+\int_{0}^{1}Q_{NSP}(\boldsymbol{g}_{0}(u),u)du-\int_{0}^{1}Q_{NSP}(\boldsymbol{g}^{*}(x_{0},u),u)du\\
&+\sup_{\boldsymbol{g}\in\hat{\mathcal{G}}}\big|\frac{1}{J}\sum_{j=1}^{J}\hat{Q}_{NSP}(\boldsymbol{g}(u_{j}),u_{j})-\int_{0}^{1}Q_{NSP}(\boldsymbol{g}(u),u)du\big|\\
\leq &2\sup_{\boldsymbol{g}\in\hat{\mathcal{G}}}\big|\frac{1}{J}\sum_{j=1}^{J}\hat{Q}_{NSP}(\boldsymbol{g}(u_{j}),u_{j})-\int_{0}^{1}Q_{NSP}(\boldsymbol{g}(u),u)du\big|+o(1)\\
=&o_{p}(1)
\end{align*}\end{linenomath*}
where the first inequality is by the definition of $\hat{\bm{g}}$ and $\bm{g}_{0}\in\hat{\mathcal{G}}$. The last equality is by Lemma \ref{1.lemA3}.
\end{proof}
%
To derive the asymptotic distribution in Theorem \ref{thmASYMP}, the inverse of the divergence rate of $J$ needs to be slower than the rate of convergence of $\hat{\bm{g}}$ at the interior nodes, so that the inequality constraints in the minimization problem \eqref{1.eq27} are nonbinding. Lemma \ref{1.thmROCNSP} provides an upper bound on the rate of convergence of $\hat{\bm{g}}(u_{j})$.
\begin{appxlem}\label{1.thmROCNSP}
Suppose $a_{n}=o(r_{n})$. Under all the conditions in Theorem \ref{thmSUPCONS}, for node $u_{j}$,
\begin{linenomath*}\begin{align}\label{1.eq44}
\max_{u_{j}\in\mathcal{U}_{0}}||\hat{\boldsymbol{g}}(u_{j})-\boldsymbol{g}^{*}(x_{0},u_{j})||=O_{p}(\sqrt{J}r_{n})
\end{align}\end{linenomath*}
\end{appxlem}
Now we can derive the asymptotic distribution of $\hat{\bm{g}}(x_{0},u_{0})$ for some $u_{0}\in [u_{1},...,u_{J}]\cap \mathcal{U}_{0}$.
\begin{proof}[Proof of Theorem \ref{thmASYMP}]
For all $d\in S(D)$, there exists a $C>0$ such that
\begin{linenomath*}\begin{align*}
&\min_{\substack{u_{j-1},u_{j}\in \mathcal{U}_{0}\\j=2,...,J}}\left(\hat{g}_{d}(u_{j})-\hat{g}_{d}(u_{j-1})\right)\\
=&\min_{\substack{u_{j-1},u_{j}\in \mathcal{U}_{0}\\j=2,...,J}}\left[\left(\hat{g}_{d}(u_{j})-g^{*}_{d}(x_{0},u_{j})\right)-\left(\hat{g}_{d}(u_{j-1})-g^{*}_{d}(x_{0},u_{j-1})\right)
+g^{*}_{d}(x_{0},u_{j})-g^{*}_{d}(x_{0},u_{j-1})\right]\\
\geq & \inf_{u\in\mathcal{U}_{0}}\left(g^{*}_{d}(x_{0},u)-g^{*}_{d}\left(x_{0},u-\frac{1}{J}\right)\right)-2\max_{u_{j}\in\mathcal{U}_{0}}||\hat{\boldsymbol{g}}(u_{j})-\boldsymbol{g}^{*}(x_{0},u_{j})||\\
\geq & \frac{C}{J}-O_{p}\left(\sqrt{J}r_{n}\right)
\end{align*}\end{linenomath*}
where the last inequality holds because the derivative of $g^{*}_{d}(x_{0},\cdot)$ is bounded away from $0$. Under the rate for $J$ specified in Theorem \ref{thmASYMP}, $1/J$ dominates the second term in order. Therefore, for all $d\in S(D)$ and all $u_{j-1},u_{j}\in \mathcal{U}_{0}$, $\hat{g}_{d}(u_{j})>\hat{g}_{d}(u_{j-1})$ w.p.a.1. Consequently, $\hat{\boldsymbol{g}}(u_{0})$ satisfies the first order condition w.p.a.1. since the constraints for it are nonbinding.

In order to obtain the expansion of the first order condition for $\hat{\boldsymbol{g}}(u_{0})$, let us first consider a typical term in $\hat{\Psi}(\hat{g}_{d}(u_{0}))$, $\hat{F}_{Y|DXZ}(\hat{\varphi}_{d}(\hat{g}_{d}(u_{0}))|d,\hat{x}_{m},z)$, where $\hat{x}_{m}$ is a generic estimated matching point such that $(x_{0},z)$ and $(x_{m},z')$ are a matching pair. The dependence of $\hat{\varphi}_{d}$ on $\hat{x}_{m}$ is omitted for simplicity. 

Since $g^{*}_{d}(x_{0},u_{0})$ is in the interior of $S(Y|d,x_{0})$, $\varphi_{d}(g^{*}_{d}(x_{0},u_{0}))$ is unique and in the interior of $S(Y|d,x_{m})$ as well. The mapping $\hat{\varphi}(g^{*}_{d}(x_{0},u_{0}))$ is then consistent following the standard argument for consistency of extremum estimators and uniform consistency of $\hat{F}_{Y|DXZ}$. The partial derivatives $\partial_{Y}\hat{F}_{Y|DXZ}(\cdot|d,\cdot,z)$ and $\partial_{X}\hat{F}_{Y|DXZ}(\cdot|d,\cdot,z)$ are also uniformly consistent of $f_{Y|DXZ}(\cdot|d,\cdot,z)$ and $\partial_{X}F_{Y|DXZ}(\cdot|d,\cdot,z)$ over a small neighborhood of any interior point in the support under Assumptions \ref{1.assREGMP} and \ref{1.assREGNSP} and the rates of the bandwidths specified in Theorem \ref{thmASYMP}. By manipulating the first order condition of the minimization problem equation \eqref{1.eq29} with $y=g^{*}_{d}(x_{0},u_{0})$, we have
\begin{linenomath*}\begin{align*}
&\hat{\varphi}_{d}(g^{*}_{d}(x_{0},u_{0}))-\varphi_{d}(g_{d}^{*}(x_{0},u_{0}))\\
=&\frac{\hat{F}_{Y|DXZ}(g^{*}_{d}(x_{0},u_{0})|d,x_{0},z)-\hat{F}_{Y|DXZ}(g^{*}_{d}(x_{m},u_{0})|d,x_{m},z')}{f_{Y|DXZ}(g^{*}_{d}(x_{m},u_{0})|d,x_{m},z')}+O_{p}(a_{n}),
\end{align*}\end{linenomath*}

Similarly,
\begin{linenomath*}\begin{align*}
\hat{F}_{Y|DXZ}(\hat{\varphi}_{d}(\hat{g}_{d}(u_{0}))|d,\hat{x}_{m},z)=\hat{F}_{Y|DXZ}(\hat{\varphi}_{d}(\hat{g}_{d}(u_{0}))|d,x_{m},z)+O_{p}(a_{n})
\end{align*}\end{linenomath*}

Therefore, we can expand $\hat{F}_{Y|DXZ}(\hat{\varphi}_{d}(\hat{g}_{d}(u_{0}))|d,\hat{x}_{m},z)$ around\newline $\hat{F}_{Y|DXZ}(\varphi_{d}(g^{*}_{d}(x_{0},u_{0}))|d,x_{m},z)$ by
\begin{linenomath*}\begin{align*}
&\hat{F}_{Y|DXZ}(\hat{\varphi}_{d}(\hat{g}_{d}(u_{0}))|d,\hat{x}_{m},z)-\hat{F}_{Y|DXZ}(\varphi_{d}(g^{*}_{d}(x_{0},u_{0}))|d,x_{m},z)\\
=&\hat{F}_{Y|DXZ}(\hat{\varphi}_{d}(\hat{g}_{d}(u_{0}))|d,x_{m},z)-\hat{F}_{Y|DXZ}(\varphi_{d}(g^{*}_{d}(x_{0},u_{0}))|d,x_{m},z)+O_{p}(a_{n})\\
=&\hat{F}_{Y|DXZ}(\hat{\varphi}_{d}(\hat{g}_{d}(u_{0}))|d,x_{m},z)-\hat{F}_{Y|DXZ}(\hat{\varphi}_{d}(g^{*}_{d}(x_{0},u_{0}))|d,x_{m},z)\\
&+\hat{F}_{Y|DXZ}(\hat{\varphi}_{d}(g^{*}_{d}(x_{0},u_{0}))|d,x_{m},z)-\hat{F}_{Y|DXZ}(\varphi_{d}(g^{*}_{d}(x_{0},u_{0}))|d,x_{m},z)+O_{p}(a_{n})\\
=&\partial_{g}F_{Y|DXZ}(\varphi_{d}(g^{*}_{d}(x_{0},u_{0}))|d,x_{m},z)\cdot (\hat{g}_{d}(u_{0})-g^{*}_{d}(x_{0},u_{0}))\\
&+\frac{f_{Y|DXZ}(g^{*}_{d}(x_{m},u_{0})|d,x_{m},z)}{f_{Y|DXZ}(g^{*}_{d}(x_{m},u_{0})|d,x_{m},z')}\left(\hat{F}_{Y|DXZ}(g^{*}_{d}(x_{0},u_{0})|d,x_{0},z)-\hat{F}_{Y|DXZ}(g^{*}_{d}(x_{m},u_{0})|d,x_{m},z')\right)\\
&+O_{p}(a_{n})
\end{align*}\end{linenomath*}
where $\partial_{g}F_{Y|DXZ}(\varphi_{d}(g^{*}_{d}(x_{0},u_{0}))|d,x_{m},z)$ is the derivative of $F_{Y|DXZ}(\varphi_{d}(\cdot)|d,x_{m},z)$ evaluated at $g^{*}_{d}(x_{0},u_{0})$.

Take the first order partial derivative of the objective function $\frac{1}{J}\sum_{j=1}^{J}\hat{Q}_{NSP}(\bm{g}(u_{j}),u_{j})$ with respect to $\bm{g}(u_{0})$ and set it to zero. As $\sqrt{nh_{g}}a_{n}=o(1)$, under some manipulation, we obtain the following first order condition for $\hat{\boldsymbol{g}}(u_{0})$:
\begin{linenomath*}\begin{align*}
&\hat{\boldsymbol{g}}(u_{0})-\boldsymbol{g}^{*}(x_{0},u_{0})=-\big(\Pi_{NSP}'\boldsymbol{W}_{g}(u_{0})\Pi_{NSP}\big)^{-1}\cdot\Pi_{NSP}'\boldsymbol{W}_{g}(u_{0})\cdot\left[A_{n}+o_{p}\left(\frac{1}{\sqrt{nh_{g}}}\right)\right]\normalsize
\end{align*}\end{linenomath*}
where 
\begin{linenomath*}\begin{align*}
A_{n}=&\left(\hat{\Psi}^{*}\left(\boldsymbol{g}^{*}(x_{0},u_{0})\right)-\Psi\left(\boldsymbol{g}^{*}(x_{0},u_{0})\right)\right)\\
&\footnotesize+\begin{pmatrix}
0\\
0\\
\sum_{d=1}^{3}p_{d}(x_{m1},0)\phi_{d1}\cdot\big(\hat{F}_{Y|DXZ}(g^{*}_{d}(x_{0},u_{0})|d,x_{0},0)-\hat{F}_{Y|DXZ}(g^{*}_{d}(x_{m1},u_{0})|d,x_{m1},1)\big)\\
\sum_{d=1}^{3}p_{d}(x_{m2},1)\phi_{d2}\cdot\big(\hat{F}_{Y|DXZ}(g^{*}_{d}(x_{0},u_{0})|d,x_{0},1)-\hat{F}_{Y|DXZ}(g^{*}_{d}(x_{m2},u_{0})|d,x_{m2},0)\big)
\end{pmatrix},
\end{align*}\end{linenomath*}
and $\Pi_{NSP}$, $\phi_{d1}$, $\phi_{d2}$ are as defined in Appendix \ref{secSA2}. In $\hat{\Psi}^{*}$, the generalized propensity scores, the matching points, and the mapping $\varphi_{d}$ are all equal to the true values. The only estimated parts are the conditional cumulative distribution functions in $\Psi$. The bias of $A_{n}$ is $O(h_{g}^{2})$ by the standard argument since $g^{*}_{d}(x_{0},u_{0})$ is in the interior of $S(Y|d,x_{0})$ for all $d$. To derive the asymptotic variance, let
\begin{linenomath*}\begin{equation*}
\mathbb{G}_{d}(y,x,z)\equiv \hat{F}_{YDXZ}(y,d,x,z)-F_{Y|DXZ}(y|d,x,z)\hat{f}_{DXZ}(d,x,z).
\end{equation*}\end{linenomath*}
Recall that by Lemma \ref{1.thmMEQ},
\begin{linenomath*}\begin{align*}
F_{Y|DXZ}(g^{*}_{d}(x_{0},u_{0})|d,x_{0},0)-F_{Y|DXZ}(g^{*}_{d}(x_{m1},u_{0})|d,x_{m1},1)&=0,\\
F_{Y|DXZ}(g^{*}_{d}(x_{0},u_{0})|d,x_{0},1)-F_{Y|DXZ}(g^{*}_{d}(x_{m2},u_{0})|d,x_{m2},0)&=0.
\end{align*}\end{linenomath*}
Therefore, $A_{n}$ is asymptotically equivalent as the following expression. 
\begin{linenomath*}\begin{equation*}\footnotesize
\begin{pmatrix}
\sum_{d=1}^{3}p_{d}(x_{0},0)\dfrac{\mathbb{G}_{d}(g_{d}^{*}(x_{0},u_{0}),x_{0},0)}{f_{DXZ}(d,x_{0},0)}\\[2em]
\sum_{d=1}^{3}p_{d}(x_{0},1)\dfrac{\mathbb{G}_{d}(g_{d}^{*}(x_{0},u_{0}),x_{0},1)}{f_{DXZ}(d,x_{0},1)}\\[2em]
\sum_{d=1}^{3}p_{d}(x_{m1},0)\left[\dfrac{\mathbb{G}_{d}(g_{d}^{*}(x_{m1},u_{0}),x_{m1},0)}{f_{DXZ}(d,x_{m1},0)}+\phi_{d1}\dfrac{\mathbb{G}_{d}(g_{d}^{*}(x_{0},u_{0}),x_{0},0)}{f_{DXZ}(d,x_{0},0)}-\phi_{d1}\dfrac{\mathbb{G}_{d}(g_{d}^{*}(x_{m1},u_{0}),x_{m1},1)}{f_{DXZ}(d,x_{m1},1)}\right]\\[2em]
\sum_{d=1}^{3}p_{d}(x_{m2},1)\left[\dfrac{\mathbb{G}_{d}(g_{d}^{*}(x_{m2},u_{0}),x_{m2},1)}{f_{DXZ}(d,x_{m2},1)}+\phi_{d2}\dfrac{\mathbb{G}_{d}(g_{d}^{*}(x_{0},u_{0}),x_{0},1)}{f_{DXZ}(d,x_{0},1)}-\phi_{d2}\dfrac{\mathbb{G}_{d}(g_{d}^{*}(x_{m2},u_{0}),x_{m2},0)}{f_{DXZ}(d,x_{m2},0)}\right]
\end{pmatrix}
\end{equation*}\end{linenomath*}\normalsize

The variance of each $\mathbb{G}_{d}$ follows Theorem 2.2 in \cite{li2008nonparametric}:
\begin{linenomath*}\begin{equation*}
\mathbb{V}(\mathbb{G}_{d}(y,x,z))=\frac{\kappa f_{DXZ}(d,x,z) F_{Y|DXZ}(y|d,x,z)\cdot (1-F_{Y|DXZ}(y|d,x,z))}{nh_{g}}+o\left(\frac{1}{nh_{g}}\right)
\end{equation*}\end{linenomath*}
 
By i.i.d., the covariance between $\mathbb{G}_{d}(y,x,z)$ and $\mathbb{G}_{d'}(y',x',z')$ for any different $(y,d,x,z)$ and $(y',d',x',z')$ is equal to
\begin{linenomath*}\begin{align*}
&\mathbb{C}\left(\mathbb{G}_{d}(y,x,z), \mathbb{G}_{d'}(y',x',z')\right)\\
=&\frac{1}{nh^{2}_{g}}\mathbb{E}\Big[\left(L\left(\frac{y-Y_{i}}{h_{0}}\right)-F_{Y|DXZ}(y|d,x,z)\right)\left(L\left(\frac{y'-Y_{i}}{h_{0}}\right)-F_{Y|DXZ}(y'|d',x',z')\right)\\
&\cdot K\left(\frac{X_{i}-x}{h_{g}}\right)K\left(\frac{X_{i}-x'}{h_{g}}\right)\cdot\mathbbm{1}(D_{i}=d)\mathbbm{1}(D_{i}=d')\mathbbm{1}(Z_{i}=z)\mathbbm{1}(Z_{i}=z')\Big]+o\left(\frac{1}{nh_{g}}\right)
\end{align*}\end{linenomath*}
Whenever $z'\neq z$ or $d'\neq d$, the expectation on the right hand side is equal to zero. When $z'=z$ and $d'=d$ but $x\neq x'$, for large enough $n$, $h_{g}$ is sufficiently small such that $|x'-x|>2h_{g}$. As $K(\cdot)=0$ outside $[-1,1]$, the product of the kernel functions is always zero. Hence, the expectation is zero whenever $(d,x,z)\neq (d',x',z')$. Therefore, $\mathbb{C}\left(\mathbb{G}_{d}(y,x,z),\mathbb{G}_{d'}(y',x',z')\right)=o(1/nh_{g})$ and is negligible.

As $h_{g}^{2}\sqrt{nh_{g}}\to 0$, by central limit theorem and the delta method, the asymptotic distribution in the theorem is obtained.
\end{proof}
\subsection{Proofs of Lemmas}\label{1.appE2}
\subsubsection{Proof of Lemma \ref{1.lemA1}}
The following lemma generalizes Lemma A.5 in \cite{li2008nonparametric} for a bounded $Y$.
\begin{appxlem}\label{1.lemA6} Under the conditions in Lemma \ref{1.lemA1}, the following equation holds for all $(d,z)\in S(D,Z)$:
\begin{linenomath*}\begin{align*}
\sup_{\substack{y\in [\underline{y}_{d},\bar{y}_{d}]\\x\in S_{0}(X)}}\Big|\mathbb{E}\left(L\left(\frac{y-Y_{i}}{h_{0}}\right)|D_{i}=d,X_{i}=x,Z_{i}=z\right)-F_{Y|DXZ}(y|d,x,z)\Big|=O(h_{0})\label{1.eqA20}
\end{align*}\end{linenomath*}
\end{appxlem}
\begin{proof}[Proof of Lemma \ref{1.lemA6}]
By definition,
\begin{linenomath*}\begin{align*}
\mathbb{E}\left(L\left(\frac{y-Y_{i}}{h_{0}}\right)|D_{i}=d,X_{i}=x,Z_{i}=z\right)=&\int_{\underline{y}_{dx}}^{\bar{y}_{dx}}L\left(\frac{y-y'}{h_{0}}\right)dF_{Y|DXZ}(y'|d,x,z)
\end{align*}\end{linenomath*}
Let $\nu=\frac{y-y'}{h_{0}}$. By change of variable and integration by part, the right hand side can be rewritten as
\begin{linenomath*}\begin{align*}
&-\int_{\frac{y-\bar{y}_{dx}}{h_{0}}}^{\frac{y-\underline{y}_{dx}}{h_{0}}}L(\nu)dF_{Y|DXZ}\left(y-\nu h_{0}|d,x,z\right)\\
=&L\left(\frac{y-\bar{y}_{dx}}{h_{0}}\right)+\int_{\frac{y-\bar{y}_{dx}}{h_{0}}}^{\frac{y-\underline{y}_{dx}}{h_{0}}}L'(\nu)F_{Y|DXZ}(y-\nu h_{0}|d.x,z)d\nu\\
=&L\left(\frac{y-\bar{y}_{dx}}{h_{0}}\right)+F_{Y|DXZ}(y|d,x,z)\left(L\left(\frac{y-\underline{y}_{dx}}{h_{0}}\right)-L\left(\frac{y-\bar{y}_{dx}}{h_{0}}\right)\right)\\
&+\frac{h_{0}^{2}}{2}\int_{\frac{y-\bar{y}_{dx}}{h_{0}}}^{\frac{y-\underline{y}_{dx}}{h_{0}}}\nu^{2} f'_{Y|DXZ}(c(\nu)|d,x,z) L'(\nu)d\nu
\end{align*}\end{linenomath*}
where $c(\nu)\in S(Y|d,x)$. Note that the second order Taylor expansion can be applied to $F_{Y|DXZ}(y-\nu h_{0}|d,x,z)$ around $F_{Y|DXZ}(y|d,x,z)$ even if $y$ may lie outside $S(Y|d,x)$ and $F_{Y|DXZ}(\cdot|d,x,z)$ is not differentiable everywhere between $y$ and $y-\nu h_{0}$. This is because, for instance if $y<\underline{y}_{dx}$, we have $F_{Y|DXZ}(y|d,x,z)=F_{Y|DXZ}(\underline{y}_{dx}|d,x,z)$, so it is equivalent to expand $F_{Y|DXZ}(y-\nu h_{0}|d,x,z)$ around $F_{Y|DXZ}(\underline{y}_{dx}|d,x,z)$. Under Assumption \ref{1.assREGNSP}, $F_{Y|DXZ}(\cdot|d,x,z)$ is three times differentiable on $[\underline{y}_{dx},y-\nu h_{0}]$ for any $\nu\in [(y-\bar{y}_{dx})/h_{0},(y-\underline{y}_{dx})/h_{0}]$. 

Rearranging the terms, we obtain
\begin{linenomath*}\begin{align*}
&\Big|\mathbb{E}\left(L\left(\frac{y-Y_{i}}{h_{0}}\right)|D_{i}=d,X_{i}=x,Z_{i}=z\right)-F_{Y|DXZ}(y|d,x,z)\Big|\\
\leq & \Big|L\left(\frac{y-\bar{y}_{dx}}{h_{0}}\right)\big(1-F_{Y|DXZ}(y|d,x,z)\big)-\left(1-L\left(\frac{y-\underline{y}_{dx}}{h_{0}}\right)\right)F_{Y|DXZ}(y|d,x,z)\Big|+O(h_{0}^{2})
\end{align*}\end{linenomath*}
Recall that the derivative of the function $L$ is the kernel $K$, which is supported on $[-1,1]$. If $\underline{y}_{dx}+h_{0}<y<\bar{y}_{dx}-h_{0}$, then $L(\frac{y-\bar{y}_{dx}}{h_{0}})=0$ and $L(\frac{y-\underline{y}_{dx}}{h_{0}})=1$. The absolute value on the right hand side of the last inequality equals zero. If $y\leq \underline{y}_{dx}+h_{0}$, $L(\frac{y-\bar{y}_{dx}}{h_{0}})=0$ for small enough $h_{0}$. Then the the absolute value is bounded from above by $F_{Y|DXZ}(\underline{y}_{dx}+h_{0}|d,x,z)$, which is $O(h_{0})$ uniformly in $x$ by the mean value theorem. Similarly, if $\bar{y}_{dx}-h_{0}\leq y$, then $L(\frac{y-\underline{y}_{dx}}{h_{0}})=1$ and the absolute value is no greater than $1-F_{Y|DXZ}(\bar{y}_{dx}-h_{0}|d,x,z)=O(h_{0})$ uniformly. 
\end{proof}
\begin{appxrem}\label{apprem1}
From the proof, the bias is $O(h_{0})$ because near the boundaries of $S(Y|d,x)$, the leading term in the bias (the absolute value) is $O(h_{0})$. If the supremum is taken over an interior subset of $S(Y|d,x)$ instead, that term equals zero. The bias is then reduced to $O(h_{0}^{2})$. 
\end{appxrem}
Now we are ready to prove Lemma \ref{1.lemA1}.
By construction, $\hat{F}_{Y|DXZ}(\cdot|d,x,z)\in [0,1]$ and is weakly increasing on $\mathbb{R}$. Also we have $\underline{y}_{d}\leq \underline{y}_{dx}<\bar{y}_{dx}\leq \bar{y}_{d}$. Therefore,
\begin{linenomath*}\begin{align*}
&\sup_{\substack{y\in S(Y|d)\\x\in S_{0}(X)}}\big|\hat{F}_{Y|DXZ}(y|d,x,z)-F_{Y|DXZ}(y|d,x,z)\big|\\
\leq &\sup_{x\in S_{0}(X)}\sup_{y\leq \underline{y}_{dx}}\big|\hat{F}_{Y|DXZ}(y|d,x,z)-F_{Y|DXZ}(y|d,x,z)\big|\\
&+\sup_{x\in S_{0}(X)}\sup_{y\geq \bar{y}_{dx}}\big|\hat{F}_{Y|DXZ}(y|d,x,z)-F_{Y|DXZ}(y|d,x,z)\big|\\
&+\sup_{x\in S_{0}(X)}\sup_{ \underline{y}_{dx}\leq y\leq \bar{y}_{dx}}\big|\hat{F}_{Y|DXZ}(y|d,x,z)-F_{Y|DXZ}(y|d,x,z)\big|\\
=& \sup_{x\in S_{0}(X)}\hat{F}_{Y|DXZ}(\underline{y}_{dx}|d,x,z)
+\sup_{x\in S_{0}(X)}\big(1-\hat{F}_{Y|DXZ}(\bar{y}_{dx}|d,x,z)\big)\\
&+\sup_{x\in S_{0}(X)}\sup_{ \underline{y}_{dx}\leq y\leq \bar{y}_{dx}}\big|\hat{F}_{Y|DXZ}(y|d,x,z)-F_{Y|DXZ}(y|d,x,z)\big|\\
\leq & 3\sup_{x\in S_{0}(X)}\sup_{ \underline{y}_{dx}\leq y\leq \bar{y}_{dx}}\big|\hat{F}_{Y|DXZ}(y|d,x,z)-F_{Y|DXZ}(y|d,x,z)\big|\\
=& 3\sup_{x\in S_{0}(X)}\sup_{ \underline{y}_{dx}\leq y\leq \bar{y}_{dx}}\big|\frac{\hat{F}_{YDXZ}(y,d,x,z)-F_{Y|DXZ}(y|d,x,z)\hat{f}_{DXZ}(d,x,z)}{\hat{f}_{DXZ}(d,x,z)}\big|
\end{align*}\end{linenomath*}
Since the denominator is free of $y$ and is uniformly consistent for $f_{DXZ}(d,x,z)$ on $S_{0}(X)$ under the imposed regularity conditions, its infimum is bounded away from $0$. 

For the numerator, denote it by $\hat{q}(y,x)$:
\begin{linenomath*}\begin{align*}
\sup_{\substack{\underline{y}_{dx}\leq y\leq \bar{y}_{dx}\\x\in S_{0}(X)}}\big|\hat{q}(y,x)\big|
\leq \sup_{\substack{\underline{y}_{dx}\leq y\leq \bar{y}_{dx}\\x\in S_{0}(X)}}\big|\mathbb{E}(\hat{q}(y,x))\big|+\sup_{\substack{\underline{y}_{dx}\leq y\leq \bar{y}_{dx}\\x\in S_{0}(X)}}\big|\hat{q}(y,x)-\mathbb{E}(\hat{q}(y,x))\big|
\end{align*}\end{linenomath*}
Noting that the function $L(\cdot)\cdot K(\cdot)$ is bounded and Lipschitz continuous, the second term on the right hand side is $O_{p}\left(\frac{\log(n)}{\sqrt{nh_{g}}}\right)$ by establishing the exponential type inequalities following the standard argument (e.g. \cite{li2007nonparametric}, pp. 36-40). For the first term, by the law of iterated expectation, 
\begin{linenomath*}\begin{align*}
&\mathbb{E}(\hat{q}(y,x))\\
=&\frac{1}{h_{g}}\sum_{d,z}\int \left[\mathbb{E}_{Y|DXZ}\left(L\left(\frac{y-Y_{i}}{h_{0}}\right)\big|d,x',z\right)-F_{Y|DXZ}(y|d,x,z)\right]K\left(\frac{x'-x}{h_{g}}\right)f_{DXZ}(d,x',z)dx'
\end{align*}\end{linenomath*}
By Lemma \ref{1.lemA6}, $\mathbb{E}_{Y|DXZ}\left(L\left(\frac{y-Y_{i}}{h_{0}}\right)|d,x',z\right)=F_{Y|DXZ}(y|d,x',z)+O(h_{0})$ uniformly. Therefore, $\big|\mathbb{E}(\hat{q}(y,x))\big|$ is uniformly bounded by
\begin{linenomath*}\begin{align*}
\Big|\frac{1}{h_{g}}\sum_{d,z}\int\big[F_{Y|DXZ}(y|d,x',z)-F_{Y|DXZ}(y|d,x,z)\big]K\big(\frac{x'-x}{h_{g}}\big)f_{DXZ}(d,x',z)dx'\Big|+O(h_{0})
\end{align*}\end{linenomath*}

Since the supremum is taken over $[\underline{y}_{dx},\bar{y}_{dx}]$, the function $F_{Y|DXZ}(y|d,\cdot,z)$ may not be differentiable everywhere between $x$ and $x'$. For instance, $y=\bar{y}_{dx}$ may not be in $S(Y|d,x')$. However, $F_{Y|DXZ}(y|d,\cdot,z)$ is Lipschitz continuous by Assumption \ref{1.assREGNSP}. Let $C$ be its Lipschitz constant. For some $C'>0$, we have
\begin{linenomath*}\begin{align*}
\sup_{\substack{\underline{y}_{dx}\leq y\leq \bar{y}_{dx}\\x\in S_{0}(X)}}\big|\mathbb{E}(\hat{q}(y,x))\big|\leq & C\sup_{x\in S_{0}(X)}\sum_{d,z}\int \big|\frac{x'-x}{h_{g}}\big|\cdot K\left(\frac{x'-x}{h_{g}}\right)f_{DXZ}(d,x',z)dx'+O(h_{0})\\
\leq & h_{g} C'\int_{0}^{1}|\nu| K(\nu)d\nu+O(h_{0})\\
=& O(h_{g})
\end{align*}\end{linenomath*}

\begin{appxrem}
Here the bias is of $O(h_{g})$ because of the nonsmoothness of $F_{Y|DXZ}(y|d,\cdot,z)$ at the boundaries. When the supremum is taken over some neighborhood of $y$ in the interior of $S(Y|d,x)$, the second order Taylor expansion of $F_{Y|DXZ}(y|d,x',z)$ around $x$ is feasible because for small enough $h_{g}$, $F_{Y|DXZ}(y|d,\cdot,z)$ is three times differentiable on $[x-h_{g},x+h_{g}]$ under Assumption \ref{1.assREGNSP}. Meanwhile, in that case, the bias by smoothing the cumulative distribution function estimator is also reduced to $O(h_{0}^{2})$ (see Remark \ref{apprem1}). Then the order of the bias would be $O(h_{g}^{2})$.
\end{appxrem}
\subsubsection{Proof of Lemma \ref{1.lemA2}}
By the triangle inequality, 
\begin{linenomath*}\begin{align*}
&\sup_{y\in [\underline{y}_{d},\bar{y}_{d}]}\big|\hat{F}_{Y|DXZ}(\hat{\varphi}_{d}(y;\hat{x}_{m})|d,\hat{x}_{m},z)-F_{Y|DXZ}(\varphi_{d}(y;x_{m})|d,x_{m},z)\big|\\
\leq &\sup_{y\in [\underline{y}_{d},\bar{y}_{d}]}\big|\hat{F}_{Y|DXZ}(\hat{\varphi}_{d}(y;\hat{x}_{m})|d,\hat{x}_{m},z)-F_{Y|DXZ}(\hat{\varphi}_{d}(y;\hat{x}_{m})|d,\hat{x}_{m},z)\big|\\
&+\sup_{y\in [\underline{y}_{d},\bar{y}_{d}]}\big|F_{Y|DXZ}(\hat{\varphi}_{d}(y;\hat{x}_{m})|d,\hat{x}_{m},z)-F_{Y|DXZ}(\varphi_{d}(y;x_{m})|d,x_{m},z)\big|\\
\leq &\sup_{\substack{y\in [\underline{y}_{d},\bar{y}_{d}]\\x\in S_{0}(X)}}\big|\hat{F}_{Y|DXZ}(y|d,x,z)-F_{Y|DXZ}(y|d,x,z)\big|\\
&+\sup_{y\in [\underline{y}_{d},\bar{y}_{d}]}\big|F_{Y|DXZ}(\hat{\varphi}_{d}(y;\hat{x}_{m})|d,x_{m},z)-F_{Y|DXZ}(\varphi_{d}(y;x_{m})|d,x_{m},z)\big|+O_{p}(a_{n})
\end{align*}\end{linenomath*}
where the last inequality holds because of Lipschitz continuity of $F_{Y|DXZ}(\cdot|d,\cdot,z)$. Note the first term on the right hand side is $O_{p}(r_{n})$ by Lemma \ref{1.lemA1}. Hence, we only need to show that
\begin{linenomath*}\begin{equation*}
\sup_{y\in [\underline{y}_{d},\bar{y}_{d}]}\big|F_{Y|DXZ}(\hat{\varphi}_{d}(y;\hat{x}_{m})|d,x_{m},z)-F_{Y|DXZ}(\varphi_{d}(y;x_{m})|d,x_{m},z)\big|=O_{p}(r_{n}+a_{n})
\end{equation*}\end{linenomath*}
This is straightforward if $\sup_{y\in [\underline{y}_{d},\bar{y}_{d}]}|\hat{\varphi}_{d}(y;\hat{x}_{m})-\varphi_{d}(y;x_{m})|=O_{p}(r_{n}+a_{n})$. However, since the minimizer of the population problem of equation \eqref{1.eq29} is not unique when $y\not\in [\underline{y}_{dx_{0}},\bar{y}_{dx_{0}}]$, the uniform rate of $\hat{\varphi}_{d}(y;\hat{x}_{m})$ is hard to establish. Instead, I first show that the equation holds for $z'$:
\begin{linenomath*}\begin{equation*}
\sup_{y\in [\underline{y}_{d},\bar{y}_{d}]}\big|F_{Y|DXZ}(\hat{\varphi}_{d}(y;\hat{x}_{m})|d,x_{m},z')-F_{Y|DXZ}(\varphi_{d}(y;x_{m})|d,x_{m},z')\big|=O_{p}(r_{n}+a_{n}).
\end{equation*}\end{linenomath*}
To show the equation holds, the left hand side satisfies
\begin{linenomath*}\begin{align*}
&\sup_{y\in[\underline{y}_{d},\bar{y}_{d}]}\big|F_{Y|DXZ}(\hat{\varphi}_{d}(y;\hat{x}_{m})|d,x_{m},z')-F_{Y|DXZ}(\varphi_{d}(y;x_{m})|d,x_{m},z')\big|\\
\leq &\sup_{y\in[\underline{y}_{d},\bar{y}_{d}]}\big|\hat{F}_{Y|DXZ}(\hat{\varphi}_{d}(y;\hat{x}_{m})|d,\hat{x}_{m},z')-F_{Y|DXZ}(\varphi_{d}(y;x_{m})|d,x_{m},z')\big|+O_{p}(r_{n}+a_{n})
\end{align*}\end{linenomath*}
by Lipschitz continuity of $F_{Y|DXZ}$ and the uniform rate of convergence of $\hat{F}_{Y|DXZ}$. To bound the first term, recall that by Theorem \ref{1.thmMEQ}, $F_{Y|DXZ}(\varphi_{d}(y;x_{m})|d,x_{m},z')=F_{Y|DXZ}(y|d,x_{0},z)$. Hence, 
\begin{linenomath*}\begin{align*}
\sup_{y\in[\underline{y}_{d},\bar{y}_{d}]}\big|\hat{F}_{Y|DXZ}(y|d,x_{0},z)-F_{Y|DXZ}(\varphi_{d}(y;x_{m})|d,x_{m},z')\big|=O_{p}(r_{n}).
\end{align*}\end{linenomath*}
Therefore,
\begin{linenomath*}\begin{align*}
&\sup_{y\in[\underline{y}_{d},\bar{y}_{d}]}\big|\hat{F}_{Y|DXZ}(\hat{\varphi}_{d}(y;\hat{x}_{m})|d,\hat{x}_{m},z')-F_{Y|DXZ}(\varphi_{d}(y;x_{m})|d,x_{m},z')\big|\\
\leq& \sup_{y\in[\underline{y}_{d},\bar{y}_{d}]}\big|\hat{F}_{Y|DXZ}(\hat{\varphi}_{d}(y;\hat{x}_{m})|d,\hat{x}_{m},z')-\hat{F}_{Y|DXZ}(y|d,x_{0},z)\big|+O_{p}(r_{n})\\
\leq& \sup_{y\in[\underline{y}_{d},\bar{y}_{d}]}\big|\hat{F}_{Y|DXZ}(\varphi_{d}(y;x_{m})|d,\hat{x}_{m},z')-\hat{F}_{Y|DXZ}(y|d,x_{0},z)\big|+O_{p}(r_{n})\\
\leq &\sup_{y\in[\underline{y}_{d},\bar{y}_{d}]}\big|\hat{F}_{Y|DXZ}(\varphi_{d}(y;x_{m})|d,\hat{x}_{m},z')-F_{Y|DXZ}(\varphi_{d}(y;x_{m})|d,x_{m},z')\big|+O_{p}(r_{n})\\
\leq &\sup_{y\in[\underline{y}_{d},\bar{y}_{d}]}\big|\hat{F}_{Y|DXZ}(\varphi_{d}(y;x_{m})|d,\hat{x}_{m},z')-F_{Y|DXZ}(\varphi_{d}(y;x_{m})|d,\hat{x}_{m},z')\big|+O_{p}(r_{n}+a_{n})\\
=&O_{p}(r_{n}+a_{n})
\end{align*}\end{linenomath*}
where the second inequality follows from the definition of $\hat{\varphi}_{d}(y;\hat{x}_{m})$.

Finally, let us establish the relationship between $F_{Y|DXZ}(y'|d,x_{m},z)-F_{Y|DXZ}(y|d,x_{m},z)$ and $F_{Y|DXZ}(y'|d,x_{m},z')-F_{Y|DXZ}(y|d,x_{m},z')$ for any $y,y'\in[\underline{y}_{d},\bar{y}_{d}]$. Recall that by Assumption \ref{1.assENSP}, $S(Y|d,x_{m},z)=S(Y|d,x_{m},z')=S(Y|d,x_{m})$. Therefore, we have the following two equations:
\begin{linenomath*}\begin{align*}
&F_{Y|DXZ}(y'|d,x_{m},z)-F_{Y|DXZ}(y|d,x_{m},z)\\
=&f_{Y|DXZ}(\tilde{y}_{1}|d,x_{m},z)\cdot \left[\left(\underline{y}_{dx_{m}}\lor\left(y'\land \bar{y}_{dx_{m}}\right)\right)-\left(\underline{y}_{dx_{m}}\lor\left(y\land \bar{y}_{dx_{m}}\right)\right)\right],
\end{align*}\end{linenomath*}
and
\begin{linenomath*}\begin{align*}
&F_{Y|DXZ}(y'|d,x_{m},z')-F_{Y|DXZ}(y|d,x_{m},z')\\
=& f_{Y|DXZ}(\tilde{y}_{2}|d,x_{m},z')\cdot \left[\left(\underline{y}_{dx_{m}}\lor\left(y'\land \bar{y}_{dx_{m}}\right)\right)-\left(\underline{y}_{dx_{m}}\lor\left(y\land \bar{y}_{dx_{m}}\right)\right)\right],
\end{align*}\end{linenomath*}
where $\tilde{y}_{1}$ and $\tilde{y}_{2}$ are mean values lying in $[\underline{y}_{dx_{m}},\bar{y}_{dx_{m}}]$.
By Assumption \ref{1.assREGNSP}, $\frac{f_{Y|DXZ}(\tilde{y}_{1}|d,x_{m},z')}{f_{Y|DXZ}(\tilde{y}_{2}|d,x_{m},z)}$ is uniformly bounded on $S(Y|d,x_{m})\times S(Y|d,x_{m})$. Meanwhile, the terms in the brackets in the two equations are equal. Therefore, there exists a constant $C>0$ such that 
\begin{linenomath*}\begin{align*}
&\sup_{y\in [\underline{y}_{d},\bar{y}_{d}]}\big|F_{Y|DXZ}(\hat{\varphi}_{d}(y;\hat{x}_{m})|d,x_{m},z)-F_{Y|DXZ}(\varphi_{d}(y;x_{m})|d,x_{m},z)\big|\\
\leq &C\sup_{y\in [\underline{y}_{d},\bar{y}_{d}]}\big|F_{Y|DXZ}(\hat{\varphi}_{d}(y;\hat{x}_{m})|d,x_{m},z')-F_{Y|DXZ}(\varphi_{d}(y;x_{m})|d,x_{m},z')\big|\\
=& O_{p}(r_{n}+a_{n})
\end{align*}\end{linenomath*}
\subsubsection{Proof Lemma \ref{1.lemA3}}
By the triangle inequality, 
\begin{linenomath*}\begin{align*}
&\sup_{\boldsymbol{g}\in\hat{\mathcal{G}}}\big|\frac{1}{J}\sum_{j=1}^{J}\hat{Q}_{NSP}(\boldsymbol{g}(u_{j}),u_{j})-\int_{0}^{1}Q_{NSP}(\boldsymbol{g}(u),u)du\big|\\
\leq & \sup_{\boldsymbol{g}\in\hat{\mathcal{G}}}\big|\frac{1}{J}\sum_{j=1}^{J}\hat{Q}_{NSP}(\boldsymbol{g}(u_{j}),u_{j})-\frac{1}{J}\sum_{j=1}^{J}Q_{NSP}(\boldsymbol{g}(u_{j}),u_{j})\big|\\
&+\sup_{\boldsymbol{g}\in\hat{\mathcal{G}}}\big|\frac{1}{J}\sum_{j=1}^{J}Q_{NSP}(\boldsymbol{g}(u_{j}),u_{j})-\int_{0}^{1}Q_{NSP}(\boldsymbol{g}(u),u)du\big|\\
\leq & \sup_{\boldsymbol{y}\in\prod_{d}S(Y|d)}\big|\hat{Q}_{NSP}(\boldsymbol{y},u_{j})-Q_{NSP}(\boldsymbol{y},u_{j})\big|\\
&+\sup_{\boldsymbol{g}\in\hat{\mathcal{G}}}\big|\frac{1}{J}\sum_{j=1}^{J}Q_{NSP}(\boldsymbol{g}(u_{j}),u_{j})-\int_{0}^{1}Q_{NSP}(\boldsymbol{g}(u),u)du\big|\\
=&\sup_{\boldsymbol{g}\in\hat{\mathcal{G}}}\big|\frac{1}{J}\sum_{j=1}^{J}Q_{NSP}(\boldsymbol{g}(u_{j}),u_{j})-\int_{0}^{1}Q_{NSP}(\boldsymbol{g}(u),u)du\big|+o_{p}(1)
\end{align*}\end{linenomath*}

Let $u_{0}=0$. For some $C>0$, the first term satisfies
\begin{linenomath*}\begin{align*}
&\sup_{\boldsymbol{g}\in\hat{\mathcal{G}}}\big|\frac{1}{J}\sum_{j=1}^{J}Q_{NSP}(\boldsymbol{g}(u_{j}),u_{j})-\int_{0}^{1}Q_{NSP}(\boldsymbol{g}(u),u)du\big|\\
= & \sup_{\boldsymbol{g}\in\hat{\mathcal{G}}}\big|\frac{1}{J}\sum_{j=1}^{J}Q_{NSP}(\boldsymbol{g}(u_{j}),u_{j})-\sum_{j=1}^{J}\int_{\frac{j-1}{J}}^{\frac{j}{J}}Q_{NSP}(\boldsymbol{g}(u),u)du\big|\\
\leq & \sup_{\boldsymbol{g}\in\mathcal{G}^{*}}\big|\sum_{j=1}^{J}\big(\frac{1}{J}Q_{NSP}(\boldsymbol{g}(u_{j}),u_{j})-\int_{\frac{j-1}{J}}^{\frac{j}{J}}Q_{NSP}(\boldsymbol{g}(u),u)du\big)\big|\\
=& \sup_{\boldsymbol{g}\in\mathcal{G}^{*}}\big|\sum_{j=1}^{J}\big(\int_{\frac{j-1}{J}}^{\frac{j}{J}}Q_{NSP}(\boldsymbol{g}(u_{j}),u_{j})du-\int_{\frac{j-1}{J}}^{\frac{j}{J}}Q_{NSP}(\boldsymbol{g}(u),u)du\big)\big|\\
\leq& \sup_{\boldsymbol{g}\in\mathcal{G}^{*}}\sum_{j=j}^{J}\big(\int_{\frac{j-1}{J}}^{\frac{j}{J}}C\cdot \big[||\boldsymbol{g}(u_{j})-\boldsymbol{g}(u)||+|u_{j}-u|\big]du\big)
\end{align*}\end{linenomath*}
In the first inequality, $\hat{\mathcal{G}}$ is changed to $\mathcal{G}^{*}$. This is feasible by the flatness of the cumulative distribution functions outside the support: For any increasing $\boldsymbol{g}(u)$ that takes values outside $\prod_{d}S(Y|d,x_{0})$, there exists an increasing function taking values only in $\prod_{d}S(Y|d,x_{0})$ that yields the same value of $Q_{NSP}$. See details in the proof of Theorem \ref{1.thmIDNSP}. On the other hand, $\hat{G}$ only contains piecewise affine increasing functions, while $\mathcal{G}^{*}$ contains any increasing functions. Hence the supremum over $\hat{\mathcal{G}}$ is no greater than that over $\mathcal{G}^{*}$. The equality following it is due to the fact that $\int_{\frac{j-1}{J}}^{\frac{j}{J}}Q_{NSP}(\boldsymbol{g}(u_{j}),u_{j})du=\frac{1}{J}Q_{NSP}(\boldsymbol{g}(u_{j}),u_{j})$. The last inequality follows differentiability of $Q_{NSP}(\cdot,\cdot)$ on $\prod_{d=1}^{3}S(Y|d,x_{0})\times [0,1]$, the mean value theorem, and the boundedness of all the partial derivatives of $Q_{NSP}$. 

Note that for $u\in [(j-1)/J,j/J]$, $u<u_{j}=j/J$. Since $\bm{g}\in\mathcal{G}^{*}$, $\bm{g}$ is increasing. Therefore, $\bm{g}(u_{j-1})\leq \bm{g}(u)\leq \bm{g}(u_{j})$ and
\begin{linenomath*}\begin{align*}
&\sup_{\boldsymbol{g}\in\mathcal{G}^{*}}\sum_{j=j}^{J}\big(\int_{\frac{j-1}{J}}^{\frac{j}{J}}C\cdot \big[||\boldsymbol{g}(u_{j})-\boldsymbol{g}(u)||+|u_{j}-u|\big]du\big)\\
\leq& \sup_{\boldsymbol{g}\in\mathcal{G}^{*}}\sum_{j=j}^{J}\big(\int_{\frac{j-1}{J}}^{\frac{j}{J}}C\cdot \big[||\boldsymbol{g}(u_{j})-\boldsymbol{g}(u_{j-1})||+|u_{j}-u_{j-1}|\big]du\big)\\
=& \frac{C}{J}\cdot\sup_{\boldsymbol{g}\in\mathcal{G}^{*}}\sum_{j=1}^{J}||\boldsymbol{g}(u_{j})-\boldsymbol{g}(u_{j-1})||+\frac{C}{J}\sum_{j=1}^{J}(u_{j}-u_{j-1})\\
\leq & \frac{C}{J}\sup_{\boldsymbol{g}\in\mathcal{G}^{*}}\sum_{j=1}^{J}\sum_{d=1}^{3}(g_{d}(u_{j})-g_{d}(u_{j-1}))+\frac{C}{J}\sum_{j=1}^{J}(u_{j}-u_{j-1})\\
\leq &\frac{C}{J}\sum_{d=1}^{3}(\bar{y}_{dx_{0}}-\underline{y}_{dx_{0}})+\frac{C}{J}\\
=&o(1)
\end{align*}\end{linenomath*}
\subsubsection{Proof of Lemma \ref{1.thmROCNSP}}
It is straightforward that
\begin{linenomath*}\begin{equation*}
\max_{u_{j}\in \mathcal{U}_{0}}||\hat{\boldsymbol{g}}(u_{j})-\boldsymbol{g}^{*}(x_{0},u_{j})||\leq \sqrt{\sum_{\{j:u_{j}\in\mathcal{U}_{0}\}}(\hat{\boldsymbol{g}}(u_{j})-\boldsymbol{g}^{*}(x_{0},u_{j}))'(\hat{\boldsymbol{g}}(u_{j})-\boldsymbol{g}^{*}(x_{0},u_{j}))}
\end{equation*}\end{linenomath*}
Now I derive an upper bound on the right hand side.

By uniform consistency (Theorem \ref{thmSUPCONS}), $\hat{\boldsymbol{g}}(\mathcal{U}_{0})\subseteq \prod_{d}S(Y|d,x_{0})$ w.p.a.1. Under this event, $\Psi(\cdot)$ is differentiable at $\hat{\boldsymbol{g}}(u_{j})$. Since $a_{n}=o(r_{n})$, the generalized propensity score estimators and $\hat{x}_{m}$ in $\hat{\Psi}$ do not affect the rate of convergence of $\hat{\Psi}$, which is $O_{p}(r_{n})$ by Lemma \ref{1.lemA2}. By $\Psi(\boldsymbol{g}^{*}(x_{0},u_{j}))=\bm{u}_{j}$ and the mean value theorem,
\begin{linenomath*}\begin{align*}
\tilde{\Pi}_{NSP}(u_{j})\cdot (\hat{\boldsymbol{g}}(u_{j})-\boldsymbol{g}^{*}(x_{0},u_{j}))=&\Psi(\hat{\boldsymbol{g}}(u_{j}))-\Psi(\boldsymbol{g}^{*}(x_{0},u_{j}))\\
=& \hat{\Psi}(\hat{\boldsymbol{g}}(u_{j}))-\bm{u}_{j}+O_{p}(r_{n})
\end{align*}\end{linenomath*}
where $\tilde{\Pi}_{NSP}(u_{j})$ is the Jacobian of $\Psi$ at the mean value. By uniform consistency of $\hat{\bm{g}}$ and the full rank condition in Theorem \ref{1.thmIDNSP}, $\tilde{\Pi}_{NSP}(u_{j})$ is invertible w.p.a.1. Then w.p.a.1,
\begin{linenomath*}\begin{equation*}
\hat{\boldsymbol{g}}(u_{j})-\boldsymbol{g}^{*}(x_{0},u_{j})=\big(\tilde{\Pi}_{NSP}'(u_{j})\tilde{\Pi}_{NSP}(u_{j})\big)^{-1}\tilde{\Pi}_{NSP}'(u_{j})\cdot\big(\hat{\Psi}(\hat{\boldsymbol{g}}(u_{j}))-\bm{u}_{j}\big)+O_{p}(r_{n})
\end{equation*}\end{linenomath*}

Let $\bm{W}_{gn}(u)$ be the weighting matrix used to obtain $\hat{\bm{g}}(u)$. By boundedness of all the conditional densities, there exists a constant $C>0$ such that
\begin{linenomath*}\begin{align*}
&\sum_{\{j:u_{j}\in\mathcal{U}_{0}\}}(\hat{\boldsymbol{g}}(u_{j})-\boldsymbol{g}^{*}(x_{0},u_{j}))'(\hat{\boldsymbol{g}}(u_{j})-\boldsymbol{g}^{*}(x_{0},u_{j}))\notag\\
\leq& C \sum_{j=1}^{J}
\big(\hat{\Psi}(\hat{\boldsymbol{g}}(u_{j}))-\bm{u}_{j}\big)'\bm{W}_{gn}(u_{j})\big(\hat{\Psi}(\hat{\boldsymbol{g}}(u_{j}))-\bm{u}_{j}\big)+O_{p}(Jr_{n}^{2})\\
\leq &C\sum_{j=1}^{J}\big(\hat{\Psi}(\boldsymbol{g}_{0}(u_{j}))-\boldsymbol{u}_{j})\big)'\bm{W}_{gn}(u_{j})\big(\hat{\Psi}(\boldsymbol{g}_{0}(u_{j}))-\boldsymbol{u}_{j}\big)+O_{p}(Jr_{n}^{2})\\
\leq &C\sum_{j=1}^{J}\big(\hat{\Psi}(\boldsymbol{g}_{0}(u_{j}))-\Psi(\boldsymbol{g}_{0}(u_{j}))\big)'\bm{W}_{gn}(u_{j})\big(\hat{\Psi}(\boldsymbol{g}_{0}(u_{j}))-\Psi(\boldsymbol{g}_{0}(u_{j}))\big)+O_{p}(Jr_{n}^{2})\\
=& O_{p}(Jr_{n}^{2})
\end{align*}\end{linenomath*}
where the function $\boldsymbol{g}_{0}\in \hat{\mathcal{G}}$ is the same as in the proof of Theorem \ref{thmSUPCONS}. The second inequality thus follows the definition of $\hat{\bm{g}}$. Since $\boldsymbol{g}_{0}$ satisfies $\boldsymbol{g}_{0}(u_{j})=\boldsymbol{g}^{*}(x_{0},u_{j})$ for all $j=1,...,J$, $\Psi(\boldsymbol{g}_{0}(u_{j}))=\bm{u}_{j}$. The third inequality follows. The last equality is by the uniform rate of convergence of $\hat{\Psi}$. This completes the proof.
\renewcommand{\thesection}{S.B}
\setcounter{table}{0}
\renewcommand{\thetable}{S.B.\Roman{table}}	
\section{Additional Simulation Results}\label{secSB}
This section presents the simulation results under different $x_{0}$ or parameter values in the model in Section \ref{1.sec6}. Table \ref{1.tab2} and Table \ref{1.tab3} present the results for $x_{0}=-0.3$ and $0.3$. The results are similar to Table \ref{1.tab1}. Table \ref{1.tab4} shows the results for a weaker instrument. The parameters $(\alpha,\beta)$ are now set to one fifth of the values in the benchmark case in Section \ref{1.sec6}, while the smallest eigenvalue of the matrix $\Pi_{SP}'\Pi_{SP}$ is $1/200$ of the benchmark case. The variances of the estimator blow up, but the biases are still very small. Table \ref{1.tab5} shows the results for different correlation coefficients between $U$ and $V$. The results are again very similar to the benchmark case $\rho=0.5$ in Table \ref{1.tab1}. The relevant results in Table \ref{1.tab1} are also recorded in Tables \ref{1.tab3} and \ref{1.tab4} for the readers' convenience.
\begin{table}[H]
\centering
\caption{$x_{0}=-0.3$. $\boldsymbol{m}^{*}(-0.3)=(1.05,2.1,2.45)$}\label{1.tab2}
\begin{tabular}{C{1.4cm} C{1cm} C{1.3cm} C{1.3cm} C{1.3cm} C{1.3cm} C{1.3cm} C{1.3cm} C{1.3cm}}
\hline
\hline
&$n$ & Average & $\text{Bias}^{2}$& Variance & MSE & 90\% & 95\% & 99\% \\
\hline
\multirow {3}{*}{$\hat{m}_{1}(-0.3)$} &1000 &1.02 &0.001&0.14&0.14&93\% & 97\%& 99.2\% \tabularnewline
                                                    &2000 &1.04 &$2\cdot 10^{-4}$ &0.07&0.07& 91\%& 96.4\% &99.2\%  \tabularnewline
                                                    &3000 &1.04 &$10^{-4}$&0.05&0.05& 89.8\%&94.4\% &99\% \tabularnewline
\hline 
\multirow {3}{*}{$\hat{m}_{2}(-0.3)$} & 1000 &2.09&$10^{-4}$&0.85&0.85&93.8\%& 96.8\%& 99\%\tabularnewline
                                                     &2000 &2.04 &0.003&0.34&0.35& 92.8\%& 96.8\%&99.6\%\tabularnewline
                                                     &3000 &2.04 & 0.004&0.25&0.26& 90\%&95.8\% &99.4\%\tabularnewline
\hline 
\multirow {3}{*}{$\hat{m}_{3}(-0.3)$} & 1000 &2.36&0.01&0.21&0.22&91.4\%& 96\%& 98.4\%\tabularnewline
                                                     &2000 &2.41 &0.002&0.10&0.10& 93.2\%& 96.6\%&98.6\%\tabularnewline
                                                     &3000 &2.43 & $3\cdot 10^{-4}$&0.07&0.07& 89.2\%&94.4\% &99.2\%\tabularnewline
\hline 

\multirow {3}{*}{$\mathcal{J}_{x}$} &1000&             &           &            &     &  88.6\%  &   94.6\%   &99\%\tabularnewline
                                                     &2000&             &           &            &            &  90.4\%  &   95.2\%   &99.2\%\tabularnewline
                                                     &3000&             &           &            &            &  91.2\%  &   97\%   &99.4\%\tabularnewline
\hline 
\multirow {3}{*}{$\mathcal{J}_{SP}$} &1000&             &           &            &  &  93.4\%  &   96.6\%   &99.2\%\tabularnewline
                                                     &2000&             &           &            &            &  92.2\%  &   96.4\%   &99.2\%\tabularnewline
                                                     &3000&             &           &            &            &  91\%  &   94.2\%   &99.2\%\tabularnewline
\hline 
\end{tabular}
\end{table}
\begin{table}[H]
\centering
\caption{$x_{0}=0.3$. $\boldsymbol{m}^{*}(0.3)=(1.95,3.9,4.55)$}\label{1.tab3}
\begin{tabular}{C{1.4cm} C{1cm} C{1.3cm} C{1.3cm} C{1.3cm} C{1.3cm} C{1.3cm} C{1.3cm} C{1.3cm}}
\hline
\hline
&$n$ & Average & $\text{Bias}^{2}$& Variance & MSE & 90\% & 95\% & 99\% \\
\hline
\multirow {3}{*}{$\hat{m}_{1}(0.3)$} &1000 &1.95 &$2\cdot 10^{-5}$&0.11&0.11&89.4\%& 93.4\% & 97.2\% \tabularnewline
                                                     &2000 &1.98 &$8\cdot 10^{-4}$ &0.05&0.05& 90.2\%& 95\%  &98.6\%  \tabularnewline
                                                     &3000 &1.95 &$2\cdot 10^{-6}$&0.04&0.04& 89\%&94.2\%&99\% \tabularnewline
\hline 
\multirow {3}{*}{$\hat{m}_{2}(0.3)$} & 1000 &3.73&0.03&0.87&0.90&90.2\%& 93.8\% & 97\% \tabularnewline
                                                     &2000 &3.72 &0.03&0.38&0.41& 92\%&95.2\%&99\% \tabularnewline
                                                     &3000 &3.81 & 0.01&0.27&0.28& 89.6\%&95.6\% &99.6\% \tabularnewline
\hline 
\multirow {3}{*}{$\hat{m}_{3}(0.3)$} & 1000 &4.55&$10^{-5}$&0.27&0.27&92.4\%& 95.4\%& 97.6\% \tabularnewline
                                                     &2000 &4.56 &$2\cdot 10^{-4}$&0.13&0.13& 93.8\%& 97.4\%&99.2\%\tabularnewline
                                                     &3000 &4.53& $3\cdot 10^{-5}$&0.08&0.08& 93\%&97.8\%&99.2\% \tabularnewline
\hline 
\multirow {3}{*}{$\mathcal{J}_{x}$} &1000&             &           &            &     &  88.6\%  &   94.6\%   &99\%\tabularnewline
                                                     &2000&             &           &            &            &  88.8\%  &   95.2\%   &99.2\%\tabularnewline
                                                     &3000&             &           &            &            &  92\%  &   95.6\%   &99\%\tabularnewline
\hline 
\multirow {3}{*}{$\mathcal{J}_{SP}$} &1000&             &           &            &  &  90.2\%  &   94.2\%   &97.4\%\tabularnewline
                                                     &2000&             &           &            &            &  91.4\%  &   96.8\%   &99\%\tabularnewline
                                                     &3000&             &           &            &            &  90.8\%  &   96.2\%   &99.2\%\tabularnewline
\hline 
\end{tabular}
\end{table}



\begin{table}[H]
\centering
\caption{Different Strengths of $(Z,X)$. $n=2000$}\label{1.tab4}
\begin{tabular}{C{1.4cm} C{1cm} C{0.8cm} C{1.1cm} C{1.2cm} C{1.2cm} C{1.1cm} C{1.1cm} C{1.1cm} C{1.1cm}}
\hline
\hline
$(\alpha,\beta)$& min. eig. & & Average & $\text{Bias}^{2}$& Variance & MSE & 90\% & 95\% & 99\% \\
\hline
\multirow {3}{*}{(0.8,0.4)}& \multirow {3}{*}{0.02} & $\hat{m}_{1}(0)$ &1.51 & $3\cdot 10^{-5}$  &0.06&0.06& 91.6\%& 96\% &99\%     \tabularnewline
                                                               &         & $\hat{m}_{2}(0)$ &2.88 &0.01&0.37&0.39& 89.6\%& 95\%   &99\%  \tabularnewline
                                                            &            & $\hat{m}_{3}(0)$ &3.49 &$2\cdot 10^{-4}$&0.12&0.12  & 92.2\%& 97.2\%  &98.8\%    \tabularnewline
\hline 
      \multirow {3}{*}{(0.16,0.08)}&\multirow {3}{*}{$10^{-4}$} & $\hat{m}_{1}(0)$ &1.48 & $3\cdot 10^{-4}$  &13.10&13.10& 93\%& 96.4\%&99.2\%  \tabularnewline
                                                                  &      & $\hat{m}_{2}(0)$ &3.11 &0.01&50.82&50.83& 93.2\%& 95.6\%&99.4\%  \tabularnewline
                                                                  &      & $\hat{m}_{3}(0)$ &3.44 &0.004&11.04&11.05  & 91.8\%& 95.2\%&99.2\%  \tabularnewline
\hline                                                               
\end{tabular}
\end{table}

\begin{table}[H]
\centering
\caption{Different Degree of Endogeneity. $n=2000$}\label{1.tab5}
\begin{tabular}{C{1.2cm} C{0.8cm} C{1.3cm} C{1.3cm} C{1.3cm} C{1.3cm} C{1.3cm} C{1.3cm} C{1.3cm}}
\hline
\hline
&$\rho$ & Average & $\text{Bias}^{2}$& Variance & MSE & 90\% & 95\% & 99\% \\
\hline
\multirow {3}{*}{$\hat{m}_{1}(0)$} & 0.3 &1.50 &$2\cdot 10^{-5}$&0.06&0.06&92.4\% & 96.8\%& 99.6\% \tabularnewline
                                                     &0.5 &1.51 & $3\cdot 10^{-5}$  &0.06&0.06& 91.6\%& 96\% &99\%   \tabularnewline
                                                    &0.7 &1.51 & $4\cdot 10^{-5}$ &0.05&0.05& 91.2\%&95.4\%&99.6\% \tabularnewline                                                                                                 
\hline 
\multirow {3}{*}{$\hat{m}_{2}(0)$} & 0.3 &2.88&0.01&0.38&0.39&91.8\%& 96.4\%& 99.4\%  \tabularnewline
                                                    &0.5 &2.88 &0.01&0.37&0.39& 89.6\%& 95\%  &99\%  \tabularnewline
                                                    &0.7 &2.88 & 0.01&0.35&0.37& 91\%&96.2\% &99\%  \tabularnewline
                                                                                               
\hline 
\multirow {3}{*}{$\hat{m}_{3}(0)$} &0.3 &3.49&$5\cdot 10^{-5}$&0.12&0.12&93\%& 96.8\% & 99.2\%  \tabularnewline
                                                     &0.5 &3.49 &$2\cdot 10^{-4}$&0.12&0.12  & 92.2\%& 97.2\%   &98.8\%  \tabularnewline
                                                     &0.7 &3.49 & $2\cdot 10^{-4}$&0.11&0.11& 92\%&96.8\% &99.4\% \tabularnewline                                                                                                
\hline 
\end{tabular}
\end{table}
\end{appendices}
\bibliography{/Users/Ecthelion/Documents/Econometrics/NonparametricIdentification/Paper/bib}
\makeatletter\@input{tupp.tex}\makeatother